%% file: main_merge.tex
\pdfoutput=1
\newif\ifdraft \draftfalse

\newif \ifsubmission \submissiontrue
\newif \ifshort \shorttrue

\submissionfalse
\draftfalse

\makeatletter \@input{tex.flags} \makeatother

\ifsubmission

\documentclass{sig-alternate-05-2015}
\usepackage{acmcopyright}
\else
\documentclass[11pt]{article}
\usepackage{enumitem}
\fi

\newcommand{\shortbreak}{{\ifshort \\ \fi}}

\newcommand{\longquad}{{\ifshort \;\; \else  \qquad \fi}}

\newcommand{\thelongref}[1]{\ifshort the extended version\else \cref{#1}\fi}
\newcommand{\preflong}[1]{\ifshort\unskip\else(\cref{#1})\fi}
\newcommand{\LABEL}[1]{\label{#1}}

\newcommand{\SUBSECTION}[1]{\subsection{#1}}

\usepackage{xspace}

\let\proof\relax
\let\endproof\relax
\usepackage{amsmath, amssymb, amsthm,amsfonts}

\ifsubmission
\usepackage{flushend}
\usepackage{times}
\else
\usepackage{fullpage}
\usepackage[nottoc,notlot,notlof]{tocbibind}
\fi

\usepackage{xcolor}
\definecolor{DarkGreen}{rgb}{0.1,0.5,0.1}
\definecolor{DarkRed}{rgb}{0.5,0.1,0.1}
\definecolor{DarkBlue}{rgb}{0.1,0.1,0.5}

\ifsubmission
\usepackage[numbers,sort&compress]{natbib}
\else
\usepackage{natbib}
\fi

\newcommand{\V}{\mathcal{V}}
\newcommand{\D}{\Pi}
\renewcommand{\H}{\ensuremath{\mathcal{H}}\xspace}
\newcommand{\G}{\ensuremath{\mathcal{G}}\xspace}
\newcommand{\F}{\ensuremath{\mathcal{F}}\xspace}
\newcommand{\X}{\ensuremath{\mathcal{X}}\xspace}

\newcommand{\p}{\ensuremath{p}\xspace}

\newcommand{\dem}{\textrm{Dem}}
\newcommand{\com}{\textbf{compress}\xspace}

\newcommand{\decom}{\textbf{decompress}\xspace}
\newcommand{\PAC}{\ensuremath{\textsc{PAC}}\xspace}

\usepackage{tikz}
\usepackage{tikz-qtree}

\newcommand{\ar}[1]{\ifdraft\textcolor{brown}{[Aaron: #1]}\else\ignorespaces\fi}
\newcommand{\jh}[1]{\ifdraft\textcolor{red}{[Justin: #1]}\else\ignorespaces\fi}

\newcommand{\jmcomment}[1]{\ifdraft\textcolor{green}{[Jamie: #1]}\else\ignorespaces\fi}
\newcommand{\rynote}[1]{\ifdraft\textcolor{blue}{[Ryan: #1]}\else\ignorespaces\fi}

\newcommand\Z{\mathbb{Z}}

\newcommand\R{\mathbb{R}}
\newcommand\I{\mathbf{1}}

\newcommand{\cB}{\mathcal{B}}
\newcommand{\cC}{\mathcal{C}}
\newcommand{\cD}{\mathcal{D}}

\newcommand{\cG}{\mathcal{G}}

\newcommand{\cI}{\mathcal{I}}

\newcommand{\cL}{\mathcal{L}}
\newcommand{\cM}{\mathcal{M}}

\newcommand{\cP}{\mathcal{P}}

\newcommand{\vx}{\vec{x}}

\newcommand{\dom}{\text{dom}}

\DeclareMathOperator*{\myargmax}{\arg\!\max}

\newcommand{\bp}{\mathbf{p}}

\newcommand{\1}{\mathbf{1}}

\newcommand{\argmax}{\mathrm{argmax}}

\renewcommand{\tilde}{\widetilde}

\newcommand{\sam}{n}
\newcommand{\ngood}{m}
\newcommand{\pd}{\ensuremath{\textsc{PD}}\xspace}
\newcommand{\VC}{\ensuremath{\textsc{VC}}\xspace}
\newcommand{\err}{\textrm{err}}

\newcommand{\A}{\mathcal{A}}
\newcommand{\E}{\mathbb{E}}
\newcommand{\ex}[1]{\mathbb{E}\left[#1\right]}
\DeclareMathOperator*{\Expectation}{\mathbb{E}}
\newcommand{\Ex}[2]{\Expectation_{#1}\left[#2\right]}
\DeclareMathOperator*{\Probability}{\mathbb{P}}
\newcommand{\prob}[1]{\mathbb{P}\left[#1\right]}
\newcommand{\Prob}[2]{\Probability_{#1}\left[#2\right]}

\newcommand{\eps}{\varepsilon}

\renewcommand{\hat}{\widehat}

\newcommand{\util}[2]{u_{#2} (#1 ; p)}
\newcommand{\val}[2]{v_{#2} (#1)}
\newcommand{\weight}[3]{w_{#2}^{(#3)} (#1)}

\renewcommand{\dem}[1]{\cD_{#1}}
\newcommand{\mindem}[1]{\cD^*_{#1}}
\newcommand{\maxdem}[1]{\cD^\bullet_{#1}}
\newcommand{\canon}[1]{B^*_{#1}}

\newcommand{\ndem}[4]{\ensuremath{C^{#4}(#1;#2;#3)}}
\newcommand{\welfare}[2]{\textrm{Welfare}_{#1}(#2)}
\newcommand{\optw}[1]{\textrm{Opt-Welfare}_{#1}}
\newcommand{\rwel}[2]{\textrm{RWelfare}_{#2}(#1)}
\newcommand{\wel}[2]{\textrm{Welfare}_{#2}(#1)}
\newcommand{\maxod}[1]{OD^\bullet(#1;\bp)}

\newcommand{\maxval}{\ensuremath{H}\xspace}

\newcommand{\typebun}{\cG}

\renewcommand{\bp}{p}

\usepackage[]{hyperref}
\hypersetup{
    unicode=false,          
    pdftoolbar=true,        
    pdfmenubar=true,        
    pdffitwindow=false,      
    pdftitle={},    
    pdfauthor={}
    pdfsubject={},   
    pdfnewwindow=true,      
    pdfkeywords={keywords}, 
    colorlinks={\ifsubmission false\else true\fi},       
    linkcolor=DarkRed,          
    citecolor=DarkGreen,        
    filecolor=DarkRed,      
    urlcolor=DarkBlue,          
}
\usepackage[capitalize]{cleveref}

\newtheorem{theorem}{Theorem}[section]
\newtheorem{lemma}[theorem]{Lemma}

\newtheorem{corollary}[theorem]{Corollary}
\newtheorem{conjecture}[theorem]{Conjecture}
\newtheorem{definition}[theorem]{Definition}
\newtheorem{remark}[theorem]{Remark}
\newtheorem{claim}[theorem]{Claim}

\title{Do Prices Coordinate Markets?}

\ifsubmission
\numberofauthors{5}
\author{
\alignauthor
Justin Hsu\titlenote{%
  Supported in part by NSF grants CCF-1101389, TWC-1513694 and CNS-1065060 and
  by a grant from the Simons Foundation (\#360368).}\\
       \affaddr{Department of Computer and Information Science}\\
       \affaddr{University of Pennsylvania}\\
       \affaddr{USA}\\
        \email{justhsu@cis.upenn.edu}
\alignauthor
Jamie Morgenstern\titlenote{Supported in part by the Warren Center.}\\
        \affaddr{Departments of Computer and Information Science and Economics}\\
       \affaddr{University of Pennsylvania}\\
       \affaddr{USA}\\
       \email{jamiemmt@cs.cmu.edu}
\alignauthor
Ryan Rogers \titlenote{Supported in part by an NSF grant CNS-1253345}\\
       \affaddr{Department of Applied Mathematics and Computational Science}\\
       \affaddr{University of Pennsylvania}\\
       \affaddr{USA}\\
       \email{ryrogers@sas.upenn.edu}
\and
\alignauthor
Aaron Roth\titlenote{%
  Supported in part by an NSF grant CCF-1101389
  and an Alfred P. Sloan Research Fellowship.}\\
       \affaddr{Department of Computer and Information Sciences}\\
       \affaddr{University of Pennsylvania}\\
       \affaddr{USA}\\
       \email{aaroth@cis.upenn.edu}
\alignauthor
Rakesh Vohra\\
       \affaddr{Economics Department}\\
       \affaddr{University of Pennsylvania}\\
       \affaddr{USA}\\
       \email{rvohra@seas.upenn.edu}
}
\date{}
\else

\author{%
  Justin Hsu\thanks{%
    Department of Computer and Information Science, University of Pennsylvania.
    Email: \texttt{justhsu@cis.upenn.edu}
    Supported in part by NSF grants CCF-1101389, TWC-1513694 and CNS-1065060 and
    by a grant from the Simons Foundation (\#360368)}
  \and
  Jamie Morgenstern\thanks{%
    Warren Center Postdoctoral Fellow, Departments of Computer and Information Science and Economics,
    University of Pennsylvania.
    Email: \texttt{jamiemmt@cs.cmu.edu}}
  \and Ryan Rogers\thanks{%
    Department of Applied Mathematics and Computational Science, University of Pennsylvania.
    Email: \texttt{ryrogers@sas.upenn.edu}}
  \and Aaron Roth\thanks{%
    Department of Computer and Information Sciences, University of Pennsylvania.
    Email: \texttt{aaroth@cis.upenn.edu}
    Supported in part by an NSF grant CCF-1101389 and an Alfred P. Sloan Research Fellowship.}
  \and Rakesh Vohra\thanks{%
    Economics Department, University of Pennsylvania.
    Email: \texttt{rvohra@seas.upenn.edu}}}
\fi

\begin{document}

\sloppy

\ifsubmission
\CopyrightYear{2016} 
\setcopyright{acmlicensed}
\conferenceinfo{STOC '16,}{June 18--21, 2016, Cambridge, MA, USA}
 \crdata{978-1-4503-4132-5/16/06}
\doi{http://dx.doi.org/10.1145/2897518.2897559}
\else
\fi

\ifsubmission
\else
\begin{titlepage}
\fi

\maketitle

\input{abstract}
\ifsubmission



\else
\end{titlepage}
\fi

\ifsubmission
  \shorttrue
  \newcommand{\mathmath}{$ }
\renewcommand{\frac}[2]{\tfrac{#1}{#2}}
  \setcounter{page}{1}
\else
\newpage
 \tableofcontents
\newpage
  \shortfalse
  \newcommand{\mathmath}{$$ }
\fi

\input{paper}

\ifsubmission

\else

\paragraph*{Acknowledgments}
We warmly thank Renato Paes Leme, Kazuo Murota, Akiyoshi Shioura, and the
participants at the Workshop on Complexity and Simplicity in Economics at the
Simons Institute for the Theory of Computing for insightful discussions on
earlier versions of this work. We also thank the anonymous reviewers for their
helpful comments.

\newpage
\fi

\bibliographystyle{plainnat}
\bibliography{./header,./refs}

\ifsubmission
\else
  \appendix
  \input{appendix-multiple}
  \input{appendix-genericity}

\input{appendix-pricedim}

\fi
\end{document}

%% file: abstract.tex
\begin{abstract}
  Walrasian equilibrium prices have a remarkable property: they allow
  each buyer to purchase a bundle of goods that she finds the most
  desirable, while guaranteeing that the induced allocation over all
  buyers will globally maximize social welfare. However, this clean
  story has two caveats: first, the prices may induce indifferences. In fact, the
  minimal equilibrium prices \emph{necessarily} induce
  indifferences. Accordingly, buyers may need to coordinate with one
  another to arrive at a socially optimal outcome---the prices alone are
  \emph{not} sufficient to coordinate the market; second, although natural procedures converge to
  Walrasian equilibrium prices on a fixed population, in practice
  buyers typically observe prices without participating in a
  price computation process. These prices cannot be perfect Walrasian
  equilibrium prices, but instead somehow reflect distributional
  information about the market. 

To better understand the performance of Walrasian prices
when facing these two problems, we give two results. First, we propose a
mild genericity condition on valuations under which the minimal
Walrasian equilibrium prices induce allocations which result in low
over-demand, no matter how the buyers break ties.
In fact, under genericity the over-demand of any good can be
bounded by $1$, which is the best possible at the minimal prices. We demonstrate our results
for unit demand valuations and give an extension to
matroid based valuations (MBV), conjectured to be equivalent to
gross substitute valuations (GS).

Second, we use techniques from learning theory to argue that the
over-demand and welfare induced by a price vector converge to their
expectations uniformly over the class of all price vectors, with respective sample
complexity linear and quadratic in the number of goods in the market. These
results make no assumption on the form of the valuation
functions.

These two results imply that under a mild genericity
condition, the exact Walrasian equilibrium prices computed in a market
are guaranteed to induce both low over-demand and high welfare when
used in a new market where agents are sampled independently from
the same distribution, whenever the number of agents is larger than
the number of commodities in the market.
\end{abstract}

%% file: paper.tex
\input{intro}

\input{related}

\input{model}

\input{lb}

\input{matchings}

\input{gs}

\input{price-dim}

%% file: intro.tex
\section{Introduction}
The power of Walrasian equilibrium is often summarized with a pithy
slogan: prices coordinate markets. However, this is not exactly
true---a Walrasian equilibrium specifies a price for each good
\emph{and} an assignment of goods to buyers.  The assignment is just
as important as the prices, since there can be multiple bundles of
goods which maximize a buyer's utility at given prices.  If buyers
select arbitrarily among these bundles, they may over-demand some
goods.

One way to avoid this problem is to assume strictly concave valuation functions
ensuring that each buyer always has a unique, utility maximizing bundle of
goods. However, this idea does not apply in many natural economic settings; for
instance, when goods are indivisible.
For another possible solution, perhaps we could use Walrasian
prices that eliminate the coordination problem. For example, with a
single good and distinct valuations, a price \emph{strictly} between
the largest and second largest buyer valuations will eliminate
indifferences. However, how do we arrive at such prices? Relying on a
coordinator with full knowledge of the market to solve the coordination
problem defeats the purpose of markets.

In this work, we consider a different approach: natural assumptions on
buyers to simplify the coordination problem.
We focus our investigation on \emph{minimal} Walrasian equilibrium prices,
which result from many natural market
dynamics\footnote{Minimal Walrasian prices are focal in other ways:
  e.g., in matching markets they correspond to VCG prices.}
\citep{KC82}. We aim to
allocate bundles of $m$ types of indivisible goods
$g$, each with some supply $s_g$, to $n$ buyers who have matroid based
Valuations\footnote{%
  ``Matroid based valuations'' are a structured subclass of
  \emph{gross substitutes} valuations. In fact, \citet{OL15}
  conjecture that the class of matroid based valuations is
  \emph{equal} to the class of gross substitutes valuations, the
  largest class of valuations for which Walrasian equilibrium prices
  are guaranteed to exist.} and quasi-linear preferences; the
assignment (``unit demand'') model is an important special case. We
detail our model in \cref{sec:model}.

We begin our technical results in \cref{sec:lb}, by showing that indifferences
at the minimal Walrasian prices can be a serious problem---goods can be in the
demand correspondence of \emph{every} buyer, possibly leading to all buyers
demanding the same good. Clever tie-breaking does not help;
for any tie-breaking rule, the induced over-demand can be as large as
$\Omega(n)$. (We provide a simple example in \cref{sec:lb}).
Even worse, we observe that the minimal Walrasian equilibrium prices
always induce indifferences.

Hence, we cannot hope to rule out all over-demand at the minimal prices. But, we give a
``genericity'' condition on buyer valuations that is the next-best thing: the
over-demand for each good $g$ will be at most $1$, independent of its
supply $s_g$ and the tie-breaking strategy used by buyers.  Therefore,
as the supply grows, worst-case over-demand becomes negligible. We
warm up with the assignment model in \cref{sec:matchings}, where our
condition is simple to state: buyer valuations for goods should be
linearly independent over the coefficients $\{-1,0,1\}$. To follow, we generalize
our techniques to the matroid based valuations case in \cref{sec:gs}. The situation is
significantly more complicated, but the core genericity definition and proof
strategy are in the same spirit.

After we show that exact minimal Walrasian prices generically
induce low over-demand, a natural question is whether this property
holds when the same prices are used on ``similar buyers''; we turn to
this question in \cref{sec:generalize}.  More formally, imagine a
sample $N_1$ of $n$ buyers drawn from an unknown distribution $\D$
of buyer valuations. The goods are priced using the minimal Walrasian
equilibrium prices, computed from the valuations of buyers in
$N_1$. Now, keeping these prices fixed, we draw a fresh sample $N_2$
of $n$ buyers from $\D$, who each choose some bundle from their demand
correspondence at the given prices (breaking ties arbitrarily). Will
the over-demand and welfare on $N_2$ be close to the over-demand and
welfare on $N_1$?

If the supply $s_g$ of a good $g$ is small, the
difference in over-demand between $N_1$ and $N_2$ may be large when
compared to $s_g$.  However, we show that if
$s_g \geq \tilde{O}(m/\epsilon^2)$, the demand for any good $g$ on
sample $N_2$ will be within a $1\pm \epsilon$ factor of the supply
$s_g$ of good $g$.  Note that the supply requirement is independent of
the market size $n$. Similarly, if the optimal welfare for $N_1, N_2$
is large, the induced welfare of the two markets will
be within a $1\pm\epsilon$ of one another (and within a $1-\epsilon$
factor of the optimal welfare for $N_2$).

Furthermore, we are able to prove these bounds \emph{without} any assumption
on the structure of the valuation functions. This lack of structure makes it
difficult to argue directly
about notions of combinatorial dimension like VC dimension, and so we
take a different approach which may be of independent interest. Using a
recent compression argument of \citet{daniely2014multiclass}, we show
that assuming fixed but unknown prices, the class of functions
\emph{predicting} a buyer's demanded bundle at those prices is
learnable using $\tilde{O}(m/\epsilon^2)$ many samples.  Because this
is a multi-class learning problem, learning does not imply
uniform convergence. However, the binary
problem of predicting demand for a \emph{particular} good is a
$1$-dimensional projection of the bundle prediction problem, and hence
is also learnable with the same number of samples. By a
classical result of \citet{ehrenfeucht1989general}, learning and uniform
convergence have the same sample complexity in binary
prediction problems. So, we can bound the VC dimension, and thus the sample
complexity for uniform convergence for demand. Moreover, our bound is
tight---even for unit demand buyers, the VC dimension of the class of
demand predictors is $\Omega(m)$.

Welfare, unlike demand, corresponds to a real-valued prediction
problem, and so the sample complexity needed for uniform convergence
cannot be bounded by bounding the sample complexity of
learning. Instead, we directly bound the pseudo-dimension of the class
of welfare predictors by $\tilde{O}(m^2)$, again without making any
assumptions about the form of the valuation functions. We show that if
the optimal welfare is
$\tilde{\Omega}\left(m^4 \sqrt{n}/\epsilon^2\right)$, the
welfare induced by the Walrasian prices $p$ for $N_1$ when applied to
$N_2$ is within a $1-\epsilon$ factor of optimal.

%% file: related.tex
\paragraph*{Related work}

We follow a long line of work on understanding how markets behave
under limited coordination.  If buyers' valuations are strictly
concave and items are divisible,~\citet{arrow1954existence} show that
there exist item prices $p$ such that each agent has a unique
utility-maximizing bundle at $p$, and that when each agent selects her
unique utility-maximizing bundle the market clears.  With indivisible
items, anonymous equilibrium item pricings may not exist; even when
they do, finding such prices is
NP-hard~\citep{deng02complexity}).~\citet{mount1974informational}
consider the size of the message space \jh{Is size/dimensionality
equivalent to bits of communication, or is there more structure?}
needed to compute Walrasian equilibria,
and~\citet{nisan06communication} show that polynomial communication is
sufficient.  This stands in sharp contrast to the situation for
submodular buyers, where exponential communication is needed to
compute prices which support an efficient
allocation~\citep{nisan06communication}.\ifshort\else\footnote{
the submodular setting, prices may be over arbitrary bundles and
individualized to each buyer, in contrast to the gross substitutes
setting, where efficient allocations can be supported by anonymous
item pricings.}\fi

Our work is also related to the growing area of learning for
mechanism design, where a mechanism is selected from some class of
mechanisms as a function of sampled buyers.  Recent work measures
the sample complexity of revenue maximization in the
single-parameter~\citep{B+03,balcan2005mechanism,blum2005near,
G+06,elkind2007,samplingbalcan07,
balcan2008reducing,hartline2009simple,
BBDSics11,CGMsoda13,dhangwatnotai2014revenue,
CR14,medina2014learning,huang2014making,RS15,morgenstern15pseudo} and
multi-parameter settings \citep{devanur2011prior,dughmi2014sampling}.

In particular, sample complexity results for revenue and welfare
maximization using item pricings are known.  Most papers consider
buyers that make decisions sequentially, avoiding over-demand from
uncoordinated resolution of indifferences.  In the unlimited supply
setting (in which over-demand cannot
arise),~\citet{balcan2008reducing} show how to learn approximately
revenue-optimal prices with polynomial sample complexity using a
covering argument. Extending this work, \citet{samplingbalcan07} show
how to handle large but limited supply settings and to handle welfare
maximization. These papers imply that demand concentration will have
convergence rates of
$\tilde{O}(\frac{k^2}{\epsilon}\ln\frac{1}{\delta})$; our result
improves this by a factor of $k$ (where $k$ is the number of distinct
types of goods).  This percolates into the guarantees regarding
welfare concentration, in that our concentration result requires a
weaker lower bound on the supply of each good. Our precise convergence
rate for welfare concentration is looser than theirs in several
parameters, in part because we allow buyers to break indifferences in
a worst-case way. Using similar techniques to
\citet{samplingbalcan07}, we could recover similar convergence rates
for welfare.  Similarly, for the more restricted setting of budgeted
buyers in the online adwords setting,~\citet{devanuradwords09} show a
similar generalization guarantee for prices. This is generalized to
aribtrary valuations by~\citet{agarwaldynamic09}, whose convergence
rate depends logarithmically on $n$ and linearly on $k$;
~\citet{molinaro2013geometry} removes this dependence on $n$ but has
quadratic dependence on $k$. Our convergence rate is independent of
$n$ and depends only linearly on $k$, and are therefore faster than
any of these previous results.\footnote{We thank Nikhil Devanur for
pointing out this related line of work.}

We achieve these tighter convergence rates using compression schemes
to derive several uniform convergence results; this compression tool
was first used in the context of game theory
by~\citet{balcan2014learning}, who use it directly to upper bound the
PAC complexity of a learning problem rather than to imply uniform
convergence over a class.

Our definition of genericity for unit demand valuations has also been called
\emph{indepedence} by \citet{BCK14,ABH09}.  In fact, our bound of over-demand for
generic unit demand valuations can be obtained from results due to \citet{BCK14},
although over-demand was not the focus of their work.  However, our analysis
extends to more general valuation functions.

Our genericity results rely on a novel \emph{swap graph} construction, which
bears a slight resemblance to the exchangability graph
of~\citet{murota1996valuated}.  The exchangability graph also has nodes defined
by goods, but without considering a Walrasian allocation and pricing. In
contrast, our swap graph is models indifferences at equilibrium.

%% file: model.tex
\section{Model}
\label{sec:model}
We consider a market with $m$ indivisible goods, where good
$g$ has supply $s_g \geq 1$.  We will write the set of goods as
$[\ngood] = \{1,\dots, \ngood \}$ and denote the bundles of the
$\ngood$ goods as $\cG = 2^{[\ngood]}$.  The market will also have a
set $N$ of $\sam$ buyers, where each buyer demands at most one copy of
each good.\ifshort\else\footnote{%
  More formally, copies of a good beyond the first copy have marginal
  valuation zero.}\fi{} For simplicity, we consider valuations
defined over subsets of goods rather than arbitrary sets of
copies; see \thelongref{sec:copies-appendix}
for a formal treatment. For each buyer $q\in N$, let
$v_q : \cG \to [0, \maxval]$ denote $q$'s valuation function. We will assume
$v_q$ is monotone and $v_q(\emptyset) = 0$.
Our goal is to produce a feasible assignments of goods to buyers.

\begin{definition}[Allocation]
  An \emph{allocation} $\mu:N \to \cG$ assigns each
  buyer a bundle such that the whole assignment is feasible:
$$
\sum_{q \in N} \1\{ g \in \mu(q)\} \leq s_g \longquad \text{for all } g \in [\ngood].
$$
\end{definition}
As is typical, we consider quasi-linear utility functions
$u_q: \cG\times \R_{\geq 0}^\ngood \to \R_{\geq 0}$ defined by
$$
\util{S}{q} = v_q(S) - \sum_{g \in S} p_g  \longquad
\text{for all } q \in N \longquad \text{and} \longquad S \in \cG.
$$
We will consider prices assigned to goods---each \emph{copy} of a good has
the same price. We will write prices as vectors
$\bp = (p_g)_{g \in [\ngood]} \in \R^\ngood_{\geq 0}$ or as functions
over bundles $p: \cG \to \R_{\geq 0}$ such that
$p(S) = \sum_{g \in S} p_g$.
\begin{definition}[Demand correspondence]
  The \emph{demand correspondence} for buyer
  $q \in N$ at prices $\bp$ is $\dem{q}(\bp) = \myargmax_{S \in \cG}\{\util{S}{q} \}$.
  We call bundles $S \in \dem{q}(\bp)$ \emph{demand} bundles. Note
  that $\dem{q}(\bp)$ contains only bundles with non-negative utility,
  since $\util{\emptyset}{q} = 0$ for every $p$.
\end{definition}
We focus our investigation on \emph{Walrasian
  equilibria}, defined by  a pricing and an allocation.

\begin{definition}[Walrasian equilibrium]
 For valuations $\{ v_q \}_{q\in N}$, we say that a pair
  $(\bp,\mu)$ of prices $\bp$ and allocation $\mu = (\mu_q)_{q \in N}$
  is a \emph{Walrasian equilibrium} (WE) if both:
\begin{itemize}
\item $\mu_q \in \dem{q}(\bp)$ for all $q \in N$; and
\item $p_g = 0$ implies $\sum_{q \in N} \1\{ g \in \mu(q)\} < s_g$ .
\end{itemize}

We call the price vector $\bp$ a \emph{Walrasian equilibrium price
  vector}.  Note that there may be many distinct Walrasian equilibrium
price vectors: in fact, the set of all Walrasian prices forms a
lattice. The \emph{minimum} Walrasian equilibrium price vector $\bp$
is the Walrasian equilibrium price vector that is coordinate-wise
minimal amongst all Walrasian equilibrium price vectors.\ifshort\else\footnote{The
  fact that the price vectors form a lattice guarantees the existence
  the minimal price.}\fi

Likewise, we call the allocation $\mu$ a \emph{Walrasian allocation}. While
there may be multiple distinct Walrasian allocations, it is known that all such
allocations must maximize welfare.

\LABEL{def:WE}
\end{definition}

In general, Walrasian equilibrium prices $\bp$ are not sufficient to
coordinate a corresponding allocation $\mu$ because buyers might have
indifferences ($|\dem{q}(p)|>1$). If buyers choose their bundle
arbitrarily, the resulting allocation can violate
supply constraints.  To measure the amount of violation, we make the
following two natural definitions.
\begin{definition}[Demanders and over-demand]
The set $U(g;\bp)$ of \emph{demanders} for a good $g \in [\ngood]$ at price
$\bp$ is the set of buyers that have some demand set containing $g$:
$$
U(g;\bp) = \left\{q \in N:  \exists D \in \dem{q}(\bp) \longquad \text{ where } g
  \in D \right\}.
$$

Then, the \emph{over-demand} $OD(g; \bp)$ for $g$ at prices $\bp$ is
the number of demanders beyond the supply of a good:
\[
OD(g;\bp) = \max\left\{ \left| U(g;\bp)   \right| - s_g, 0\right\}.
\]
That is, the over-demand is the worst-case excess demand if bidders
break ties in their demand correspondence arbitrarily.
\end{definition}

To build intuition, we focus the first part of our paper on unit demand bidders,
where
$$
v_q(S) = \myargmax_{g \in S} \{\val{g}{q} \} \longquad \text{for all } S \in \cG,
 $$
 and non-empty bundles $S\in \dem{q}(\bp)$ are \emph{demand
goods}.\ifshort\else\footnote{Any demand bundle with $|S|>1$ must have some
  good $g\in S$ for which $v(S) = v(g)$, by the definition of unit
  demand valuations.}\fi

%% file: lb.tex
\section{Lower bound}
\LABEL{sec:lb}

To build intuition for why tie breaking at equilibrium can lead to
infeasibility, we give an example of a market with $n$ buyers in which
the over-demand of a good can be $\Omega(n)$ regardless of how buyers break
ties, so long as they cannot coordinate with one another
after seeing the market instance.\footnote{This lack of coordination is
  formalized by requiring that the tie breaking rule buyer $q$ uses to
  select among favorite bundles must be independent of the valuations
  of buyers $q' \neq q$.}

\begin{lemma}
  There exist unit demand valuations such that at the minimal
  Walrasian prices $\bp$, some good has over-demand $\Omega(\sam)$.
\LABEL{lem:bad1}
\end{lemma}
\begin{proof}
  Consider a market with $\sam$ unit demand buyers $N = [n]$ and
  $\ngood = \sam$ distinct goods. For a
  distinguished good $g \in [\ngood]$, every buyer $q$ has valuation
  $\val{q}{q} = \val{g}{q} = 1$ and $\val{g'}{q} = 0$ for
  all $g' \not\in \{q,g\}$. The minimal
  Walrasian equilibrium prices are $\bp = \mathbf{0}$, and the unique
  max-welfare allocation is $\mu_q = q$.
  At these prices, $g$ is demanded by
  every buyer. Hence, $OD(g;\bp) = \sam-1$.
\end{proof}

Note that if buyers resolve indifferences uniformly at random, $\sam/2$ buyers
will attempt to buy good $g$.  While one might hope that the right choice of
tie breaking rules would solve this problem, we show that it cannot.  Formally,
we will suppose that each buyer $q$ has a tie breaking rule $e_q^{\bp}$ which
can depend on $v_q$ \emph{and} the price vector $\bp$.  Then, we consider the
over-demand of a good when all bidders use their tie breaking rules.
\begin{definition}[Tie-breaking over-demand]
Given a set of prices $\bp$, buyers $N$ and tie breaking rule $e^\p_q$ for buyer
$q$, let the demanders of $g$ be
$$
U^{e}(g;p) = \{q \in N: g \in \cD_q(p) \text{ and } g =
e_q^p(\cD_q(p)) \},
$$
and the \emph{tie breaking over-demand} with respect to $e^p$ be
$$
OD^{e}(g;p)
= \max\{ |\{ q \in U^{e}(g;p)\} |-s_g , 0\} .
$$
\end{definition}

Without loss of generality (by the min-max principle), it
suffices to consider deterministic tie breaking rules when
constructing a randomized lower bound instance.

\begin{lemma}\LABEL{lem:bad2}
  There exists a distribution over unit demand valuations such that
  for any set of tie breaking rules, the expected tie breaking over-demand from
  $n$ buyers is $\Omega(n)$.
\end{lemma}
\begin{proof}
  The distribution we construct will contain instances similar to the
  one from \cref{lem:bad1}, with $\sam$ unit demand buyers and
  $\ngood = \sam$ distinct goods. We choose a permutation $\sigma$
  over the goods $[\ngood]$ and a
  distinguished good $g^* \in [\ngood]$ uniformly at random. Once
  $\sigma$ and $g^*$ are chosen, we define buyer valuations such that
  for all buyers $q$, $\val{\sigma(q)}{q} = \val{g^*}{q} = 1$ and
  $\val{h}{q} = 0$ for all $h \not\in \{\sigma(q),g^*\}$. In
  this market, the minimal Walrasian equilibrium prices is
  $\bp = \mathbf{0}$ and the unique max-welfare allocation
  has $\mu_q = \sigma(q)$ for all $q \in N$.  At equilibrium, there
  is a single buyer $q^*$ for whom $\sigma(q^*) =g^*$ and
  $\dem{q^*}(p) = \{g^*\}$. For every other buyer,
  $\sigma(q) \neq g^*$ and $q$ has exactly two goods in their demand
  correspondence: $\dem{q}(p) = \{\sigma(q),g^*\}$. Each such buyer
  will select a good
  $e_q^p(\{\sigma(q), g^*\}) \in \{\sigma(q), g^*\}$ using their
  tie breaking rule.

 Over the randomness of the instance ($\sigma$ and $g^*$) we have $
 \Prob{}{e_q^p(\{\sigma(q),g^*\})= g^*} = 1/2.$ We can then lower bound the expected
 tie breaking over-demand for the distinguished good as
 \begin{align*}
 &\Ex{\sigma,g^*}{OD^{e}(g^*;p)} \\
 &= \sum_{q \in N} \Prob{}{\sigma(q) = g^*} \\
 &+ \Prob{}{\sigma(q) \neq g^*}\cdot
 \Prob{}{e_q^p(\{\sigma(q),g^*\})= g^*} - \underbrace{1}_{\text{supply}} \\
 & = \frac{1}{2} \cdot n(1-1/n) = \frac{n-1}{2}
 \end{align*}
which completes the proof.
\end{proof}

While this result shos that over-demand can be high
without coordination, it seems rather artifical---the buyers'
valuations are extremely similar. We will soon give a
simple and natural condition which rules out this example, and more
generally ensures that the over-demand for any good is at most $1$ at
the minimal Walrasian prices regardless of tie breaking. This bound is the best
possible: for \emph{any} buyer valuations, minimal Walrasian prices
always induce over-demand of at least $1$ for every good with positive
price.
\begin{lemma}
  \LABEL{lem:od-1}
  Fix any set of buyer valuations $\{ v_q : q \in [\sam]\}$, and let
  $\bp$ be a minimal Walrasian equilibrium price vector.
  For any good $g$ with positive price $p(g)>0$, we have
  $OD(g;\bp) \geq 1$.
\end{lemma}
\begin{proof}
  Let $\mu$ be a Walrasian allocation for $\p$. By the Walrasian
  equilibrium condition,
  $|q : g \in \mu_q| = s_g$ for any good with $p(g) > 0$. Suppose that
  $OD(g;\bp) = 0$, i.e., $g$ is in some
  bundle in buyer $q$'s demand set if and only if
  $g \in \mu_q$. In this case there exists
  $\epsilon > 0$ such that if we set
  $p(g) \leftarrow p(g) - \epsilon$, $(p, \mu)$ remains a Walrasian
  equilibrium---the allocation is unchanged, and every buyer
  continues to receive an allocation in their demand set. But this
  contradicts minimality of $\bp$.
\end{proof}

%% file: matchings.tex
\section{Unit demand} \LABEL{sec:matchings}

Now that we have seen how indifferences at equilibrium can lead
to over-demand, we consider whether the
over-demand is large for ``typical'' instances.  To build intuition,
we start with the special case of unit demand valuations, where
$v_q(S) = \argmax_{g \in S} \{ \val{g}{q}\}$. Such valuations can be encoded
with a real number for each pair of buyer $q \in N$ and good $g \in [\ngood]$,
denoted $\val{g}{q}$. Thus, without loss of generality
an allocation is a many-to-one matching between
buyers and goods---each buyer should be matched to at most $1$ good,
and each good $g$ should be matched to at most $s_g$ buyers.

We now give conditions on unit demand valuations to ensure that when buyers buy
\emph{arbitrary} (singleton) demand sets from their demand correspondence at the
minimal Walrasian prices, the
resulting allocation has high welfare and low
over-demand. Accordingly, we need to reason precisely about how
the equilibrium prices depend on the valuations.  Getting access to
this relation is surprisingly tricky---typical characterizations of
Walrasian equilibrium prices are not enough for our needs.
For instance, two standard characterizations show that unit demand
Walrasian prices i) are dual variables to a particular linear program
(the ``many-to-one matching linear program''), and ii) are computed
from ascending price auction dynamics.  The first observation reduces computing
prices to an optimization problem, but
it does not provide fine-grained
information about how the prices depend on the valuations. The second
observation is useful for computing prices, but the auction may proceed in a
complicated manner, obscuring the relationship between the prices and
valuations.

\SUBSECTION{Swap Graph}

Accordingly, we define a graph called the \emph{swap graph} of a Walrasian
equilibrium $(\bp,\mu)$. This graph directly encodes buyer indifferences induced
by the equilibrium prices. Furthermore, the swap graph allows us to read off
equations involving the prices and valuations. We define the swap graph as
follows.

\begin{definition}[Swap graph]
  The \emph{swap graph} $G = (V,E)$ defined with respect to a
  Walrasian equilibrium $(\bp, \mu)$ has a node for each good $g$ and
  an additional null node $\bot$ representing the empty
  allocation: $V = [\ngood] \cup \{\bot\}$.  There is a directed
  edge $(a,b) \in E$ for $a \neq b, b \neq \bot$ for each buyer $q$ that
  receives good $a$ in $\mu$ but also demands $b$,
  i.e. if $\mu_q =a$ and $b \in \dem{q}(\bp)$ for some
  buyer $q \in N$.  Note that while there may be parallel edges---representing
  the same indifferences by different buyers---there are no self loops.
\end{definition}

Since we will phrase our arguments in terms of the swap graph in the remainder
of the section, we will first recast \cref{lem:od-1} using our new langauge.

\begin{corollary}
  For buyers with unit demand valuations and a minimal Walrasian equilibrium
  $(\bp,\mu)$, every node in the swap graph $G$ with in-degree zero has price
  zero.
  \LABEL{lem:pzero}
\end{corollary}
Almost by definition; the over-demand of a good $g$ is its in-degree in
the swap graph.
\begin{lemma}
  Let $G$ be the swap graph corresponding to a Walrasian
  equilibrium. If a node $g$ in $G$ has in-degree $d$, then
  $OD(g;\bp) \leq d$.
\LABEL{lem:deg_od}
\end{lemma}
\begin{proof}
  By construction of the swap graph, a node with in-degree $d$
  corresponds to a good $g$ with $d$ buyers with $g$ in
  their demand correspondence but not in their allocation.  Because
  $\mu$ is a feasible allocation, at most $s_g$ buyers
  can be allocated good $g$ in $\mu$, and since $\mu$ is an equilibrium
  allocation, $g$ is in the demand correspondence for each of these buyers.  Thus,
  there can be at most $s_g + d$ demanders for good $g$. By
  definition of over-demand, we have $OD(g;\bp) \leq d$.
\end{proof}
So to bound the maximum over-demand of any good, it suffices to bound the
in-degree for every good in the swap graph. While the in-degree may be large in
the worst case, we can introduce a simple condition on valuations that will rule
out these pathological market instances.

\SUBSECTION{Generic Valuations}
Recall that in \cref{sec:lb}, we showed that
over-demand can be high at the minimal Walrasian equilibrium
prices. So, to provide a better bound on over-demand, we need
additional assumptions; ideally, a condition that
will hold ``typically''.  In the lower bound instance
\ifshort for \cref{lem:bad2}\else from \cref{sec:lb}\fi,
over-demand is large because the
buyers have valuations that are too similar. Indeed,
consider a market with two goods $a$ and $b$ where all buyers have the
same \emph{difference} in valuations between $a$ and $b$. If
\emph{some} buyer is indifferent between $a$ and $b$---by
\cref{lem:od-1}, this must be the case at minimal Walrasian prices---\emph{all}
buyers are indifferent. This observation motivates our genericity
condition.

\begin{definition}[Generic valuations]
  \LABEL{def:generic-unit}
  A set of valuations $\{ \val{g}{q} \in \R : q \in N, g \in [\ngood]\} $
  is \emph{generic} if they are linearly independent over $\{ -1,0,1\}$, i.e.
  \[
    \sum_{q \in N} \sum_{g \in [\ngood]} \alpha_{q,g} \val{g}{q} = 0
    \qquad
    \text{ for } \alpha_{q,g} \in \{ -1,0,1\}
\ifsubmission
  \] 
  \[
\else 
    \quad 
\fi
    \text{ implies } \alpha_{q,g} = 0 \text{ for all } q \in N, g \in [\ngood].
  \]
\end{definition}

\begin{remark}
  Note that this condition holds with probability $1$ given any
  continuous perturbation of a profile of valuation functions, and so
  for many natural distributions, a profile of valuation functions
  will ``generically'' (i.e., with high probability) satisfy our
  condition. We also show how to discretely perturb a
  fixed set of valuations to satisfy our
  condition in \thelongref{sec:perturb}.
\end{remark}

\begin{remark}
Our definition of generic in the unit demand setting is also called
\emph{independence} by \citet{ABH09,BCK14}, although over-demand was not the
focus in these works.  In fact, \citet{BCK14} give an alternative way to bound
over-demand for generic unit demand valuations (\cref{thm:unit_od}). We present
the unit-demand case using the swap graph construction, in order to generalize
smoothly to broader classes of valuations.
\end{remark}

\SUBSECTION{Over-Demand}
Now, we are ready to present the main technical result of this section:
When buyers with generic valuations select an
\emph{arbitrary} good in their demand correspondence given minimal
Walrasian equilibrium prices, over-demand is low and welfare is
high. We will show that the \emph{in-degree} for
any node in our swap graph is at most $1$. This will imply that no good has
over-demand more than $1$, regardless of its supply.
We proceed via a series of properties of the swap graph.  First, under
genericity, the swap graph is acyclic.

\begin{lemma}
  The swap graph $G$, defined with respect to Walrasian equilibrium
  $(\bp,\mu)$ and generic valuations
  $\{\val{g}{q} : q \in N, g \in [\ngood]\}$, is acyclic.
  \LABEL{lem:acyclic_match}
\end{lemma}
\begin{proof}
  Since the null node $\bot$ has no incoming edges by construction, it
  cannot be part of any cycle. Thus, suppose that there is a cycle
  $a_0 \to a_1 \to \dots \to a_k \to a_0$ of non-null nodes.  We label
  the buyers so that buyer $q_i$ is allocated $a_i$ in the allocation
  $\mu$, but also has good $a_{i+1}$ in their demand correspondence as
  well (taking the subscript modulo $k+1$) for $i = 0, \dots, k$.

  By construction, all buyers are distinct and $k \geq 1$.
  Furthermore, each edge represents an
  indifference relationship for some buyer. In particular,
  \[
    \val{a_q}{q} - p(a_q) = \val{a_{q+1}}{q} - p(a_{q+1})
    \qquad \text{for all } \quad
    q = 0, \dots, k.
  \]
  Summing these equations and canceling prices, we have
  \[
    \sum_{i = 0}^k \val{a_i}{q_i} = \sum_{i = 0}^k \val{a_{i+1}}{q_i},
  \]
  contradicting the genericity assumption since all buyers $q$ are distinct.
  Hence, the swap graph must be acyclic.
\end{proof}

Because the swap graph is acyclic, we can choose a partial order
of the nodes so that all edges go from smaller nodes to larger nodes. For the
remainder of the argument, we assume nodes are labeled by such an ordering
(i.e. we now have for every edge $(a_i,a_j) \in E$ implies $i < j$).  Now, the
price of a good can be written in terms of the valuations of smaller goods.

\begin{lemma}
  For every good $g$, the price $p(g)$ can be written as a linear combination of
  valuations $\val{j}{q}$ over $\{-1,0,1\}$ for $q \in N$ and $j < g$.
  Specifically, for every set of goods $g_1 < g_2 < \ldots < g_{k} < g$ such
  that $g_1 \to \ldots \to g_k \to g$ forms a path in the graph and $g_1$ has
  in-degree $0$, there are buyers $q_1,\ldots,q_k$ such that $\mu_{q_i} = g_i$
  and
  \[
    p(g) = \sum_{i = 1}^{k} \left( \val{g_{i+1}}{q_i} - \val{g_{i}}{q_i} \right)
  \]
  where $g_{k+1} = g$.   If the first node $g_1 = \bot$, we define $\val{\bot}{q_1} = 0$.
  \LABEL{lem:lincomb}
\end{lemma}
\begin{proof}
  We proceed by induction on the path length $k$.  In the
  base case, $k=1$ and we have path $g_1 \to g$.  If $g_1 = \bot$ then
  buyer $q_1$ must be indifferent between the empty allocation and
  good $g$, so $p(g) = v_{q_1}(g)$.  If $g_1 \neq \bot$, then
  buyer $q_1$'s indifference between $g_1$ and $g$ yields
  $$
  v_{q_1}(g_1) - p(g_1)= v_{q_1}(g) - p(g) ,
  $$
  so $p(g) = v_{q_1}(g)- v_{q_1}(g_1) + p(g_1)$.
  Note that $g_1$ has in-degree zero so
  \cref{lem:pzero} implies $p(g_1) = 0$. This shows the base case.

  For the inductive case, we assume that we can write prices in our
  desired form for any good with a path of length at most $k-1$ from a
  source node. Consider a good $g$ that has path length $k$ from a
  source node.  Because buyer $q_k$ is indifferent between $g$ and
  $g_k$,
  $$
  v_{q_k}(g_k) - p(g_k)= v_{q_k}(g) - p(g)
  $$
  so $p(g) = v_{q_k}(g)- v_{q_k}(g_k) + p(g_{k})$.
  Note that $p(g_k)$ has a path length of at most $k-1$ from a source
  node, and so we can apply the induction hypothesis to get complete the
  induction:
  $$
  p(g) = v_{q_k}(g)- v_{q_k}(g_k) + \sum_{i = 1}^{k-1} \left(v_{q_i}(g_{i+1}) - v_{q_i}(g_i) \right).
  \qedhere
  $$
\end{proof}

Finally, we can bound the in-degree of any node under genericity.

\begin{lemma}
  For generic buyer valuations, every node in the
  the swap graph defined with respect to a Walrasian equilibrium
  $(\bp, \mu)$ with minimal Walrasian prices has in-degree
  at most $1$.
  \LABEL{lem:indeg1}
\end{lemma}
\begin{proof}
  Suppose otherwise and let good $g$ be the smallest indexed good with
  at least two incoming edges. Then there are two sequences of goods
$$
g_1 \leq \ldots \leq g_k \leq g
\qquad \text{and} \qquad
g_1'\leq \ldots\leq g_{k'}' \leq g
$$
such that $g_k \neq g_{k'}'$, $g_1 \to \ldots \to g_k \to g$ and
$g_1' \to \ldots \to g_{k'}' \to g$ form paths in the graph, and $g_1$
and $g_1'$ have in-degree $0$. We will write $g = g_{k+1} = g'_{k'+1}$.  By
\cref{lem:lincomb}, we can express $p(g)$ in two distinct ways:
$$
p(g) = \sum_{i = 1}^{k} \left( \val{g_{i+1}}{q_i} - \val{g_{i}}{q_i} \right)
$$
and
$$
p(g) = \sum_{i = 1}^{k'} \left( \val{g_{i+1}'}{q_i'} - \val{g_{i}'}{q_i'}
\right) ,
$$
where buyer $q_i$ has $\mu_{q_i} = g_i$ and also has
good $g_{i+1}$ in her demand set, and buyer $q_i'$ has
$\mu_{q_i'} = g_i'$ and has good $g_{i+1}'$ in her demand
set. Taking the difference, we have
$$
  \sum_{i = 1}^{k} \left( \val{g_{i+1}}{q_i} - \val{g_{i}}{q_i} \right) -
  \sum_{i = 1}^{k'} \left( \val{g_{i+1}'}{q_i'} - \val{g_{i}'}{q_i'}\right) = 0
  .
 $$
 Since we either have $q_k \neq q_{k'}$ or $g_k \neq g'_{k'}$, the above linear
 combination is not trivial and contradicts genericity.
\end{proof}

Finally, by \cref{lem:deg_od} the over-demand for any good is at most
its in-degree in the swap graph, so we have bounded over-demand under generic
valuations.
\begin{theorem}
  \LABEL{thm:unit_od}
  For any set of unit demand buyers with generic valuations and for
  $\bp$ the minimal Walrasian equilibrium price vector, the over-demand for any
  good $g \in [\ngood]$ is at most $1$.
\end{theorem}

As a result, when generic buyers face minimal\footnote{%
  We show in \ifshort the full version \else \cref{sec:nonmin} \fi
  that even with generic valuations, \emph{non-minimal}
  Walrasian prices can still induce high over-demand, further
  justifying our focus on minimal Walrasian prices.}
prices $\bp$ and buy a good in their demand set while resolving indifferences
arbitrarily, the excess demand of any good is at most $1$.

\SUBSECTION{Welfare}
Now that we have considered over-demand under genericity, what about welfare?
If buyers break ties arbitrarily, it is not hard to see that welfare may be very
bad: buyers who are indifferent between receiving a good and receiving nothing
may all decide to demand nothing, giving zero welfare. However, if we simply
rule out this specific kind of indifference, we can show that genericity implies
near-optimal welfare.

\ifshort
In our calculations of welfare, $\wel{B}{N}$, for some (possibly
infeasible) allocation $B$, we assume over-demand is resolved in a
worst-case way---i.e. we assume goods are allocated to buyers that
demand them in the way that minimizes welfare, while obeying the
supply constraints.
\else If $B$ is not a feasible allocation, we assume $\wel{B}{N}$ is
calculated by resolving over-demand in a worst-case way,  as~\cref{lem:welfare-equiv} describes below.
As an intermediate step in our calculations we bound a relaxed notion of welfare, which
\emph{assumes supply is sufficient to satisfy the demand of all
  buyers}. Given bundles $B_1, \ldots, B_\sam$, we define the
\emph{relaxed} welfare of this (pseudo-)allocation as follows.
\begin{definition}[Relaxed Welfare]
   Given $\sam$ bundles $B_1, \dots, B_\sam$, the \emph{relaxed
    welfare} for a market $N$ is
  $\rwel{B_1, \ldots, B_\sam}{N} = \sum_{q\in N} v_q(B_q)$.
  \end{definition}
  Note that because $B_1, \ldots, B_n$ may not be a feasible
  allocation, $\rwel{B_1, \ldots, B_\sam}{N}$ can in principle be
  larger than the optimal welfare obtainable over feasible allocations
  in the market over buyers $N$, which we denote $\optw{N}$. However,
  subject to our genericity conditions, we need only consider
  allocations that over-allocate any good by at most $1$.  For such
  allocations, the difference between $\rwel{B}{N}$ and the welfare of
  a corresponding feasible solution is small (and in general, if the over-demand can be bounded by a small quantity, the gap between welfare and relaxed welfare is small).  Furthermore, if $B$ is a feasible allocation then
  $\wel{B}{N} = \rwel{B}{N}$.

\begin{lemma}\LABEL{lem:welfare-equiv}
  Consider a set of bundles $B_1, \ldots, B_\sam$ such that
  $|\{q : g\in B_q\}| \leq s_g +d$ for all $g$ and some $d\leq s_g$. We consider
  two possible ways to convert $\{ B_i \}$ into a feasible allocation $\{
  \hat{B}_i \}$.
\begin{itemize}
\item Let $\sigma:N \to [n]$ be an ordering of the buyers. For any buyer $q\in
  N$ with bundle $B_q$, if for some $g \in B_q$ more than $s_g$ buyers $q'$ who precede $q$ in the ordering
  demand $g$, let $\hat{B}_q = \emptyset$.
  Else, let $\hat{B}_q = B_q$. (i.e. we let buyers choose bundles in order while supply remains)

  Then, $\hat{B}_1, \ldots, \hat{B}_\sam$ is feasible, and
\[
  \rwel{B_1, \ldots, B_\sam}{N} \leq
  \wel{\hat{B}_1, \ldots, \hat{B}_\sam}{N} -
  d\cdot \ngood \cdot \maxval .
\]
\item For each $g$, randomly select a subset $N_g$ of $s_g$ buyers with $g\in B_q$; let $g\in
  \hat{B}_q$ if $q\in N_g$. (i.e. independently for each good, allocate its supply uniformly at randomly to its set of demanders)

  Then, $\hat{B}_1, \ldots, \hat{B}_\sam$ is feasible, and if buyers are
  subadditive:
\[
\left(1- \frac{d}{d +\min_g s_g}\right) \cdot
  \rwel{B_1, \ldots, B_\sam}{N} \leq
 \ex{\wel{\hat{B}_1, \ldots \hat{B}_\sam}{N}} .
\]
\end{itemize}
\end{lemma}
The proof of~\cref{lem:welfare-equiv} can be found in~\cref{app:welfare-equiv}.

\fi
For the welfare argument only, we assume that a buyer
chooses a \emph{nonempty} demand bundle whenever possible.
Subject to this restriction, buyers can break ties however they
like. Then, we can lower bound the welfare of these buyers $N$
in terms of the optimal welfare $\optw{N}$.

\begin{theorem}\LABEL{thm:unit_welfare}
  Consider any set of buyers $N$ with generic unit demand valuations
  bounded in $[0, \maxval]$ and minimal Walrasian equilibrium
  prices $\bp$. For each buyer $q$, let $b_q\in \dem{q}(p)$ be some
  arbitrary set in their demand correspondence, assuming only that
  $b_q \neq \emptyset$ whenever $|\dem{q}(p)| > 1$.  Then the welfare
  obtained by the resulting allocation is nearly optimal:
  \[
    \wel{b_1,b_2, \dots, b_n}{N}  \geq \optw{N} -  2 \cdot\ngood\cdot \maxval.
  \]
\end{theorem}
\ifshort\else
\begin{proof}
  Let $\mu$ be a Walrasian allocation for prices $\p$ and, thus, a
  welfare-optimal allocation so that
  $\optw{N} = \sum_{q \in N }v_q(\mu_q)$. We also know that
  $\mu_q \in \dem{q}(p)$ for all buyers $q$, by the properties of
  Walrasian equilibria. Because we also have $b_q \in \dem{q}(p)$ for
  every $q$, we know
$$ \optw{N} - \sum_{g \in [m]} p(g) s_g = \sum_{q \in N} \left[ v_q(\mu_q) -
  p(\mu_q) \right] = \sum_{q \in N} \left[ v_q(b_q) - p(b_q) \right].$$
Rearranging,
\begin{align*}
\sum_{q \in N}& \left[ v_q(b_q) - p(b_q) \right] = \sum_{q \in N} v_q(b_q) - \sum_{g \in [m]} p(g) s_g  - \sum_{g \in [m]} p(g) \cdot od(g,b) + \sum_{g \in [m]} p(g) \cdot ud(g,b)
\end{align*}
where $od(g,b)$ and $ud(g,b)$ are defined to be the over-demand and under-demand, respectively, of good $g$ in allocation $b$ compared to the optimal allocation $\mu$:
\begin{align*}
od(g,b) &= \max\left\{\sum_{q \in N} \left[ \1\{g = b_q \} - \1\left\{ g = \mu_q \right\}\right] ,0\right\}
\\
ud(g,b) &= \max\left\{\sum_{q \in N} \left[ \1\{g = \mu_q \} - \1\left\{ g = b_q \right\}\right] ,0\right\}
\end{align*}

Note that $od(g,b) \geq 0$, and that minimal Walrasian prices are at
most the maximum value $\maxval$, and hence we can rearrange the above
equations to obtain
$$
\optw{N} \leq \sum_{q \in N} v_q(b_q) + \sum_{g \in [m]} ud(g,b) \cdot \maxval.
$$
Hence, it suffices to upper bound the total under-demand induced over all goods $g$ by $b$.

Next, note that by assumption, $b_q = \emptyset$ if and only if
$\dem{q}(p) = \{\emptyset\}$, in which case it must also be that
$\mu_q = \emptyset$. Hence, we have that
$|\{q : b_q = \emptyset\}| \leq |\{q : \mu_q = \emptyset\}|$.  We then
have
$$
 \sum_{q \in N } |\mu_q| \leq \sum_{q \in N} |b_q| = \sum_{g \in [\ngood]} \left(
   |\mu_q| + od(g,b) - ud(g,b)\right),
 \qquad \text{and so} \qquad
\sum_{g \in [\ngood]} ud(g,b)  \leq \sum_{g \in [\ngood]} od(g,b) .
$$

Under genericity, \cref{thm:unit_od} implies $od(g,b) \leq 1$, and hence
$$
\sum_{g \in [\ngood]} ud(g,b)  \leq \sum_{g \in [\ngood]} od(g,b) \leq \ngood .
$$
Thus, we obtain:
$$
\optw{N} \leq \sum_{q \in N} v_q(b_q) + \ngood\cdot\maxval = \rwel{b_1, \ldots, b_\sam}{N} + \ngood\cdot \maxval
$$
Finally, applying \Cref{lem:welfare-equiv} completes the proof.
\end{proof}
\fi

\ifshort
\else
\paragraph{Generalization Results for Unit Demand Buyers}
\LABEL{sec:unit-generalization}
We now state our main generalization results specialized to unit
demand buyers. \Cref{sec:generalize} states the analogous
generalization theorems for buyers with arbitrary valuation functions,
and provides the proofs of the general statements. Fixing some
tie-breaking rule $e: 2^\cG\times \V \to \cG$ for choosing a demanded
set from their demand correspondence;\footnote{For this section,
  assume only that $e(\X) \neq \emptyset $ unless
  $\X = \{\emptyset\}$.} we will use
$\ndem{g}{p}{N}{e} = \sum_{q\in N}\I[g\in e(\dem{q}(\p), v_q)]$ to
denote the \emph{number} of copies of good $g$ demanded in market $N$
at prices $\p$, when buyer $q$ uses $e(\cdot, v_q)$ to break ties. Our
first theorem bounds the over-demand for $g$ on $N'$.

\begin{theorem}\LABEL{thm:unit-demand-bernstein}
  Fix some pricing $\p$ and two sampled markets $N, N'$ of unit demand
  buyers, with $|N| = |N'| = \sam$. Then, fixing a good $g$, if
  $\ndem{N}{g}{\p}{e} \leq s_g + 1$, with probability $1-\delta$,
  \rynote{Would be nice to define the $OD$ according to the tie
    breaking, which we could then state as being at most the LHS of
    the next equation.}
\[\ndem{g}{\p}{N'}{e} - s_g  \leq  O\left(\sqrt{\ngood\cdot s_g\ln\frac{1}{\delta}} + \ngood\ln\frac{1}{\delta}\right). \]
Thus, for any $0<\alpha<4/5$, if
$s_g = \Omega\left(\ngood\ln\frac{1}{\delta}/\alpha^2\right)$, we have
\[\ndem{g}{p}{N'}{e} \leq (1+\alpha) s_g.\]
\end{theorem}

Our second theorem guarantees that the exact Walrasian prices $\bp$
computed on a market $N$ will induce approximately optimal welfare
when used on a new set of buyers $N'$ sampled from the same
distribution.  Let $\welfare{N}{\p}$ denote the welfare buyers in $N$
achieve at prices $\p$, assuming over-demand for a good $g$ is
resolved in a worst-case way for welfare.

\begin{theorem}\LABEL{thm:unit-welfare-bernstein}
 Fix two sampled markets $N, N'\sim \D$ of unit demand buyers for which $|N| = |N'|
  =\sam$.
 For any $\alpha \in (0, 4/5)$, if $p$ are welfare-optimal prices for $N$ and
\[
\optw{N}=\Omega\left(\frac{\maxval^{3}\sam^{.5}\ngood^2\log^4(\ngood)\ln^{2}\left(\ngood\right)\ln^{2}\frac{1}{\delta}}{\alpha^2}\right)
 \]
then, with probability $1-\delta$,
\[
\welfare{N'}{\p} \geq (1- \alpha)\optw{N'}.
\]
\end{theorem}
\fi

%% file: gs.tex
\section{Towards gross substitutes}\LABEL{sec:gs}
Now that we have seen our techniques in the unit demand setting, we will
generalize our results to more general valuations.  Ideally, we would like to
extend our results to buyers who have \emph{gross substitutes} valuations.
\begin{definition}[Gross substitutes]
  \LABEL{def:gs}
  A valuation function $v_q$ satisfies \emph{gross substitutes} (GS)
  if for every price vectors $\bp' \geq \bp$ and
  $S \in \dem{q}(\bp)$, there is a bundle $S' \in \dem{q}(\bp')$ with
\[
S \cap \{ g \in [m] : p(g) = p'(g)\}\subseteq S'.
\]
Note that unit demand valuations also satisfy GS.
\end{definition}

If all buyers have GS valuations, there always exist Walrasian equilibria, and
the minimal equilibrium can be found by following a tat\^onnement procedure
\citep{KC82}.  \citet{GS99b} show that the class GS is in some sense the most
general class of valuations that are guaranteed to have WE.

While GS is an attractive class to target, its definition is axiomatic rather
than constructive; \cref{def:gs} shows us what properties a GS valuation must
satisfy, rather than how to concretely construct GS valuations.
This poses a problem for defining genericity: it is not obvious which, if any,
valuations in GS satisfy a candidate definition of genericity, so it may be that
generic valuations do not even have Walrasian equilibria!

Ultimately, we will prove our results for buyers with \emph{matroid
  based valuations} (MBV), a subclass of gross substitutes valuations
that is conjectured to be equal to all of gross substitutes
\citep{OL15}. Such valuations can be explicitly constructed from a set
of numeric weights and a matroid, giving us a path to define generic
valuations and certify that they are contained in GS.

However, there are still several obstacles to overcome, and the arguments are
significantly more involved than for unit demand. Roughly, the central
difficulty is establishing a connection between valuations on bundles and
valuations on the items in the bundle, since we will require genericity on
values for single goods.  While this connection is immediate in the case of unit
demand---favorite bundles are simply single goods---the situation for MBV is
more delicate. For instance, some bundles may contain ``irrelevant'' goods,
which don't contribute at all to the valuation. Hence, we propose a more complex
swap graph to capture the richer structure of the valuations.

Despite the complications, our high-level argument parallels the unit demand
case. We will first define a swap graph and connect it to over-demand. Then, we
define a generic version of MBV (GMBV) and prove properties about the swap graph
for valuations in GMBV. We only restrict to MBV when defining genericity; our
arguments in the first step apply to GS valuations.

\SUBSECTION{Swap Graph with GS Valuations}

To define the swap graph for gross substitute valuations, the core idea is to
have an edge
$(a, b)$ represent a single \emph{swap} of good $a$ for good $b$ in some larger
bundle, naturally generalizing our construction for unit demand. The main
challenge is ensuring that we faithfully model buyer indifferences---which are
between \emph{bundles} of goods---via indifferences of single swaps. More
precisely, in order to bound over-demand with arbitrary tie-breaking, we must
ensure that if a buyer is indifferent between her Walrasian allocation
and some other bundle $B$, then there must be at least one incoming edge to every
good in $B$. If the swap graph satisfies this property, we can describe
potential demand for goods in terms of in-degree of nodes like we did for
unit demand.

At first, we might hope that if a buyer is indifferent between bundles $B_1$ and
$B_2$, every good in $B_1$ can be exchanged for a good in $B_2$ while preserving
utility; this would ensure that each good in $B_2$ has an incoming edge. While
this property turns out to not be true for general bundles, it is true for the
``smallest'' bundles in a buyer's demand correspondence.

\begin{definition}[Minimal demand correspondence]
  For price vector $\bp$ and buyer $q \in N$, the \emph{minimum demand
    correspondence} is
  $$
\mindem{q}(\bp) = \{ S \in \dem{q}(\bp) : T \notin \dem{q}(\bp) \text{ for all } T \subsetneq S \} .
$$
We call bundles $B \in \mindem{q}(\bp)$ \emph{minimum demand
  bundles}, or \emph{minimum bundles} for short.
\end{definition}

Crucially, the minimum demand correspondence for any buyer with a GS valuation
forms the bases of a matroid \citep{BLN13}.
(We give a self-contained proof in \thelongref{sec:omit_gs}.)
\begin{lemma}[\citet{BLN13}] \LABEL{lem:mindem-bases} Let
  $v_q: \typebun \to [0, \maxval]$ be a GS valuation. For any price
  vector $\bp$, the minimum demand correspondence $\mindem{q}(\bp)$
  forms the set of bases of some matroid.
\LABEL{lem:bases_orig}
\end{lemma}
By standard facts from matroid theory \citep{Matroid},
all elements of $\mindem{q}(\bp)$ are the same size. More importantly,
the bases satisfy the exchange property: for $B_1, B_2$ in the basis set $\cB$
of a matroid and for every $b \in B_1\setminus B_2$, there exists $b' \in B_2
\setminus B_1$ such that
\[
B_1 \cup b' \setminus b \in \cB \longquad \text{ and } \longquad B_2 \cup b \setminus b' \in \cB.
\]

Thus if all bundles in a buyer's demand correspondence have the same size, then
we can take $\cB$ to be $\mindem{q}(\bp)$ and we have exactly what we need: for
any two bundles in the correspondence, every good can be swapped for some other
good while remaining in the correspondence.  Unfortunately, buyers may prefer
bundles of different sizes. We will now introduce the swap graph, and then
explain how it models buyers who prefer bundles of different sizes.

\begin{definition}[Swap graph]
Let buyers have GS valuations, $(\bp,\mu)$ be a WE, and for each buyer $q \in
N$ fix a minimum demand bundle $M_q \in \mindem{q}(\bp)$ where $M_q \subseteq
\mu_q$.  Define the \emph{swap graph} $G(\bp,\mu,(M_q)_{q\in N})$ to
have a node for every good $g \in [m]$ and an additional node which we refer to
as the \emph{null node} $\bot$. There is a directed edge from $(a,b)$ for every
buyer $q \in N$ such that $a \in M_q$, $b \notin \mu_q$, and there exists $B
\in \mindem{q}(\bp)$ with $b \in B$ where
\[
M_q \cup b \setminus a \in \mindem{q}(\bp).
\]
Further, the graph contains an edge from $\bot$ to good $b$ for
each buyer $q \in N$ with $b \notin \mu(q)$ such that (i) $b \in B$ with
$B \in\dem{q}(\bp)$ and $B\setminus b \in \mindem{q}(\bp)$, and (ii) $b$ has
positive price.
\LABEL{def:gs_swap}
\end{definition}

There are two main conditions when specifying edges from the null good: (i)
involves demand bundles, while (ii) requires positive price. These two
conditions address two distinct problems.

The first condition models cases where the demand
correspondence has bundles of different sizes. Suppose there is some
bundle $B$ strictly larger than the minimum
cardinality demand bundle. We need all goods in $B$ to have an
incoming edge, to reflect over-demand if a buyer selects $B$. We
cannot ensure these edges via matroid properties, since $B$ is not a
minimum cardinality demand bundle and so it is not a basis of the matroid.
However, there is always a minimum bundle $B^{min}$ \emph{contained in} $B$,
which \emph{is} a matroid basis, so we have incoming edges to
$B^{min}$.

The following interpolation lemma shows that for every good $g \in B \setminus
B^{min}$ the bundle $B^{min} \cup g$ is also a demand bundle, so the swap graph
will also have edges from $\bot$ to $g$, covering all goods in $B$ as desired.
\begin{lemma}
  \LABEL{lem:demand-interpol}
  Consider a buyer with a GS valuation. If we have two bundles $B, B'$
  such that $B \subseteq B'$, and $B, B'$ are both bundles in the
  buyer's demand correspondence
  at price $\bp$, then
  \[
     u(B ; \bp)= u(B''; \bp) = u(B'; \bp)
   \]
  for any $B''$ with $B\subseteq B'' \subseteq B'$.
\end{lemma}
\begin{proof}
  Let $B''$ satisfy the conditions in the lemma statement and define
  $\{b_1, \dots, b_k \} = B'' \setminus B$; we proceed by induction on $k$.
  Since GS valuations are submodular,
  \[
    u(B';\bp) - u(B' \setminus b_1;\bp) \leq u(B \cup b_1;\bp) - u(B;\bp) ,
  \]
  so
  \[
    u(B';\bp) \leq u(B \cup b_1;\bp) - u(B;\bp) + u(B' \setminus b_1 ;\bp) .
  \]
  Since $B$ is most demanded, $u(B \cup b_1;\bp) - u(B;\bp) \leq 0$.
  If $B \cup b_1$ is strictly worse, then $B' \setminus b_1$ must have
  strictly higher utility than $B'$, which contradicts $B'$ being a
  most demanded bundle. So, we must have
  $u(B \cup b_1;\bp) = u(B;\bp)$.

 By induction we assume that $B \cup \{b_1,\dots, b_{t-1} \}$ is a
 most demanded bundle.  Using the same argument as in the base case,
 we can use submodularity to show that
 \[
   u(B \cup \{b_1,\dots, b_{t-1} \}; \bp) = u(B \cup \{b_1,\dots, b_{t} \};\bp) .
 \]
 Hence, $B''$ is a most demanded bundle at price $\bp$.
\end{proof}

The second condition for edges from the null good---giving edges only to goods
with positive price---handles a subtle case.
Up to this point, we have argued informally that if $B$ is a bundle in
$\mindem{q}(\bp)$, then all goods in $B \setminus M_q$ will have an incoming
edge. We plan to bound the over-demand by the in-degree, but there is an
important corner case: goods with price $0$. Such goods can be freely added
to \emph{any} buyer's bundle, ruining the bound on over-demand.  However, all is
not lost: the problem stems from buyers demanding goods with with zero marginal
valuation. We call bundles with such goods \emph{degenerate}, and assume (or
require) that buyers do not select them.

\begin{definition}[Non-degenerate]
\LABEL{def:typemin}
A bundle $S$ is \emph{non-degenerate with respect to $S'$} if goods in
$S'$ have non-zero marginal value:
\[
  \val{S \setminus g}{} < \val{S}{}
  \longquad \text{for each } g \in S \cap S'.
\]
 When $S' = [m]$, we will say $S$ is \emph{non-degenerate}.
The \emph{non-degenerate correspondence} for buyer $q \in N$ at price
$\bp$ is the set of demand bundles defined by
\[
\maxdem{q}(\bp) = \{S \in \dem{q}(\bp) : \val{S\setminus g}{q} < \val{S}{q}
\longquad \forall g \in S \}.
\]
\end{definition}
Minimum bundles are an example of non-degenerate bundles.
\begin{lemma}
  \label{lem:min-nondegen}
  Let $v$ be a GS valuation function.  If $S$ is a minimum demand
  bundle with respect to this valuation and some set of prices, then
  $S$ is non-degenerate.
\end{lemma}
\begin{proof}
  Fix a buyer and a price vector $\bp$ over goods and some minimum
  demand bundle $S$.  By hypothesis, the buyer's utility $\util{S}{}$
  must strictly decrease if we eliminate a good from $S$,
  i.e. $\forall g \in S$ we have
  $\util{S \setminus g}{} < \util{S}{}$.  We need to show that
  $\val{S\setminus g}{} < \val{S}{}$ for every $g \in S$ to prove that
  $S$ is non-degenerate. Because $S$ is a minimum demand bundle,
  \begin{align*}
    \util{S}{} &= \val{S}{} - p(S) > \val{S\setminus g}{} - p(S \setminus g) \\
    &> \val{S\setminus g}{} + p(g) \geq \val{S\setminus g}{}
  \end{align*}
  using non-negativity of the prices.
\end{proof}
If we only consider buyers who purchase bundles in $\maxdem{q}(\bp)$,
there is no problem with zero-price goods $g$: for buyers with non-degenerate bundles that contain $g$, that good must automatically be in any minimum bundle for those buyers.
Hence, buyers cannot be indifferent between non-degenerate bundle $B$ that contains $g$ and $B\setminus g$.
So, we only need edges in the swap graph from $\bot$ to goods with positive
price.

We are almost ready to formally connect the swap graph with over-demand, but
there is one last wrinkle. Since a buyer is assigned a bundle rather than a
single good, they may have two different swaps to the same good: perhaps $M_q
\setminus a_1 \cup b$ and $M_q \setminus a_2 \cup b$ are both demanded. We want
to count this buyer as causing demand $1$ rather than $2$, since they will select
a bundle of distinct goods.  This consideration motivates our definition of
\emph{buyer in-degree}.

\begin{definition}[Buyer in-degree]
  Consider buyers
  with GS valuations $\{v_q\}_{q \in N}$, the minimal
  Walrasian prices $\bp$ and allocation $\mu$, and the
  corresponding swap graph $G$.  The \emph{buyer in-degree} of
  node $b$ is the number of distinct buyers with an edge
  directed to $b$.
\end{definition}


To prevent buyers from piling on zero-priced goods with zero marginal value, we
will require buyers to take non-degenerate bundles. Then, the following
definition of over-demand is natural.
\begin{definition}[Non-degenerate over-demand]
The \emph{non-degenerate over-demand} $\maxod{g}$ for a good $g\in [m]$ at
Walrasian prices $\bp$ is
\[
\maxod{g}= \max\{| U^\bullet(g; \bp) | - s_g, 0 \},
\]
where
\[
 U^\bullet(g; \bp) = \{ q \in N : \exists B \in \maxdem{q}(\bp) \text{ and } g
\in  B \} .
\]
\end{definition}
Like the unit demand case, this is simply the worst-case over-demand assuming
that buyers choose an arbitrary non-degenerate bundle from their demand
correspondence, without any assumption on how they break ties.
Finally, we can relate the swap graph to over-demand.
\begin{lemma}
Let buyers have GS valuations and let $(\bp,\mu)$ be a Walrasian
equilibrium with minimal prices.  If a node $g \in [m]$ in the swap
graph $G$ has buyer in-degree at most $d$, then $\maxod{g}\leq d$.
\LABEL{lem:deg_od_gs}
\end{lemma}
\begin{proof}
  Consider any good $g$. We wish to bound the size of $U^\bullet(g; p)$, the
  set of buyers who have $g$ in a non-degenerate bundle. These buyers fall
  into two classes: buyers with $g \in M_q$, and buyers with $g \notin M_q$.
  The number of buyers of the first kind is at most the supply $s_g$, since
  $M$ allocates only fewer copies of goods than the Walrasian allocation
  $\mu$.

  To bound the number of buyers of the second kind, we will show that there is an
  edge to $g$ in the swap graph for each such buyer. Consider any buyer $q$ with
  $g \notin M_q$ and some non-degenerate bundle $B^\bullet_q \in
  \maxdem{q}(\bp)$ where $g \in B^\bullet_q$. Consider any minimum demand bundle
  $B^*_q \subseteq B^\bullet_q$. If $g \in B^*_q$, then there is an edge
  corresponding to $q$ from $M_q$ to $g$, since both $M_q$ and $B^*_q$ lie in
  $\mindem{q}(\bp)$ (which forms the set of bases for some matroid by
  \cref{lem:mindem-bases}).

  Otherwise, if $g \notin B^*_q$, \cref{lem:demand-interpol} shows that
  $B^\bullet_q(\bp) \setminus g$ is also a demand bundle:
  \[
    \val{ B^\bullet_q \setminus g}{q} - p( B^\bullet_q \setminus g)
    = \val{B^\bullet_q }{q} - p( B^\bullet_q)
  \]
  Since $B_q^\bullet$ is non-degenerate,
  \[
    \val{B^\bullet_q \setminus g}{q} < \val{B^\bullet_q }{q}
  \]
  and so $p(g) > 0$. At the same time, \cref{lem:demand-interpol} also shows
  that $B^*_q \cup g$ is a demand bundle. Combining these last two facts, there
  must be an edge corresponding to $q$ from $\bot$ to $g$ in the swap graph. So,
  the number of demanders of the second kind ($g \notin M_q$) is at most the
  buyer in-degree of $g$, which is at most $d$ by assumption.

  Now, we count up the demanders of both kinds: at most $s_g + d$ buyers demand
  $g$ in a non-degenerate bundle.  Hence, the non-degenerate over-demand is at
  most $d$.
\end{proof}

\SUBSECTION{Matroid Based Valuations}

Now that we have seen the swap graph, we move on to describe our genericity
condition on valuations.
As discussed above, it is not clear how to define genericity for general GS
valuations due to the axiomatic nature of the definition. We will define
genericity for a subclass of GS valuations studied by \citet{OL15} called
\emph{matroid based valuations} (MBV). Such valuations are defined in terms of
a simpler class of valuations, known as VIWM.

\begin{definition}[VIWM]
Let $X$ be some ground set and let $w$ be a function $w:X\rightarrow \R_{\geq 0}$. Consider a matroid
$\cM = (\cI,X)$ and define $\cM_{S} = (\cI_S,S)$ as the restricted matroid to
set $ S \subseteq X$ where $\cI_{S} = \{ T \in \cI: T\subseteq  S \}$. The
\emph{valuation induced by weighted matroid} $(\cM,w)$ is
\[
  v(S) = \max_{ T \in \cI_{ S}} \sum_{j \in  T} w(j) .
\]
\end{definition}

Some examples of valuations induced by a weighted matroid include unit
demand valuations (the matroid has independent sets that are
singletons and the empty set) and additive valuations (the matroid has
the power set $2^X$ as its independent sets).
For any matroid $\cM = (\cI, X)$, the valuations induced by $\cM$ lie in
GS \citep{FY03, Mur96,MS99}.  However, \citet{OL15} show that VIWM
is a \emph{strict} subset of GS by showing that two
operations on valuation functions preserve GS, but do not
preserve VIWM.

\begin{definition}[Merging (convolution)]
  For any $ S \subseteq X$ and valuation functions $v', v'' : 2^X \to
 \R_{\geq 0}$, the \emph{merging} of $v', v''$ is the valuation defined by

  \[
    v^{merg} (S) = \max_{(S', S'') = S} \left\{ v'(S') + v''( S'')
    \right\} .
  \]
\end{definition}

\begin{definition}[Endowment (marginal valuation)]
  Suppose $S \subseteq X$ and valuation $v' :2^{X \cup Y} \to \R_{\geq 0}$, where
  $X \cap Y = \emptyset$. Then, the \emph{endowment} of $v'$ by $Y$ is the
  valuation defined by
  \[
    v^{end}(S) = v'(S \cup Y) - v'(Y) .
  \]
\end{definition}

Accordingly, \citet{OL15} propose a new class of valuations closing
VIWM under the two operations, called MBV, and conjecture that MBV is
equal to all of GS.

\begin{definition}[Matroid based valuation]
  \emph{Matroid based valuations} (MBV) is the The smallest class of valuations
  functions containing VIWM closed under finite merges and endowments.
\end{definition}

\begin{conjecture}[\citet{OL15}]
The class of MBV is precisely the class of GS valuations: GS = MBV.
\end{conjecture}

\ifshort
\else
\SUBSECTION{Structural results for MBV}

In the case of matchings, valuations of goods are described by a single number.
When working with MBV, we will need to get a handle on the value of bundles in
terms of the weights on goods in order to reason about integer linear
combinations of weights for our genericity argument.
In this section, we show that an MBV $v$ evaluated on a non-degenerate
set $B$ can be written as the sum of weights corresponding to elements
in $B$ plus some integer linear combination of weights corresponding
to the additional endowed goods.

Note that if each buyer $q \in N$ has valuations $v_q$ in MBV, then we
can decompose $v_q$ into $\tau_q$ valuations in VIWM, say
$\left\{v_q^{(t)} : t \in [\tau_q]\right\}$ where each $v_q^{(t)}$ can be
identified with an endowed set $T_q^{(t)}$ of goods disjoint from
$[m]$, weight function $w_q^{(t)}$, and independent sets $\cI_q^{(t)}$
which make up the matroid
$\cM_q^{(t)} = \left(\cI_q^{(t)},[m] \cup T_q^{(t)} \right)$.  Note
that valuations are defined over bundles of endowed goods and the
original goods $[m]$, while the weights are defined over bundles of
endowed goods and goods $[m]$.

We will use the notation $W_q(A)$ to denote the set of weights
belonging to buyer $q$ that are evaluated on elements in the set $A$,
i.e.
$$
W_q(A) = \{w_q^{(t)}(a) : a \in A, t \in [\tau_q] \} \qquad \text{ and } \qquad W(A) = \bigcup_{q \in N} W_q(A).
$$
We also will write the weights for goods allocated to each buyer in
$\mu = (\mu_q)_{q \in N}$ as $W(\mu) = \bigcup_{q \in N} W_q(\mu_q)$.

We note that each buyer's valuations $v_q$ can be thought of as a
tree structure where each node has at most two children, and where the
$\tau_q$ leaves of the tree are VIWM
$v_q^{(t)}: 2^{[m]\cup T_q^{(t)}} \to \R_{\geq 0}$
for each
$t \in [\tau_q]$, where $T_q^{(t)}$ is the union of all goods in
endowment operations along the path from leaf $t$ to the root of the
tree; such goods are distinct from the goods $[m]$.
Internal nodes with two children represent the merge operation
applied to the two children. Each internal node
with one child represents an endowment operation applied to the child
valuation function.
\cref{fig:mbv_tree} gives an example tree of a valuation $v$ in MBV
class.

\input{vwmfig}

We let the \emph{depth} of an MBV $v$ be the depth of the
corresponding tree, i.e., the length of the longest path from the root to a
leaf.

Now, we want to show that if we have a non-degenerate set $S$ with
respect to some valuation function $v$, then $v(S)$ can be written as
an integer linear combination of weights from each non-endowed good in
$S$ plus some weights from endowed goods.  We will denote by $\dom(v)$
the domain of valuation $v$ which in particular always contains $\cG$, i.e. bundles of goods $[m]$.

We first need a useful fact: if $S$ is non-degenerate for an MBV, then
every subset $S' \subseteq S$ is also non-degenerate.

\begin{lemma}
If $v$ is MBV and $S$ is non-degenerate then for every $S'
\subseteq S$, $S'$ is also non-degenerate.
\LABEL{lem:type_min}
\end{lemma}
\begin{proof}
  We know that $MBV$ is contained in GS and every GS
  valuation is submodular.  Hence, if $S$ is non-degenerate but we
  have a set $S' \subseteq S$ and a good $g \in S' \cap [m]$ such that
  $v(S' \setminus g) = v(S')$ then we can form the string of
  inequalities, where the second follows from $v$ being submodular
$$
0<v(S) - v(S\setminus g) \leq v(S') - v(S' \setminus g) = 0.
$$
This leads to a contradiction and proves the lemma.
\end{proof}

\begin{remark}\rynote{Modified}
  Below, we will frequently manipulate formal integer linear combinations without
  needing to know precisely what the coefficients are---the important thing is
  that we are
  working with \emph{some} integer linear combination of the given weights.
  In these cases, we will abuse notation
  and use variables $\gamma$ with subscripts to
  represent arbitrary integer values, where different occurrences of
  the same variable may not represent the same value, e.g.
  $v(S) + \sum_{g \in [m]} \gamma_g w(g) = \sum_{g \in [m]} \gamma_g
  w(g)$
  does not mean that $\gamma_g$ is the same integer on both sides of
  the equation.
  \end{remark}

Now, we can connect the value of non-degenerate bundles to the weights at the
leaf VIWM.

\begin{lemma}
Let $v: \dom(v) \to \R_{\geq 0}$ be MBV and composed of $\tau$ valuations in VIWM with
endowed sets $T = \bigcup_{t \in [\tau] }T^{(t)} $ and weight functions $\{w^{(t)}: t \in [\tau] \}$.  If $S \in \dom(v)$
is non-degenerate with respect to $v$, then for each $g \in S\cap [m]$ there exists
$t_g \in [\tau]$ such that 
\begin{equation}
v(S) = \sum_{g \in S\cap [m]} w^{(t_g)}(g) + \sum_{w \in W( T)} \gamma_w w .
\LABEL{eq:min_weights}
\end{equation}

\LABEL{lem:min_weights}
\end{lemma}
\ar{This proof is really hard to follow. Needs a rewrite.}
\jh{Better or worse?}
\begin{proof}
  The proof is by induction on the depth of the valuation $v$.
  In the base case $v$ is a VIWM with matroid
  $\cM = (\cI,[m]\cup T)$ and weight function $w$.
  By definition of VIWM, we know
  \[
    v(S) = \max_{S' \in \cI_S} \sum_{g \in S'}w(g) .
  \]
  A basic fact of weighted matroids shows that the maximum can be constructed
  using the greedy algorithm \citep{Matroid}.
  That is, we arrange all elements $g
  \in S$ in decreasing order of weight $w$ and start with a set $S' =
  \emptyset$, greedily add elements $g$ to $S'$
  as long as $S' \in \cI$. We stop when we consider all
  elements or when we find an element with negative weight.
  Suppose we have an ordering $\sigma$ on $S$ in decreasing weight such that the
  greedy algorithm produces an independent set $S^*$ achieving the maximum.

  Since $S$ is non-degenerate, for each $g \in S\cap [m]$ we have $v(S\setminus
  g) < v(S)$. We claim that $g \in S^*$. If not, suppose we drop good $g$ from
  the ordering $\sigma$. We get an ordering $\sigma_{-g}$ that orders all goods
  in $S \setminus g$ in decreasing weight. Then, the greedy algorithm ran on
  orderings $\sigma$ and $\sigma_{-g}$ will produce the same independent set
  $S^*$, and hence the same value. But we know that $v(S \setminus g) < v(S)$
  strictly, so $g$ must be in the independent set $S^*$.

  Since we fixed $S^*$ before considering $g \in S \cap [m]$, the same argument
  shows that the whole set $S\cap [m]$ must be contained in $S^*$, and so
 \begin{align*}
  v(S) & = \sum_{g \in S^*} w(g) = \sum_{g \in S \cap [m]} w(g) + \sum_{h \in S^*\setminus (S \cap [m]) }w(h)
  =  \sum_{g \in S \cap [m]} w(g) + \sum_{h \in S^*\cap T }w(h) .
  \end{align*}
    This gives our base case.

We now assume that our inductive hypothesis (\cref{eq:min_weights}) holds for
any MBV $v$ with depth at most $k-1$ and any non-degenerate bundle $S$, and we prove
\cref{eq:min_weights} for a MBV
$v$ with depth $k > 1$ and any non-degenerate bundle $S$. Since $k > 1$, $v$ is
the result of either a merge or endowment operation.
In the first case, let $v$ be the result of an endowment operation from valuation
$v':2^{[m] \cup T} \to \R_{\geq 0}$ with endowment set $J \subseteq T$.  Let $S \in
\dom(v) \subseteq [m] \cup T\setminus J$ be a non-degenerate set, so that $v(S
\backslash g) < v(S)$ for every $g \in S\cap [m]$.  By definition of endowment,
$$
v(S) = v'(S\cup J) - v'(J).
$$
We need only show that $S\cup J$ is non-degenerate for valuation $v'$ to apply the
induction hypothesis, because $v'$ is a MBV and has depth $k-1$.
For contradiction, suppose that $S \cup J$ is degenerate: for some element $g^*
\in S\cap [m]$ we have
$v'(S\cup J ) = v'(S\cup J \setminus g^*)$. Then,
\[
v(S) = v'(S\cup J) - v'(J) = v'(S\cup J \setminus g^*) - v'(J) = v(S\setminus
g^*)
\]
is a contradiction, since $S$ is non-degenerate for valuation $v$. Hence, $S
\cup J$ is non-degenerate and by the
inductive hypothesis, for each $g \in S \cap [m]$ we have $t_g \in [\tau]$ such that
$$
v'(S \cup J) - v'(J) = \sum_{g \in S\cap [m]} w^{(t_g)}(g) + \sum_{w \in W(T)} \gamma_w w -
v'(J) = \sum_{g \in S \cap [m]} w^{(t_g)}(g) + \sum_{w \in W(T)} \gamma_w w. 
$$

If $v$ is not the result of an endowment operation, it must be the result of a
merge operation. In this case, we let $v$ be the merge of two valuations $v^1$ and $v^2$
in MBV, both with depth at most $k-1$.  We are given
$S\in \dom(v)$ non-degenerate  for $v$, which gives us
$$
v(S) = \max_{(S_1,S_2) = S} \{v^1(S_1) + v^2(S_2) \} .
$$
Suppose that the maximum is achieved at a partition $(S_1^*,S_2^*)$ of $S$. We
claim that $S_1^*, S_2^*$ are non-degenerate for $v^1, v^2$ respectively. For
contradiction, suppose that $S_1^*$ is degenerate for $v^1$.  Then, there is a
$g \in S_1^*\cap [m]$ such that
$$
v(S) = v^1(S^*_1) + v^2(S^*_2) = v^1(S^*_1\setminus g) + v^2(S^*_2) \leq
v(S\setminus g) ,
$$
but this is a contradiction since  $S$ is non-degenerate for $v$.  Thus, $S^*_1$
is non-degenerate for $v^1$; by symmetry also $S^*_2$ is non-degenerate for
$v^2$. Therefore, we can apply the induction hypothesis on the two valuations
$v^1$ and $v^2$: for each $g \in S \cap [m]$, there is a $t_g \in [\tau]$ such that
$$
v^1(S^*_1) + v^2(S^*_2) = \sum_{g \in S_1^* \cap [m]} w^{(t_g)}(g) + \sum_{g \in S_2^*\cap [m]}
w^{(t_g)}(g) + \sum_{w \in W(T) }\gamma_w w. 
$$
Since $S_1^*$ and $S_2^*$ form a partition of $S$, we can combine the two sums:
$$
v(S) = \sum_{g \in S \cap [m]} w^{(t_g)}(g) +\sum_{w \in W(T) } \gamma_w w ,
$$
concluding the proof.
\end{proof}
\begin{corollary}
Let $v: \dom(v) \to\R_{\geq 0}$ a MBV composed of $\tau$ valuations in VIWM with endowed
sets $T = \bigcup_{t \in [\tau]}T^{(t)}$ and weight functions $ \left\{ w^{(t)} : t \in
  [\tau] \right\}$.  If $S\in \dom(v)$ and there exists a $g \in S\cap [m]$ such
that $v(S\setminus g) < v(S)$, then there exists $t_g \in [\tau]$ such that
\begin{equation}
v(S) = w^{(t_g)}(g) + \sum_{w \in W(S\cap [m] \setminus g)} \gamma_w w + \sum_{w
  \in W(T)} \gamma_w w .
\LABEL{eq:val_weights}
\end{equation}
\LABEL{cor:val_weights}
\end{corollary}
\begin{proof}
We use the same proof method as in \cref{lem:min_weights}.
\end{proof}

\fi

\SUBSECTION{Generic MBV}

Now that we have defined MBV, defining genericity is straightforward: we simply
require genericity on the set of weights.  Each valuation in MBV can be
decomposed as a set of weights over the goods in $[m]$ as well as the goods that
have been \emph{endowed}.  
\ifshort
We will use the notation $W_q(A_q)$ to denote the set of weights belonging to
buyer $q$ that are evaluated on elements in the set $A_q$ and $W(A)$ to denote
the set of weights for all buyers, i.e.
$$
  W_q(A_q) = \{w_q^{(t)}(a) : a \in A, t \in [\tau_q] \} ,
$$
and
$$
  W(A) = \bigcup_{q \in N} W_q(A_q) .
$$
\fi

\begin{definition}[Generic matroid based valuations]
  \LABEL{def:GMBV}
  We say that a collection of valuations $\{ v_q: \cG \to [0,\maxval] : q \in N\}$
  are \emph{generic matroid based valuations} (GMBV) if each $v_q$ is a matroid
  based valuation and the weights across buyers over all goods $[m]$ and
  endowments $T_q$ introduced in endowment operations are linearly independent
  over the integers. We will write the set $A_q = [m] \cup T_q$.  Formally, the
  weights satisfy
  $$
  \sum_{w \in W(A)} \gamma_{w} w = 0
  $$
  for $\gamma_{w} \in \Z$ if and only if the coefficients are all zero.
\end{definition}

\ifshort\else
By \cref{lem:min_weights}, if a non-degenerate bundle $S$ for a GMBV $v$ then we
know that $v(S)$ is a nontrivial integer linear combination of the weights.
This property will be continually used in the results that follow in bounding
buyer-in-degree in the swap graph.

\SUBSECTION{Properties of swap graph with GMBV}
To bound over-demand, our plan is the same as the unit demand case: 1) show that
any node with no incoming edges must have price zero, as in \cref{lem:pzero} 2)
show that the swap graph is acyclic, as in \cref{lem:acyclic_match}, 3) show
that we can write the price of a good $g$ as an integer linear combination of
the weights for the goods prior to $g$ in a topological sort of the swap graph,
as in \cref{lem:lincomb}.  Once we show these three corresponding results, we
can bound \emph{buyer-in-degree} in the swap graph and hence over-demand, due to
\cref{lem:deg_od_gs}.

As mentioned before, the arguments are more involved than the corresponding
arguments for the matching case because we need to work around situations where
weights for goods may not end up in the final bundle valuation. A simpler
version of this problem can be seen in the matchings section, where proofs
frequently treat the null node $\bot$ as a special case---this is because we
treat buyers as \emph{all} having value $0$ for the null node, so it is clearly
non-generic. We require more delicate reasoning for working with GMBV, since
there are more opportunities where buyers may have marginal value $0$ for some
good.

We begin by considering the source nodes in the swap graph. As we had in the
unit demand setting in \cref{lem:pzero}, any node in our swap graph with no
incoming edge has price zero. In fact, this lemma does not require valuations to
be in GMBV; GS valuations suffice. We point out that while the analogous theorem
in the unit demand case is an immediate consequence of a fact about the minimal
Walrasian equilibrium of GS valuations, we need a bit of care to handle the more
complex swap graph in the present case.

\begin{lemma}
If agents have valuations in GS and $(\bp, \mu)$ is a Walrasian equilibrium with
minimal Walrasian prices then any non-null node in the swap graph $G(\bp,\mu,M)$
with in-degree zero must have price zero, for any choice of minimum demand
bundles $M = (M_q)_{q\in N}$ contained in $\mu$.
\LABEL{lem:pzero_gs}
\end{lemma}
\begin{proof}
Let the node corresponding to $g$ have no incoming edges. If $p(g) >0$ we can
decrease the price $p(g)$ by a small positive value
$$
\epsilon_g = \frac{1}{2}\times \min\left\{ p(g), \min_{\substack{q \in N \\ B,B' \subseteq [m]} } \left\{ \util{B}{q} - \util{B'}{q}: \util{B}{q} - \util{B'}{q}>0 \right\} \right\}
$$
while not modifying the prices $\bp$ on other goods; let the new price vector be $\bp'$.

We claim $\bp'$ along with the original allocation $\mu$ forms a Walrasian
equilibrium. Evidently $\mu$ continues to clear the market. We need to show that
each buyer continues to prefer their allocation under $\mu$. Buyers who are
assigned $g$ under $\mu$ continue to prefer their allocation, since their
utility has only increased as we decrease the price of $g$.

Now, consider buyers $q$ that are not allocated $g$ in $\mu$. By our choice of
$\epsilon_g$, the only bundles $B$ that $q$ might strictly prefer to $M_q$ at
the new prices $\bp'$ contain $g$ and satisfy
\[
  \util{B}{q} = \util{M_q}{q} .
\]
By $M_q$ being a minimum bundle, we must have $|B| \geq |M_q|$. If $|B| = |M_q|$ then $B$ is
also a minimum bundle at price $p$, and since such bundles form a Matroid Basis, there
must be some $h \in M_q$ such that $M_q \cup g \setminus h$ is also minimal demand under
prices $\bp$. By construction of the swap graph, $g$ would have an incoming
edge---impossible.

If $|B| > |M_q|$,
then take any minimum demand bundle $B' \subsetneq B$ at price $\bp$. If $g \in B'$ we again
get an in-edge to $g$ via a swap from $M_q$, while if $g \notin B'$ then $B'
\cup g$ would be a demand bundle by \cref{lem:demand-interpol} giving an edge
from $\bot$ to $g$; in both cases, contradiction.  So, buyers who are not
allocated $g$ also continue to prefer their allocation under $\mu$.

Thus, the new prices $\bp'$ and the allocation $\mu$ form a Walrasian
equilibrium with $\bp' \leq \bp$. This is a contradiction: $\bp$ are minimal
Walrasian prices by assumption.
\end{proof}

Next, we show that if we have valuations that are GMBV, then the corresponding
swap graph is acyclic.


\begin{lemma}
Let $(\bp,\mu)$ be a Walrasian equilibrium for buyers with valuations $\{
v_q: q \in N\}$ that are GMBV.  Then for any choice of minimum demand bundles $M =
(M_q: q \in N)$ where $M_q \subseteq \mu_q$, the corresponding swap graph $G(\bp,\mu,M)$ is acyclic.
\LABEL{lem:gmbv_acyclic}
\end{lemma}
\begin{proof}
Suppose $\cC = \{(a_0,a_1),
(a_1,a_2), \ldots,(a_{k-2},a_{k-1}), (a_{k-1},a_0) \}$ are edges
in the swap graph $G$ forming a simple cycle. None of these goods can be the
null node $\bot$  because $\bot$ has no incoming edges. Let buyer $q_i \in N$ have edge $(a_i,a_{i+1}) \in \cC$
(addition in subscript is mod $k$) so that $a_i \in M_{q_i}$; note that $q_i \neq q_{i + 1}$
for all $i$, since $a_{i+1} \in
\mu_{q_{i+1}}$ and $a_{i+1} \notin \mu_{q_i}$ by construction of the swap graph.  Because edges in the swap graph indicate indifferences, we have the following set of equalities
$$
\util{M_{q_i}}{q_i} = \util{M_{q_i}'}{q_i}
\quad \text{so that} \quad
\val{M_{q_i}}{q_i} - p(a_i) = \val{M_{q_i}'}{q_i}  - p(a_{i+1}) \qquad i = 0,\ldots,k-1
$$
where $M_{q_i}' = M_{q_i} \cup a_{i+1} \setminus a_i$.   Summing this equation over all the buyers in the cycle, the prices cancel and
we have
\begin{equation}
\sum_{i = 0}^{k-1} \val{M_{q_i}}{q_i} =\sum_{i = 0}^{k-1} \val{M_{q_i}'}{q_i} .
\LABEL{eq:sum_vals}
\end{equation}

Let $a$ be a good in $\cC$ appearing at position $i^*$ (since the cycle is simple,
$i^*$ is unique) and $q$ be the corresponding buyer at this position.  Let $\ell$ be the number of weights corresponding to buyer $q$ for good $a$ on
the LHS of \cref{eq:sum_vals}.  
  Since $a \in M_q$ by construction of the swap graph and $M_q$ is non-degenerate,
we have one weight from $q$ for good $a$ each time buyer $q$ shows up in cycle
$\cC$ by \cref{lem:min_weights}, that is
\begin{align*}
\sum_{i = 0}^{k-1} &\val{M_{q_i}}{q_i} = \sum_{i: q_i = q} v_q(M_q) + \sum_{i : q_i \neq q} v_{q_i}(M_{q_i}) \\
& = \ell \cdot \sum_{g \in M_q} w_q^{(t_g)}(g) + \sum_{i : q_i \neq q} v_{q_i}(M_{q_i}) + \sum_{w \in W(T)} \gamma_w w \\
& = \ell \cdot w_q^{(t_a)} (a) + \ell \cdot \underbrace{\sum_{g \in M_q
    \setminus a} w_q^{(t_g)}(g)}_{\text{No weights for good $a$}}  +
\underbrace{\sum_{i : q_i \neq q} v_{q_i}(M_{q_i})}_{\text{No weights for buyer
    $q$}}+ \sum_{w \in W(T)} \gamma_w w .
\end{align*}

On the RHS of \cref{eq:sum_vals}, we again have a weight from buyer $q$ for
good $a$ each time $q$ appears in the cycle, except for position $i^*$ where $a
\notin M_q \cup a_{i^*+1} \setminus a$.  Thus, there are exactly $\ell - 1$ weights corresponding to $q$
for good $a$ on the RHS of the equation:
\begin{align*}
\sum_{i = 0}^{k-1} &\val{M'_{q_i}}{q_i} = \sum_{i: q_i= q} \val{M_{q} \cup a_{i+1}\setminus a_i}{q} + \sum_{i: q_i\neq q} \val{M'_{q_i}}{q} \\
& = v_q(M_q \cup a_{i^*+1} \setminus a) + \sum_{\substack{i: q_i= q \\i \neq i^*}} \val{M_{q} \cup a_{i+1}\setminus a_i}{q}+ \sum_{i: q_i\neq q} \val{M'_{q_i}}{q} + \sum_{w \in W(T)} \gamma_w w\\
& =  \underbrace{v_q(M_q \cup a_{i^*+1} \setminus a) }_{\text{No weight for
    $a$}} +  \underbrace{ \sum_{\substack{i: q_i= q \\i \neq i^*}} \sum_{g \in
    M_q \cup a_{i+1} \setminus a_i} w_q^{(t'_g)}(g) }_{\ell-1 \text{ weights for
    good $a$ belonging to $q$} }+ \underbrace{ \sum_{i: q_i\neq q}
  \val{M'_{q_i}}{q}}_{\text{No weights for buyer $q$}} + \sum_{w \in W(T)}
\gamma_w w .
\end{align*}
By genericity \cref{eq:sum_vals} cannot hold, so the swap graph must be acyclic.
\end{proof}

Once we have acylicity, we can use the swap graph to read off the form of the
Walrasian prices.  Just like in the unit demand case, we can trace back from $g$
to a source node.  The price of the source node is zero, and each hop along the
swap graph adds a difference of valuations to the price until we finally arrive
at good $g$.


Because the swap graph $G(\bp,\mu,M)$ is acyclic for any $M$, we
can choose a partial order on the nodes so that all edges go from nodes earlier in
the ordering to nodes later in the ordering. For the next result we fix a
good $g$ in the swap graph and a simple path $g_1 \to g_2 \to \cdots \to g_k = g$ originating from source node $g_1$ with no incoming edges.  We label the buyers such that the edge from $g_{i-1}$ to $g_i$ belongs to $q_i$.
Note that unlike the unit demand case, this path may visit several
goods held by the same buyer (we could have $q_i = q_j$ and $i \neq j$). This
feature complicates the genericity analysis, since we need to make sure that a
single buyer's weights do not cancel themselves out.

\begin{lemma}
  Given the swap graph $G$ and the simple path defined above, if $g_1 = \bot$ then there exists $B \in \dem{q_2}(\bp)$ which contains $g_2$
  and $B\setminus g_2 \in \dem{q_2}^*(\bp)$ such that $p(g_2) =
  \val{B}{q_2} -\val{B\setminus g_2}{q_2}$ and if $k>2$ then there exists $t_j \in [\tau_{q_j}]$ for $j = 2,\dots, k$ such that
\begin{align}
p(g_k) = \weight{g_k}{q_k}{t_k} + \val{B}{q_2} - \val{B\setminus g_2}{q_2}+
\sum_{2<j < k} \gamma_j \weight{g_j}{q_j}{t_j} + &\sum_{w \in W(\mu)} \gamma_w w
+  \sum_{w \in W(T)} \gamma_w w .
\LABEL{eq:gs_prices_bot}
\end{align}
Otherwise, if $g_1 \neq \bot$ then there exists $t_j \in [\tau_{q_j}]$ for $j = 1,\dots, k$ such that
\begin{align}
p(g_k) = \weight{g_k}{q_k}{t_k} + \sum_{1<j < k} \gamma_j \weight{g_j}{q_j}{t_j}
+ &\sum_{w \in W(\mu)} \gamma_w w +  \sum_{w \in W(T)} \gamma_w w .
\LABEL{eq:gs_prices_not}
\end{align}
\LABEL{lem:gs_price_induct}
\end{lemma}

\begin{proof}
By \cref{lem:gmbv_acyclic}, the swap graph is acyclic and defines a partial order
on the goods.  Note that because buyer $q_k$ has an edge leading to $g_k$,
good $g_k$ cannot be in $\mu_{q_k}$.  We will prove the lemma by induction on
$k$, the number of goods in the path.  For the base case, we consider $k=2$.  We
will first consider the case where the source node $g_1 = \bot$.  By
construction of the swap graph we know that there exists $B \in \dem{q_2}(\bp)$
such that $B \setminus g_2 \in \mindem{q_2}(\bp)$, thus we have
$$
u_{q_2}(B;\bp) = u_{q_2}(B\setminus g_2;\bp)
\quad \text{and} \quad
p(g_2) = \val{B}{q_2} - \val{B\setminus g_2}{q_2}.
$$
Now, if $g_1 \neq \bot$ and $k=2$ then we have
$$
u_{q_2}(M_{q_2};\bp) = u_{q_2}(M_{q_2} \cup g_2 \setminus g_1;\bp)
\quad \text{and} \quad
p(g_2) = \val{M_{q_2} \cup g_2 \setminus g_1}{q_2} -  \val{M_{q_2}}{q_2} + p(g_1).
$$
Note that because $g_1$ has in-degree 0, we know from \cref{lem:pzero_gs} that
$p(g_1) = 0$.  Further, because $M_{q_2}$ and $M_{q_2} \cup g_2 \setminus g_1$
are minimum demand bundles for $q_2$ we can apply \cref{lem:min_weights}: for each
$g \in M_{q_2}$ there exists $t_g \in [\tau_{q_2}]$, and for each $g \in
M_{q_2}\cup g_2 \setminus g_1$ there exists $t_g' \in [\tau_{q_2}]$ such that
\begin{align*}
p(g_2) & =  \sum_{g \in M_{q_2}\cup g_2 \setminus g_1} \weight{g}{q_2}{t_g'} -
\sum_{g \in M_{q_2}} \weight{g}{q_2}{t_g} + \sum_{w \in
  W(T)} \gamma_w w \\
&= \weight{g_2}{q_2}{t'_{g_2}}  + \sum_{w \in W_{q_2}(\mu_{q_2} )} \gamma_w w +
\sum_{w \in W(T)} \gamma_w w .
\end{align*}


We now consider the inductive case when the length of the path is $k > 2$. Let us
assume that $g_1 = \bot$ and \cref{eq:gs_prices_bot} holds for good $g_{k-1}$.
By construction of the swap graph, we know that buyer $q_{k}$ is indifferent to
receiving $M_{q_k}$ or $M_{q_k}' = M_{q_k}\cup g_k \setminus g_{k-1}$ where $g_k
\notin \mu_{q_k}$.  We then have,
$$
p(g_k) = \val{M_{q_k}'}{q_k} - \val{M_{q_k}}{q_k} + p(g_{k-1}).
$$

We then apply the inductive hypothesis on $p(g_{k-1})$ for $g_1 = \bot$.  We will use
\cref{lem:min_weights} to conclude that for each $g \in M_{q_k}$ there exists
$t_g \in [\tau_{q_k}]$, for each $g \in M_{q_k}' $ there is a $t_g' \in
[\tau_{q_k}]$, for each $g < k$ there is a $t_g'' \in [\tau_{q_g}]$, and for
some $B \in \dem{q_2}(\bp)$ where $B \setminus g_2 \in \mindem{q_2}(\bp)$ we have
\begin{align*}
p(g_k) & = \sum_{g \in M_{q_k}'} \weight{g}{q_k}{t_g'} - \sum_{g \in M_{q_k}
}\weight{g}{q_k}{t_g} + \weight{g_{k - 1}}{q_{k - 1}}{t_{k - 1}''} + \val{B}{q_2} -  \val{B\setminus g_2}{q_2} \\
& \quad + \sum_{2<j < k-1} \gamma_j \weight{g_j}{q_j}{t_i''} + \sum_{w \in W(\mu)} \gamma_w w + \sum_{w \in W(T)} \gamma_w w \\
& = \weight{g_k}{q_k}{t'_k} + \val{B}{q_2} -  \val{B\setminus g_2}{q_2} +
\sum_{2<j < k} \gamma_j \weight{g_j}{q_j}{t_j''} + \sum_{w \in W(\mu)} \gamma_w w +
\sum_{w \in W(T)} \gamma_w w .
\end{align*}

Proving \cref{eq:gs_prices_not} for $g_1 \neq \bot$ follows similarly, where there exists $t_j \in [\tau_{q_j}]$ for each $j= 2, \cdots, k$ such that
\begin{align*}
p(g_k) & = \val{M_{q_k}'}{q_k} - \val{M_{q_k}}{q_k} + p(g_{k-1}) = \val{M_{q_k}'}{q_k} - \val{M_{q_k}}{q_k}  + w_{q_{k-1}}^{(t_{k-1})}(g_{k-1}) \\
&  \qquad \qquad + \sum_{1<j<k-1} \gamma_j w_{q_j}^{(t_j)}(g_j) + \sum_{w \in W(\mu)} \gamma_w w + \sum_{w \in W(T)} \gamma_w w \\
& = w_{q_k}^{(t_k)}(g) +  \sum_{1<j<k} \gamma_j w_{q_j}^{(t_j)}(g_j)+ \sum_{w
  \in W(\mu)} \gamma_w w + \sum_{w \in W(T)} \gamma_w w .
\end{align*}
\end{proof}

By \cref{lem:pzero_gs}, goods with zero in-degree must have price $0$.
Also, goods with positive in-degree must have strictly positive price under
genericity. (This fact is also true in the unit demand case, but we did not need
it there.)
\begin{corollary}
Let $(\bp,\mu)$ be a Walrasian equilibrium and $G(\bp,\mu, M)$ be the corresponding swap graph for any $M$.  Any good with positive in-degree in $G$ will have positive price.
\LABEL{cor:pos_indeg}
\end{corollary}
\begin{proof}
  Suppose $g$ has positive in-degree. Let $g_1\to g_2 \to \dots \to g_k = g$ be a simple
  path originating from source node $g_1$ to $g$ with the incoming edge to $g_i$
  corresponding to buyer $q_i$.  If $g_1 \neq \bot$, then
  \cref{eq:gs_prices_not} implies $p(g)>0$ because no terms could cancel with
  $\weight{g_k}{q_k}{t_k}$ by genericity.

  Otherwise, $g_1 = \bot$. If the path is of length $2$ ($k = 2$), $g$ has an
  edge from the null good and must have positive price by construction of the
  swap graph. So, suppose $k > 2$.  By construction of the swap graph, we know
  that there exists $Q_2 \in \cD_{q_2}(\bp)$ such that if we eliminate $g_2$
  from it, then we will have a minimum demand bundle, i.e. $Q_2 \setminus q_2
  \in \cD^*_{q_2}(\bp)$.  From \cref{cor:val_weights}, we can write the
  valuation $v_{q_2}(Q_2)$ in terms of the weights corresponding to elements
  inside of $Q_2$ that belong to $q_2$ and an integer linear combination of the
  weights for the endowed goods $T$, i.e.
$$
\val{Q_2}{q_2} = \weight{g_2}{q_2}{t_2} + \sum_{w \in W_{q_2}(Q_2\setminus g_2)}
\gamma_w w + \sum_{w \in W(T)} \gamma_w w .
$$
Further, because $Q_2
\setminus g_2 \in \mindem{q_2}(\bp)$, we know from \cref{lem:min_weights} that we can write $v_{q_2}(Q_2\setminus g_2)$ as a sum of weights corresponding to elements in $Q_2\setminus g_2$ plus an integer linear combination of weights from the endowed goods:
$$
\val{Q_2\setminus g_2}{q_2} =  \sum_{a \in Q_2 \setminus g_2 }
\weight{a}{q_2}{t_a} + \sum_{w \in W(T)} \gamma_w w .
$$
We then use \cref{eq:gs_prices_bot} to write $p(g)$ as
\begin{align*}
p(g)& = \weight{g}{q}{t} + \val{Q_2}{q_2} - \val{Q_2 \setminus g_2}{q_2} + \sum_{2<j<k} \gamma_j \weight{g_j}{q_j}{t_j} + \sum_{w \in W(\mu)} \gamma_w w + \sum_{w \in W(T)} \gamma_w w \\
& = \weight{g}{q}{t} + \weight{g_2}{q_2}{t_2} + \sum_{w \in W_{q_2}(Q_2\setminus
  g_2)} \gamma_w w + \sum_{2<j<k} \gamma_j \weight{g_j}{q_j}{t_j} + \sum_{w \in
  W(\mu)} \gamma_w w + \sum_{w \in W(T)} \gamma_w w .
\end{align*}
Recall that $g \neq g_2, \dots, g_{k-1}$.  If $q \neq q_2$, there would be the
weight $\weight{g}{q}{t}$ left that would not be able to cancel with any other
term, due to genericity.  Thus, we consider $q = q_2$.  It may be the case that
$Q_2$ contains $g$ and some weight from $W_{q_2}(Q_2\setminus g_2)$ cancels with
$\weight{g}{q}{t}$.  However, the weight $\weight{g_2}{q_2}{t_2}$ cannot cancel
with any other term, due to genericity and all goods along the path are distinct.  Thus, in any case there is some weight
that remains and so $p(g) \neq 0$ by genericity, so $p(g) > 0$.
\end{proof}
\fi

\SUBSECTION{Bounding the Buyer In-Degree}

\ifshort
To conclude, we show several properties of the swap graph. The arguments are
rather involved (see the extended version for details), but the overall plan
follows the unit demand case. Finally, we can bound the buyer in-degree and hence
over-demand.
\else
We are now ready to prove our final property of the swap graph, by bounding the
buyer in-degree under genericity. This result
corresponds to \cref{lem:indeg1} from the unit demand case, and crucially uses
the fact that the null node only has edges to goods with strictly positive
price; if we allowed edges from the null node to price zero goods, every bidder
would have such an edge and bidder in-degree could be as large as $n$, rather
than bounded by $1$.
\fi
\begin{theorem}
For GMBV buyers and Walrasian equilibrium $(\bp,\mu)$ with minimal Walrasian
prices, each node in the swap graph $G$ has buyer in-degree at most $1$.
\LABEL{lem:swap1_gs}
\end{theorem}
\ifshort\else
\begin{proof}
We will use \cref{lem:gs_price_induct} to prove this and so we relabel the nodes
in the swap graph $G(\bp,\mu,M)$ using the usual partial ordering---all edges go
from a lower indexed node to a higher indexed node.  Consider a node $g$ with
buyer in-degree more than 1.  We know that its price $p(g)$ satisfies
\cref{eq:gs_prices_bot} or \cref{eq:gs_prices_not}, depending on whether there
is a path originating from a source node $g_1$ to $g$ in the swap graph with
$g_1 = \bot$ or $g_1  \neq \bot$, respectively.  We will write $T$ as the set of
all endowed goods for all agents.

Suppose good $g$ has at least two incoming edges where two edges belong to
different buyers $q$ and $r$.
We consider two paths to $g$: $g_1, g_2, \dots, g_k  = g$ and $h_1, h_2,
\dots, h_\ell = g$, where $g_1$ and $h_1$ are source nodes, the edge to $g_i$
belongs to $q_i$ and the edge to $h_i$ belongs to $r_i$ with $q_k = q$ and
$r_\ell = r$.  There are three cases:
\begin{enumerate}
\item[1)] Both paths originate at a source node that is non-null, i.e. $g_1, h_1 \neq \bot$.
\item[2)] Exactly one of the source nodes along a path is from the null node,
  i.e. $g_1 = \bot, h_1 \neq \bot$.
\item[3)] The source nodes along both paths are from the null node, i.e. $g_1, h_1 = \bot$.
\end{enumerate}


For case 1), we have two paths to $g$ and by \cref{lem:gs_price_induct} we
can write $p(g)$ in two different ways. For each $g_j$ there
exists $t_j \in [\tau_{q_j}]$ for $j = 1,\dots, k$ as well as for each $h_j$
there exists $t_j' \in [\tau_{r_j}]$ for $j = 1,\dots, \ell$ such that
\begin{align*}
p(g) &= \weight{g}{q}{t_k } + \sum_{1<j < k} \gamma_j \weight{g_j}{q_j}{t_j} + \sum_{w \in
  W(\mu)} \gamma_w w +  \sum_{w \in W(T)} \gamma_w w \tag{LHS-1} \\
p(g) &= \weight{h}{r}{t_\ell' } + \sum_{1<j < \ell} \gamma_j \weight{h_j}{r_j}{t_j'} + \sum_{w \in
  W(\mu)} \gamma_w w +  \sum_{w \in W(T)} \gamma_w w . \tag{RHS-1}
\end{align*}
Consider the
weight $\weight{g_k}{q}{t_k}$ on LHS-1.  There is no other term in LHS-1 that
can cancel with it because the weights are generic and all goods $g_j \neq g$
for $j <k$.  Further, good $g$ is not in buyer $q$'s allocation $\mu_q$, so
$\weight{g}{q}{t_k} \notin W(\mu)$ and $g$ is not an endowed
good so $\weight{g}{q}{t_k} \notin W(T)$.   Also, no term in RHS-1
can cancel with $\weight{g}{q}{ t_k}$ because $r \neq q$. This gives a
non-trivial linear combination contradicting genericity, completing case
1).

For case 2), the two paths that terminate at $g$ in the swap graph have two
different types of source nodes, say $g_1 = \bot$ and
$h_1\neq \bot$.
%
When $k>2$, we have
\begin{align*}
p(g) &= \weight{g}{q}{t_k } + \val{Q_2 \cup g_2}{q_2} - \val{Q_2}{q_2}
      + \sum_{2<j < k} \gamma_j \weight{g_j}{q_j}{t_j} + \sum_{w \in W(\mu)} \gamma_w w
      +  \sum_{w \in W(T)} \gamma_w w
      \tag{LHS-2} \\
p(g) &= \weight{g}{r}{ t'_\ell} + \sum_{1<j < \ell} \gamma_j \weight{h_j}{r_j}{t_j'}
      + \sum_{w \in W(\mu)} \gamma_w w +  \sum_{w \in W(T)} \gamma_w w ,
      \tag{RHS-2}
\end{align*}
Note that if $g \notin Q_2$ then the term $\weight{g}{q}{t_k}$ for $q \neq r$
could not cancel with any term in LHS-2 due to genericity as we have argued
before.  Further, this weight would also not cancel with any term in RHS-2
because $q \neq r$ and the weights are generic.

We now consider the case that $g \in Q_2$.  Note that by monotonicity we have
$v_{q_2}(Q_2 \cup g_2) \geq v_{q_2}(Q_2 \cup g_2 \setminus g)$.  If this were an
equality we would have the following relation
$$
\util{Q_2 \cup g_2}{q_2} = \util{Q_2 \cup g_2\setminus g}{q_2} - p(g) <  \util{Q_2 \cup g_2\setminus g}{q_2}
$$
since $g$ has positive price from \cref{cor:pos_indeg}.  This is a contradiction
because $Q_2 \cup g_2$ is a demand bundle of $q_2$.  Hence we must have
$v_{q_2}(Q_2 \cup g_2) > v_{q_2}(Q_2 \cup g_2 \setminus g)$ and we can apply
\cref{cor:val_weights} to conclude that for some $t^* \in [\tau_{q_2}]$,
$$
\val{Q_2\cup g_2}{q_2} = \weight{g}{q_2}{t^*} + \sum_{w \in W_{q_2} (Q_2\cup g_2
  \setminus g)} \gamma_w w + \sum_{w \in W(T)} \gamma_w w .
$$
We have $Q_2 \in \mindem{q_2}(\bp)$ so we can apply \cref{lem:min_weights} to
write $\val{Q_2}{q_2}$ as a sum of weights, where $s^* \in [\tau_{q_2}]$, we can
rewrite LHS-2 to get
\begin{align*}
p(g) =&  \weight{g}{q}{ t_k} +  \weight{g}{q_2}{t^*}  - \weight{g}{q_2}{s^*}
      +  \sum_{w \in W_{q_2} (Q_2\cup g_2 \setminus g)} \gamma_w w \\
& \qquad +\sum_{2<j < k} \gamma_j \weight{g_j}{q_j}{t_j'}
      + \sum_{w \in W(\mu)} \gamma_w w +  \sum_{w \in W(T)} \gamma_w w.
\end{align*}
Note that we have three weights corresponding to good $g$ and it may be the case
that $t_k = t^* = s^*$. By genericity, two of these weights may cancel (if they
are identical) but there will be at least one weight that does not cancel. By
genericity this could not cancel with any term on the RHS-2, since those weights
are for buyers different from $q$.  Thus, we get a non-trivial linear
combination, contradicting genericity.

Note if $k = 2$, we know that $g = g_2$ and $v_q(Q_2 \cup g) - v_q(Q_2) > 0$
(the price of $g$ is positive) .  Thus, $v_q(Q_2 \cup g)$ must contain a weight
for $g$ belonging to $q$, by \cref{cor:val_weights}, that cannot cancel with the
weight $w_r^{(t')}(g)$ in the RHS-2 due to $q\neq r$ and genericity.  Hence,
either way in case 2) we get a contradiction.

We now consider our last case 3), $g_1,h_1 = \bot$.  We will need to consider
three subcases without loss of generality: a) $k = \ell = 2$; b) $k = 2, \ell >
2$; c) $k,\ell >2$.  For our first subcase a), we can write the price of $g$  in
two different ways, according to \cref{lem:gs_price_induct}:
\[
  p(g) = \val{Q\cup g}{q} - \val{Q}{q}
  \qquad \text{and} \qquad
  p(g) = \val{R \cup g}{r} - \val{R}{r} ,
\]
where $Q, R$ are minimum bundles for bidders $q$ and $r$ respectively, and $Q
\cup g, R \cup g$ are demand bundles (but not minimum) for bidders $q$ and
$r$. Crucially, we also know that $p(g) > 0$, by construction of the swap graph.
Therefore, $Q \cup g$ and $R \cup g$ must be non-degenerate with respect to $g$,
for instance
\[
  \val{Q}{q} - p(Q) = \val{Q \cup g}{q} - p(Q \cup g) = \val{Q \cup g}{q} - p(Q)
  - p(g) < \val{Q \cup g}{q} - p(Q) ,
\]
so $\val{Q}{q} < \val{Q \cup g}{q}$; the same argument applies to $R \cup g$.
Furthermore, since $Q$ and $R$ are minimum, they are also non-degenerate by
\cref{lem:min-nondegen}. Combined with \cref{lem:min_weights} applied to $Q \cup
g$ and $R \cup g$, we have
\begin{align*}
  p(g) & = \val{Q\cup g}{q} - \val{Q}{q} =  \weight{g}{q}{t} + \sum_{w \in W_q(Q)} \gamma_w w + \sum_{w \in W(T)} \gamma_w w,  \tag{LHS-3a}\\
  p(g) & = \val{R \cup g}{r} - \val{R}{r} = \weight{g}{r}{t'} + \sum_{w \in
    W_r(R)} \gamma_w w + \sum_{w \in W(T)} \gamma_w w . \tag{RHS-3a}
\end{align*}
However, LHS-3a consists of integer linear combinations of weights for buyer
$q$, while the right hand side RHS-3a only has integer linear combinations of
the weights for buyer $r$, leading to a non-trivial linear combination
contradicting genericity.

For subcase $b)$ we can write the prices in two ways:

  \begin{align*}
p(g) = & \val{Q \cup g}{q} - \val{Q}{q} = \weight{g}{q}{t} + \sum_{w \in W_q(Q)} \gamma_w w + \sum_{w \in W(T)} \gamma_w w,\tag{LHS-3b} \\
p(g) =& \weight{g}{r}{t'_\ell } + \val{R_2 \cup h_2}{r_2} - \val{R_2}{r_2}  + \sum_{2<j < \ell} \gamma_j \weight{h_j}{r_j}{t_j} + \sum_{w \in  W(\mu)} \gamma_w w +  \sum_{w \in W(T)} \gamma_w w \tag{RHS-3b}
 \end{align*}
 where $Q$ is some minimum demand bundle for $q$ that does not contain $g$ and
 $R_2$ is some minimum demand bundle for buyer $r_2$ that does not contain $h_2$.
If $r_2 \neq r$ then the term $\weight{g}{r}{t'_\ell}$ in RHS-3b cannot cancel with
any term in LHS-3b, due to genericity.  If $r_2 = r$ then there is no term in
RHS-3b that can cancel with $\weight{g}{q}{t}$ on the LHS-3b because of
genericity and $q \neq r$.

 We then consider our last subcase c), when $k, \ell >2$.  We then use
 \cref{lem:gs_price_induct} to write the price of good $g$ in two ways, where
 $Q_2$ is some minimum demand bundle for $q_2$ that does not contain $g_2$ and
 $R_2$ is some minimum demand bundle for buyer $r_2$ that does not contain
 $h_2$:
  \begin{align*}
p(g) &= \weight{g}{q}{ t_k} +  \val{Q_2\cup g_2}{q_2}  - \val{Q_2}{q_2}+ \sum_{2<j < k} \gamma_j \weight{g_j}{q_j}{t_j} + \sum_{w \in
  W(\mu)} \gamma_w w +  \sum_{w \in W(T)} \gamma_w w \tag{LHS-3c} \\
p(g) & = \weight{g}{r}{ t'_\ell} + \val{R_2 \cup h_2}{r_2} -
\val{R_2}{r_2}+ \sum_{2<j < \ell} \gamma_j \weight{h_j}{r_j}{t_j'} + \sum_{w \in
  W(\mu)} \gamma_w w +  \sum_{w \in W(T)} \gamma_w w \tag{RHS-3c}
\end{align*}
Note that if $g \notin Q_2$ and $g \notin R_2$, then the term $\weight{g}{q}{t_k}$
for $q \neq r$ could not cancel with any term above due to genericity as we have
argued before.  Otherwise, without loss of generality  $g \in Q_2$. (The case
with $g \in R_2$ follows by symmetry.)
As we argued above in case 2, we can rewrite the price of $g$ in LHS-3c as
\begin{align*}
p(g) =&  \weight{g}{q}{ t_k} +  \weight{g}{q_2}{t^*}  - \weight{g}{q_2}{s^*}+  \sum_{w \in W_{q_2} (Q_2\cup g_2 \setminus g)} \gamma_w w \\
& \qquad +\sum_{2<j < k} \gamma_j \weight{g_j}{q_j}{t_j'} + \sum_{w \in
  W(\mu)} \gamma_w w +  \sum_{w \in W(T)} \gamma_w w. 
\end{align*}
Following the same argument as in case 2, we know that there exists at least one
weight belonging to buyer $q$ for good $g$.  Similarly, there exists at least
one weight (note that it can be the case that $t_k = t^* = s^*$) belonging to
buyer $r$ for good $g$.  By genericity these two weights cannot cancel because
$q \neq r$.

Hence, the buyer in-degree to any node in the swap graph is at most 1.
\end{proof}


We now analyze the case where buyers are told to take a most demand bundle that
is non-degenerate, i.e. each buyer $q$ takes a bundle from $\maxdem{q}(\bp)$.
We have already shown that the swap graph has buyer in-degree at most 1 and so
we apply \cref{lem:deg_od_gs} to obtain the following result.

\begin{theorem}
Let $(\bp,\mu)$ be a Walrasian Equilibrium with minimal Walrasian prices $\bp$.
If buyers have valuations in $GMBV$ and each buyer chooses only non-degenerate
bundles at prices $\bp$, then the over-demand for any good $g \in [m]$ is at most $1$, i.e.
$$
 \maxod{g} \leq 1, \qquad \forall g \in [m].
$$
\LABEL{thm:gs_od2}
\end{theorem}
\SUBSECTION{Bounding welfare}
So far,
we have shown that over-demand can be no more than one when buyers select any
non-degenerate bundle in their demand correspondence at the minimal Walrasian prices.
Much like the case for matchings, where we assumed that buyers broke ties in
favor of selecting a good in their demand correspondence (rather than selecting
the empty set), we will assume that buyers break ties in favor of the
\emph{largest cardinality} non-degenerate bundles in their demand
correspondence---i.e. we will assume that buyers select arbitrary bundles in
what we will call the \emph{max non-degenerate} correspondence
$\cD_q^{max}(\bp)$ for buyer $q$ at price $\bp$ where
\begin{equation}
\cD_q^{max}(\bp) = \{D_q \in \cD_q^\bullet(\bp) : |D_q| \geq |D_q'|,  \quad
\forall D_q' \in \cD_q^\bullet(\bp) \} .
\end{equation}
Under this tie-breaking condition, buyers achieve nearly optimal welfare.  

\begin{theorem}
 \LABEL{thm:gs_welfare}
Let buyers have GMBV $\left\{ v_q: \cG \to [0,\maxval]\right\}_{q \in N} $ in
a market in which there are $s_g$ copies of each good $g \in [m]$.  At the
minimal Walrasian prices $\bp$, if each buyer $q$ selects any \emph{max
  non-degenerate} bundle $ B^{max}_q \in \cD^{max}_q(\bp)$, then welfare will
be near optimal:
\[
  \wel{B_1^{max}, \cdots, B_n^{max}}{N} \geq \optw{N} - 2H\cdot m .
\]
\end{theorem}
\begin{proof}
  Without loss of generality we will assume that the Walrasian allocation
  $\mu$ is non-degenerate; if it were not, we could remove from each allocated
  bundle all goods with zero marginal valuation while remaining in a buyer's
  demand set and thus we would obtain the same welfare as  the original
  allocation.  Note that the welfare from $\mu$ is equal to $\optw{N}$.

  Let $ B^{max} = ( B_1^{max}, \dots,  B_n^{max})$ denote a collection of
  max non-degenerate bundles at price $\bp$, one per buyer. We know that $|
  B_q^{max}| \geq | \mu_q|$ by maximality, so $ B^{max}$ contains
  at least as many goods as $\mu$.  Note that the utility each buyer receives
  from the allocation $ \mu$ is the same as the utility they receives from her
  bundle in $ B^{max}$.  As we did in the proof of \cref{thm:unit_welfare}, we
  sum over all buyers' utilities to get
  \begin{align*}
    \sum_{q \in N}&\left[ \val{\mu_q}{q} - p(\mu_q )\right]  =  \sum_{ q \in N } v_q( \mu_q ) - \sum_{g \in [m]} p(g) s_g  =  \sum_{q \in N } \left[ \val{ B_q^{max}}{q} - p( B_q^{max} )\right]   \\
    & =   \sum_{q \in N}
    \val{  B_q^{max} }{q} - \sum_{g \in [m]} p(g) s_g  + \sum_{g \in [m]} p(g) \cdot ud(g,
    B^{max}) - \sum_{g \in [m]} p(g) \cdot od(g, B^{max}).
  \end{align*}
  where we define the over-demand ($od$) and under-demand ($ud$),
  respectively, of good $g$ in the tuple of bundles $B^{max}$ compared to the
  optimal allocation $\mu$ as
  \begin{align*}
    od(g, B^{max}) &=  \max\left\{ \sum_{q \in N} \left( \1\{g \in B_q^{max} \} -
        \1\{g \in \mu_q \} \right) , 0\right\} \\
    ud(g, B^{max}) &= \max\left\{ \sum_{q \in N}
      \left( \1\{g \in \mu_q \} - \1\{g \in B_q^{max} \} \right) , 0\right\} .
  \end{align*}
  Note that the $od(g, B^{max})\geq 0$ and so we get the following inequality:
  $$
  \optw{N} \leq \sum_{q \in N} \val{  B_q^{max} }{q}+ \sum_{g \in [m]} p(g) \cdot
  ud(g, B^{max}) \leq \sum_{q \in N} \val{  B_q ^{max}}{q}+ H\sum_{g \in [m]}
  ud(g, B^{max}) .
  $$
  To bound the right hand side, we will bound the total under demand of all
  goods from allocation $ B^{max}$ relative to $\mu$.


  We now use the fact that $B_q^{max}$ is a max non-degenerate bundle, so we have
  \begin{align*}
    \sum_{q \in N} |\mu_q| &\leq \sum_{q \in N} |B^{\max}_q| = \sum_{q \in N}
    |\mu_q| + \sum_{g \in [m]} \left(od(g,B^{max})-  ud(g,B^{max}) \right) \\
    \sum_{g \in [m]}ud(g,B^{max}) &\leq \sum_{g \in [m]}od(g,B^{max}) .
  \end{align*}

  To bound the over demand for each good from $B^{max}$ relative to $\mu$, we
  again turn to the swap graph.  We first bound how many more goods are allocated
  by $ B^{max}$ compared to $ \mu$. We consider buyers $q \in N$ and good $g$
  where $g \in  B_q^{max}$ and $g \notin \mu_q$.
  Let $B^*_q$ be any minimum demand bundle contained in $ B_q^{max}$.  If $g \in
  B^*_q$ then we know that there will be an edge to $g$ in the swap graph
  $G(\bp,\mu,M)$ for any choice of minimum demand bundles $M = (M_1, \dots, M_n)$
  contained in $\mu$.  Otherwise, if $g \notin B^*_q$ then we know that $
  B_q^{max}\setminus g$ is also a most demanded bundle from
  \cref{lem:demand-interpol}.  We further know that $p(g) >0$, due to $\val{
    B_q^{max}}{q} > \val{ B_q^{max} \setminus g}{q}$ ($B_q^{max}$ is
  non-degenerate).  So, there is an edge from $\bot$ to $g$ in $G(\bp,\mu,M)$ for
  any choice of minimum demand bundles $M$. Hence, every good that $B^{max}$
  allocates over $ \mu$ has at least one incoming edge.

  Because the buyer in-degree in $G(\bp,\mu,M)$ is at most one for any
  choice of minimum demand bundles $M$, we know that for
  each good $g$ there is at most one buyer $q$ where $g \in B_q^{max}$ but $g
  \notin \mu_q$.  Hence, we get $od(g, B^{max}) \leq 1$ for each $g \in
  [m]$:
  $$
  \optw{N} \leq \sum_{q \in N} v_q(B_q^{max}) + \ngood\cdot\maxval
  =  \rwel{b_1, \ldots, b_\sam}{N} + \ngood\cdot \maxval.
  $$
  Finally, applying \Cref{lem:welfare-equiv} completes the proof.
\end{proof}
\fi

%% file: vwmfig.tex
\begin{figure}[h]
\begin{center}
\begin{tikzpicture}[every node/.style={circle,draw},level 1/.style={sibling distance=220pt},level 2/.style={sibling distance=150pt},level 3/.style={sibling distance=80pt} ,level distance =55pt]
    \node {$v$}
    child{node {$(v'')^{1,2}$}
    	child{node {$(v')^{1,2}$ }  child{node{$v^{1,2}$} child{node{$v^{(1)}$}
            child{node[draw=none, rectangle]{$\substack{W^{(1)}, \\ T^{(1)},
                  \cI^{(1)}}$} } } child{node {$v^{(2)}$}   child{node[draw=none,
              rectangle]{$\substack{W^{(2)}, \\ T^{(2)}, \cI^{(2)}}$} } }  edge from parent node[right,draw=none] {$J^{(1)}$}}  edge from parent node[right,draw=none] {$J^{(2)}$}}
    }
    child {node{$(v')^{3:6}$}
    	child{node{$(v')^{3,4}$}
		child{node {$v^{3,4}$}
			child{node {$v^{(3)}$} child{node[draw=none, rectangle]{$\substack{W^{(3)}, \\
              T^{(3)}, \cI^{(3)}}$} } } child{node{$v^{(4)}$}child{node[draw=none,
          rectangle]{$\substack{W^{(4)}, \\ T^{(4)}, \cI^{(4)}}$} } }
			 edge from parent node[right,draw=none] {$J^{(3)}$}}
	}
	child{node[draw] {$(v')^{5,6}$}
		child{node{$(v')^5$} child{node {$v_q^{(5)}$} child{node[draw=none,
          rectangle]{$\substack{W^{(5)}, \\ T^{(5)}, \cI^{(5)}}$} }  edge from parent
        node[right,draw=none] {$J^{(5)}$}} } child{node{$(v')^6$}
      child{node{$v^{(6)}$}child{node[draw=none, rectangle, level
          distance=10pt]{$\substack{W^{(6)}, \\ T^{(6)}, \cI^{(6)}}$} } edge from parent node[right,draw=none] {$J^{(6)}$} }}
	}
    }
    ;
\end{tikzpicture}
\caption{An example for the valuation $v$ that is MBV.  The leaf nodes represent
  VIWMs that have endowed goods $T^{(t)}$, weights $w^{(t)}$, and matroid with
  independent sets $\cI^{(t)}$ with ground set $[m] \cup T^{(t)}$ for $t = 1,
  \dots, \tau$ where $\tau= 6$ in this case. When there is a parent node with a
  single child, that represents an endowment operation with the set of endowed
  goods on the edge label, while if it has two children it is the merge of the
  two child valuations.  Note that in this example, we would require $T^{(1)} =
  T^{(2)}$ and $T^{(3)}=T^{(4)}$ because a merge operation must take two valuations
  defined over the same set of goods.  Also, $v$ has domain $2^{[m]}$, so that
  $J^{(1)} \cup J^{(2)} = T^{(1)}$, $J^{(1)}\subseteq T^{(1)}$, $J^{(2)}\subseteq T^{(1)}\setminus
  J^{(1)}$, $J^{(3)} = T^{(3)} = T^{(4)}$, $J^{(5)} = T^{(5)}$, and $J^{(6)} = T^{(6)}$.}
\LABEL{fig:mbv_tree}
\end{center}
\end{figure}
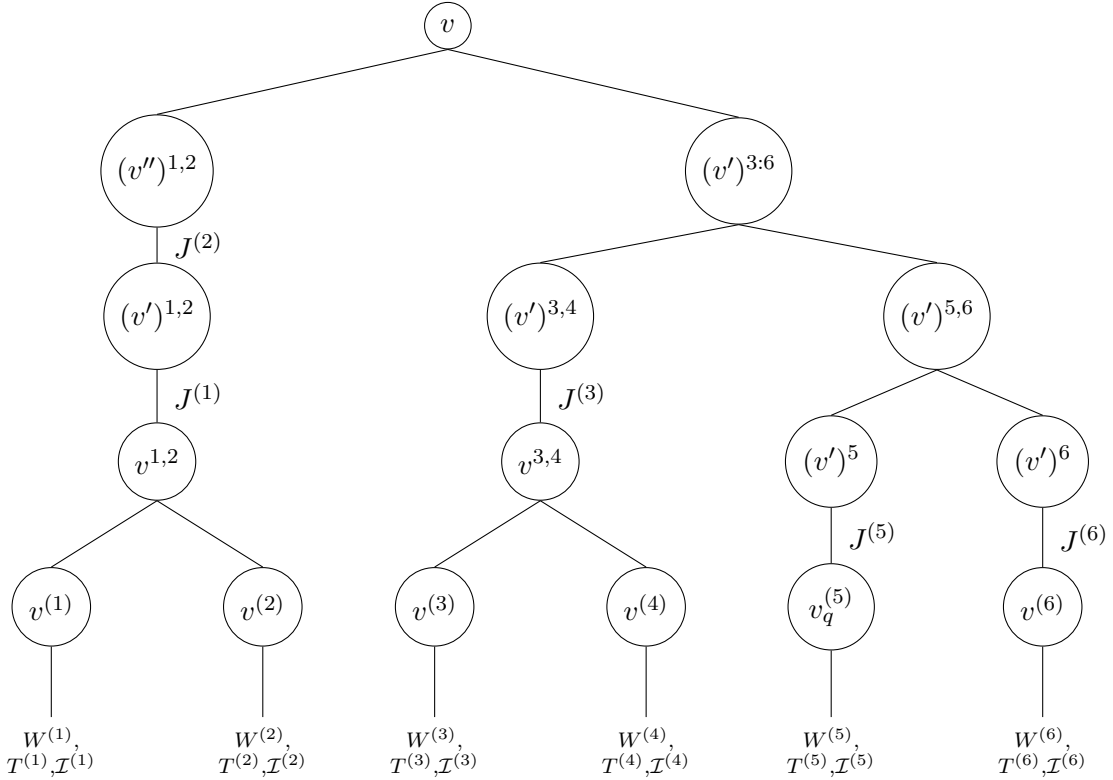

%% file: price-dim.tex
\section{Welfare and over-demand \shortbreak generalization}
\LABEL{sec:generalize}

In this section, we show that Walrasian prices \emph{generalize}: the
equilibrium prices for a market $N$ of buyers induces similar behavior
when presented to a new sample $N'$ of buyers, both in terms of the
demand for each good and in terms of welfare.  More precisely, a set
of prices which minimizes over-demand of each good \emph{and}
maximizes welfare when each buyer purchases her most-preferred bundle
retains these properties (approximately) on a \emph{new} market $N'$
when each buyer purchases her most-preferred bundle, if buyers in $N$
and $N'$ are drawn independently from the same distribution.

\SUBSECTION{Results for Arbitrary Valuations}
\LABEL{sec:arbitrary-generalization}

\ifshort
\else
We need a bit more notation to formally state the generalization
guarantees for more general valuations.
\fi
To foramlize how over-demand and welfare induced by prices
$\bp$ vary between two markets, we first fix a tie-breaking rule $e$
that buyers use to select amongst demanded bundles for a
valuation from set $\V$.  We define classes of \emph{functions}
parameterized by pricings mapping valuations to (a) bundles purchased,
(b) whether or not a particular good $g$ is purchased, and (c) the
value a buyer gets for her purchased bundle.

\ifshort
Let $e :2^\G \times \V \to \G$ denote some tie-breaking function,
which satisfies some mild technical conditions defined in the extended version.
\else

For technical convenience, we will assume \emph{all} buyers choose
bundles from their demand sets in the following systematic way (still
independent of other buyers' choices).
First, each buyer $q$ will have some set $\cL_{v_q}$ of
``infeasible'' bundles they will \emph{never} choose to buy. For instance,
we will prevent buyers from buying items for which
they get zero marginal utility: if $v(B) = v(B \cup \{g\})$ for
$g\notin B$, having $B \cup \{g\}\in \cL_{v_q}$ forces $q$ to avoid
the unnecessarily larger bundle (needed for over-demand guarantees).
Second, all buyers share some standard ``linear'' choice rule over the elements
of the remaining bundles in their demand set; we use the rule to break ties in
favor of larger demanded bundles, needed for welfare guarantees.

Below we formally define an encodable tie-breaking rule, which we will
need for our concentration results for demand and welfare.

\begin{definition}[Encodable tie-breaking] A function
  $e : 2^\G \times \V \to \G$ is an \emph{encodable rule}
  if
\begin{itemize}
\item $e(\X,v) \in \X$;
\item
  there exists a vector $y\in \left({\R_{\geq 0}}\right)^\ngood$ (called a \emph{separator}),
  and a subset $\cL_{v_q}\subseteq 2^\G$ (called the \emph{infeasible
    set} for buyer $q$)
  such that, for all $\X$, $e(\X,v) \notin \cL_{v}$, and
  $\sum_{g \in e(\X,v)} y_g > \sum_{g \in B'} y_g $ for any other
  $B'\in \X\setminus \cL_{v_q}$; and
\item at any pricing $\bp$ there is at least one utility-maximizing
  bundle for each valuation that lies outside the infeasible set
  $\cL_{v_q}$.
  \end{itemize}
\end{definition}

We show in \thelongref{lem:existence-ties} that there exists an
encodable tie-breaking rule over demand sets which selects bundles
satisfying the properties needed for our over-demand and welfare
results to hold (\cref{thm:unit_od,thm:unit_welfare}) in the case of
unit demand buyers, and the properties needed for the analogous
results \preflong{thm:gs_od2,thm:gs_welfare} in the case of GMBV
buyers.  For the remainder of this section, we fix some encodable
tie-breaking rule $e$ and state our results with respect to this
rule.
\fi
We call the bundle that $e$ selects from $\dem{q}(p)$ the
\emph{canonical} bundle for $q$ at $p$, denoted by
$\canon{q}(\p) = e\left(\dem{q}(p), v_q\right)$.  Then, for each good
$g$, let $h_{g,\p}(v_q) = \I\lbrack g\in \canon{q}(\p)\rbrack$
indicate whether or not $g$ is in $q$'s canonical bundle at prices
$\p$.  Let $\ndem{g}{\p}{N}{e} = \sum_{q\in N}h_{g,\p}(v_q)$, i.e. the
number of buyers in $N$ whose canonical bundle at $\p$ contains
$g$. For a sample of $\sam$ buyers $\{v_q\} \stackrel{i.i.d.}{\sim} \D$, let
$\ndem{g}{\p}{\D}{e}$ represent the \emph{expected} number of copies
of $g$ demanded at prices $\p$ if buyers demand canonical
bundles. Similarly, let the \emph{welfare} of a pricing $\p$ on a
market $N$ be
$\welfare{N}{p} = \sum_{q \in N} v_q(\hat{\canon{q}}(\p))$ for
$\hat{\canon{q} }(\p)$ a worst-case resolution of the over-demand for
each good%
\ifshort. \else
(See~\cref{lem:welfare-equiv}). Similarly, let
$\rwel{N}{p} = \sum_{q \in N} v_q(\canon{q}(\p))$ denote the relaxed
welfare of buyers choosing their canonical bundles at prices $p$;
analogously, for a distribution $\D$ over valuations, define the
expected relaxed welfare of buyers choosing canonical bundles at
prices $\p$ as $\rwel{\D}{\p} = \sam \Ex{v\sim\D}{v(\canon{v}(\p))}$.
In both
cases, these quantities assume there is sufficient supply to satisfy
the demand for each good.\footnote{If the supply of each good
  satisfies the inequality in \cref{thm:demand-bernstein}, then each
  good should have supply which can satisfy at least an $1-\alpha$
  fraction of buyers who choose it. See \cref{lem:welfare-equiv}
  for a more detailed discussion.}
\fi

Before presenting the technical details, we show
how the behavior induced by prices
generalizes for buyers with arbitrary valuations.  Our first theorem
bounds the over-demand when the Walrasian
prices computed for a market are applied to a new market.

\begin{theorem}\LABEL{thm:demand-bernstein}
  Fix a pricing $\p$ and two sampled markets $N, N'$ of buyers with
  arbitrary valuations, with $|N| = |N'| = \sam$.
  \ifshort\else Suppose $e$ is encodable. \fi
  For good $g$, suppose $\ndem{g}{\p}{N}{e} \leq s_g + 1$.
  \ifshort\else
  Then,
  \[
    \ndem{g}{\p}{N'}{e} - s_g = O\left(\sqrt{\ngood\cdot s_g\cdot\ln\frac{1}{\delta}} + \ngood\ln\frac{1}{\delta} \right),
  \]
  with probability $1-\delta$.
  \fi
  Then, for any $\alpha\in (0,\frac{4}{5})$, if
  $s_g = \Omega\left(\ngood \ln\frac{1}{\delta}/\alpha^2\right)$, with probability $1-\delta$,
  \[
  \ndem{g}{\p}{N'}{e} \leq (1+\alpha) s_g.
\]
\end{theorem}

Our second theorem is an analagous generalization result, for welfare
instead of over-demand.
\begin{theorem}\LABEL{thm:welfare-bernstein}
  Fix two markets $N, N'\sim \D$ for which $|N| = |N'| =\sam$. Suppose
  for each good $g$, the prices $p$ satisfy
  $\ndem{g}{\p}{N}{e} \leq s_g + 1$.  With probability $1-\delta$,
  \ifshort \else if $e$ is encodable, \fi
 for any $\alpha \in (0, 4/5)$, if $p$ are welfare-optimal for $N$ and
\[
\optw{N}=\Omega\left(\frac
    {\maxval^{3}\sam^{0.5}\ngood^4\ln^{2}\left(\ngood\right)\ln^{2}\frac{1}{\delta}}
    {\alpha^2}
\right)
 \]
then
\[
\welfare{N'}{\p} \geq (1- \alpha)\optw{N'}.
\]
\end{theorem}

\ifshort These proofs can be found in the full version. \fi
The first result
(\cref{thm:demand-bernstein}), for over-demand, relies on bounding the VC
dimension of the class of good $g$'s demand indicator functions (the
class contains a function for each pricing, labeling a valuation
$v$ positive if and only if $g$ is in the canonical bundle for $v$ at
those prices). To bound the VC dimension, we first prove tth he class of
\emph{bundle} predictors (the class contains one function for each
pricing $p$, which maps $v$ to $v$'s canonical bundles at $p$) is
\emph{linearly separable} in a space of $m+1$ dimensions. Using a
recent result by~\citet{daniely2014multiclass}, such a class admits a
compression scheme of size $m+1$ and can be
$(\epsilon, \delta)$-PAC learned with
$O\left(\frac{m}{\eps}\ln\frac{1}{\delta}\right)$ samples. Since item
predictor functions are simply projections of bundle prediction
functions, their \PAC-complexity of the item predictors is an upper bounds the
\PAC-complexity of the bundle predictors. Since item
predictors are binary valued, we can apply the classical result of
\cite{ehrenfeucht1989general} which shows the equivalence of learning
and uniform convergence for binary prediction problems, so we can upper
bound the class's VC dimension as a function of its \PAC sample
complexity.

The second result (\cref{thm:welfare-bernstein}), for welfare, proceeds by
bounding the pseudo-dimension
of the class of welfare predictor functions (containing, for each
pricing a function which map valuations to the \emph{value} of a buyer
purchasing her canonical bundle at these prices). The argument uses
the existence of a compression scheme to upper bound the number of
possible distinct labelings of valuations by bundles. Fixing a bundle
labeling of a valuation also fixes the welfare of that valuation;
thus, one can upper-bound the size of a ``shatterable'' set and the
pseudo-dimension of the class of welfare predictors.

\SUBSECTION{Learning Theory Preliminaries}

This section reviews several well-known definitions from learning
theory. We will show that:
\begin{itemize}
\item the demand of a good $g$ from a market $N$ at prices $\p$ should
  be close to the demand for $g$ on a new market $N'$; and
\item the welfare-optimal prices for $N$ should be
  approximately welfare-optimal for $N'$
\end{itemize}
under appropriate conditions.

Viewed from a learning-theoretic perspective, both of these statements
are about the \emph{generalization guarantees} of a class of
mechanisms. Learning theory provides many tools to formalize how well
properties of prices generalize from one market to the next. In
what follows, we will first describe standard tools for measuring the
generalizability of binary-valued functions, followed by the analogous
tools for understanding the generalizability of real-valued
functions. The former will be useful in measuring the concentration of
demand for a good $g$, while the latter will be useful for measuring the
concentration of welfare.

Suppose there is a domain $\V$ (for our purposes, the valuation functions),
an unknown target function $c: \V \to \{0,1\}$, and an unknown distribution
$\D$ over $\V$. We wish to understand how many labeled samples
$(v, c(v))$, with $v\sim \D$, are necessary and sufficient to be able to
compute a $\hat c$ which agrees with $c$ almost everywhere with respect
to $\D$. The distribution-independent sample complexity of learning
$c$ depends fundamentally on the ``complexity'' of the set of binary
functions $\F$ from which we are choosing $\hat c$.  We review two
standard complexity measures next.

Let $N$ be a set of $\sam$ samples from $\V$. The set $N$ is said
to be \emph {shattered} by $\F$ if, for every subset $T\subseteq N$,
there is some $c_T\in\F$ such that $c_T(v) = 1$ if $v\in T$ and
$c_T(v') = 0$ if $v'\notin T$.  That is, ranging over all $c \in \F$
induces all $2^{|N|}$ possible projections onto $N$.  The {\em VC
  dimension} of $\F$, denoted $\VC(\F)$, is the size of the largest
set $S$ that can be shattered by $\F$.

Let $\err_N(\hat c) = (\sum_{v\in N} |c(v) - \hat{c}(v)|)/|N|$ denote
the empirical error of $\hat c$ on $N$, and let
$\err(\hat c) = \E_{v\sim \D}[|c(v) - \hat{c}(v)|]$ denote the
\emph{true} expected error of $\hat c$ with respect to $\D$.  We say that
$\F$ is {\em $(\epsilon, \delta)$-\PAC learnable with sample
  complexity $\sam$} if there exists an algorithm $\A$ such that, for
all distributions $\D$ and all target functions $c\in\F$, when $\A$ is
given a sample $S$ of size $\sam$ it produces some $\hat{c}\in \F$
such that $\err(\hat{c}) < \epsilon$ with probability $1-\delta$ over
the choice of the sample.  The \PAC sample complexity of a class $\F$
can be bounded as a polynomial function of $\VC(\F)$, $\epsilon$, and
$\ln\frac{1}{\delta}$~\citep{VC}; furthermore, any algorithm which
$(\epsilon, \delta)$-PAC learns $\F$ over all distributions $\D$
\emph{must} use nearly as many samples to do so. The following theorem
states this well-known result formally.\ifshort\else\footnote{%
  The upper bound
  stated here is a quite recent result which removes a
  $\ln\frac{1}{\epsilon}$ factor from the upper bound; a slightly
  weaker but long-standing upper bound can be attributed
  to~\citet{vapnik1982estimation}.}\fi

\begin{theorem}[\citet{hanneke15,ehrenfeucht1989general}]
  \LABEL{thm:pac-convergence}
  Suppose $\F$ is a class of binary functions. Then,
  $\F$ can be $(\epsilon, \delta)$-PAC learned with a sample of size
\[\sam  = O\left(\frac{\VC(\F) + \ln\frac{1}{\delta}}{\epsilon}\right)\]
whenever $c\in\F$.  Furthermore, any $(\epsilon, \delta)$-\PAC
learning algorithm for $\F$ must have sample complexity
\[\sam = \Omega\left(\frac{\VC(\F) + \ln\frac{1}{\delta}}{\epsilon}\right).\]
\end{theorem}

There is a stronger sense in which a class $\F$ can be learned, called
\emph{uniform learnability}. This property implies that with a
sufficiently large sample, the error of \emph{every} $c\in\F$ on the
sample is close to the true error of $c$.  We say $\F$ is {\em
  $(\epsilon, \delta)$-uniformly learnable with sample complexity
  $\sam$} if for every distributions $\D$ and a sample $N$ of size
$\sam$, $|\err_N(c) - \err(c)| < \epsilon$ for every $c\in \F$, with probability
$1-\delta$.
Notice that if $\F$ is $(\epsilon, \delta)$-uniformly learnable with
$\sam$ samples, then it is also $(\epsilon, \delta)$-\PAC learnable
with $\sam$ samples.  We will use a well-known upper bound on the
uniform sample complexity of a class as a function of its VC
dimension.

\begin{theorem}[E.g.~\citet{VC}]
\LABEL{thm:uniform-convergence}
  Suppose $\F$ is a class of binary functions. Then, $\F$ can be
  $(\epsilon, \delta)$-uniformly learned with a sample of size
\[\sam  = O\left(\frac{\VC(\F)\ln\frac{1}{\epsilon} + \ln\frac{1}{\delta}}{\epsilon^2}\right).\]
\end{theorem}

Both \PAC learnability and uniform learnability of binary-valued
functions are characterized by the class's VC
dimension. When learning real-valued functions, for example, to
guarantee convergence of the welfare of pricings), we use a
real-valued analogue to VC dimension which will give a sufficient condition for
uniform convergence, called the 
the \emph{pseudo-dimension} \citep{pollard1984}.  Formally, let $c : \V \to [0,\maxval]$ be a
real-valued function over $\V$, and $\F$ be the class we are learning
over. Let $S$ be a sample drawn from $\D$, $|N|=\sam$, labeled
according to $c$.  Both the empirical and true errors of a hypothesis
$\hat c$ are defined as before, though $|\hat c(v) - c(v)|$ is now
in the interval $[0, \maxval]$ rather than in $\{0,1\}$. Let
$(r_1, \ldots, r_\sam) \in [0,\maxval]^\sam$ be a set of
\emph{targets} for $N$. We say $(r_1, \ldots, r_\sam) $
\emph{witnesses} the shattering of $N$ by $\F$ if, for each
$T\subseteq N$, there exists some $c_T\in \F$ such that
$c_T(v_q) \geq r_q$ for all $v_q \in T$ and $c_T(v_q) < r_q$ for all
$v_q \notin T$. If there exist such witnesses, we say $N$ is {\em shatterable}
by $\F$.  The {\em
  pseudo-dimension} of $\F$, denoted $\pd(\F)$, is the size of the
largest set $S$ which is shatterable by $\F$.  We will derive sample
complexity upper bounds from the following theorem, which connects
the sample complexity of uniform learning over a class of real-valued
functions to the pseudo-dimension of the class.

\begin{theorem}[E.g.~\citet{AB}]\LABEL{thm:fat-sample}
  Let $\F$ be a class of real-valued functions with range in
  $[0,\maxval]$ and pseudo-dimension $\pd(\F)$. For every
  $\epsilon > 0, \delta \in [0,1]$, the sample complexity of
  $(\epsilon, \delta)$-uniformly learning the class $\F$ is
  \[
    n = O\left(
      \left(\frac{\maxval}{\epsilon}\right)^2\left(\pd(\F)\ln \frac{\maxval}{\epsilon}
        + \ln \frac{1}{\delta} \right)\right).
  \]
\end{theorem}
Moreover, a conceptually simple algorithm achieves the guarantee in
\cref{thm:fat-sample}: simply output the function $c \in \F$ with the
smallest empirical error on the sample. These algorithms are called
\emph{empirical risk minimizers}.

\SUBSECTION{Learning from Pricings and Learning with Compression}

We now introduce several other powerful tools from learning theory,
recast in the language of mechanism design.  Let
$(v_1, B_1), \ldots, (v_\sam, B_\sam), v_q \in \V,B_q\in \cG$ represent
samples drawn from some distribution $\D$ over $\V$, labeled by
bundles $B_q \in \cG$.  For each set of prices $p$, the
functions $h_p, f_p$ will map valuations to utility-maximizing bundles
and the value the buyer of utility-maximizing bundles,
respectively. We prove bounds on the sample complexity of
uniform convergence over three classes of functions: 
\begin{itemize}
  \item[1)] the class of functions which map valuations $v_q$ to
    $v_q(B_q)$ (where $B_q$ is utility-maximizing):
    \[ \H_V = \{f_\p :  \V \to \R, f_\p(v_q) = v_q(B^*), B^* = e(\dem{q}(p))\}; \]
  \item[2)] the class of bundle predictors:
    \[\H = \{ h_\p : \V \to \cG, h_\p(v_q) = e(\dem{q}(p))\};\]
  \item[3)] the projection of $\H$ to its component \emph{good} predictors:
    \[\H_g = \{h_{g,\p} : \V \to \{0,1\}, h_{p}(v) =
      \left(h_{g,\p}(v)\right)_g\} .\]
\end{itemize}

We will show that we can learn the class $\H$ using a \emph{compression
  scheme}, a tool for proving sample complexity bounds in the
multi-label setting.
\begin{definition}
A \emph{compression scheme} for $\F : \V \to X$,  of size $d$ consists of
\begin{itemize}
  \item a \emph{compression} function \[\com : (\V \times X)^{\sam} \to (\V \times
      X)^{d},\] where $\com(N) \subseteq N$ and  $d \leq \sam$; and
  \item a \emph{decompression} function
    \[\decom : (\V\times X)^d \to \F.\]
\end{itemize}
For any $f\in \F$ and any sample
$(v_1, f(v_1)), \ldots, (v_\sam, f(v_\sam))$, the functions satisfy
\[
  \decom \circ \com ((v_1, f(v_1)), \ldots, (v_\sam, f(v_\sam))) = f'
\]
where $f'(v_q) = f(v_q)$ for each $q\in [\sam]$.
\end{definition}

Intuitively, a compression function selects a subset of $d$ most
relevant points from a sample, and based on these
points, the decompression scheme selects a hypothesis.  When such a
scheme exists, the learning algorithm $\decom \circ \com$ is an
empirical risk minimizer. Furthermore, this compression-based learning
algorithm has sample complexity bounded by a function of
$d$, which plays a role analagous to VC dimension in the sample
complexity guarantees.

\begin{theorem}[\citet{littlestone1986compression}]
Suppose $\F$ has a compression scheme of size $d$.
Then, the $\PAC$ complexity of $\F$ is at most
\mathmath\sam = O\left(\frac{d \ln\frac{1}{\epsilon} +
   \ln \frac{1}{\delta}}{\epsilon}\right)\mathmath.
\LABEL{thm:compression-sample}
\end{theorem}

While compression schemes imply useful sample complexity
bounds, it can be hard to show that a particular hypothesis class
admits a compression scheme. One general technique is to show that the class is
linearly separable in a higher-dimensional space.

\begin{definition}
A class $\F$ is
\emph{$d$-dimensionally linearly separable} if there exists a function
$\psi : \V \times \cG \to \R^d$ such that for any $f\in \F$, there
exists some $w_f\in \R^d$ with
$f(v) \in \argmax_{y}\langle w_f, \psi(v, y)\rangle$  and
 $|\argmax_{y}\langle w_f, \psi(v, y)\rangle| = 1$.
\end{definition}
It is known that a $d$-dimensional linearly separable class admits a compression
scheme of size $d$.
\begin{theorem}[{\citet{daniely2014multiclass}}]\LABEL{thm:linsep}
  Suppose $\F$ has a $d$-dimensional realizable linear separator
  $\psi$. Then, there exists a compression scheme for $\F$ of size
  $d$.
\end{theorem}
\SUBSECTION{Over-Demand Concentration}\LABEL{sec:over-demand-concentration}
In this section, we outline the proof of over-demand concentration
(\cref{thm:demand-bernstein}). The proof proceeds in two steps. First,
we bound the VC dimension of the class of good predictors $\H_g$ by $O(\ngood)$
(\cref{thm:tighter}).  For arbitrary valuations,
the argument first shows that the class of bundle predictors $\H$ is
$(\ngood+1)$-linearly
separable and is thus \PAC learnable with a sample size proportional to
$\ngood$. Since $h_{g, p}\in \H_g$ is a projection of $h\in \H$,
$\H_g$ must also be \PAC learnable with the same number of samples. Since
functions in $\H_g$ are binary classifiers and \PAC learning is
completely characterized by VC dimension in the binary setting, we can
bound the VC dimension of $\H_g$.  Then, since $\H_g$ has small VC dimension,
it is possible to bound the maximum
difference between the sampled demand on $N$ and the new demand on
$N'$ for any $h_{g,p}\in H_g$ by
$\tilde{O}\left(\VC(\H_g)\cdot \sqrt{s_g}\right)$
(\cref{thm:vc-bernstein}, whose proof follows standard arguments for
bounding \PAC sample complexity, using Bernstein's inequality in place
of Hoeffding's inequality to achieve an upper bound on error which
scales as $\sqrt{s_g}$ rather than $\sqrt{n}$).


\begin{theorem}\LABEL{thm:tighter}
  The VC dimension of $\H_g$ is at most $O\left(m\right)$%
  \ifshort.\else if the tie-breaking rule $e$ is encodable.\fi
\end{theorem}

The proof of \cref{thm:tighter} relies on the following lemma about
the linear separability of the class.
\begin{lemma}
  $\H$ is $(\ngood+1)$-linearly separable%
  \ifshort.\else, if the tie-breaking rule $e$ is encodable.\fi
\LABEL{lem:linsep}
\end{lemma}

\begin{proof}[\ifshort Proof \fi of \cref{thm:tighter}]
  \Cref{lem:linsep} states that $\H$, the class which predicts
  bundles, is $(\ngood+1)$-linearly separable. Then, \cref{thm:linsep}
  implies that there exists a compression scheme of size $m+1$ for $\H$.
  Thus, one can \PAC-learn $\H$ with at most
  $\sam = O\left(\frac{m\ln\frac{1}{\epsilon} + \ln
      \frac{1}{\delta}}{\epsilon}\right)$
  samples, by \cref{thm:compression-sample}.  Since $\H_g$ is a
  projection of $\H$ onto its the $g$th coordinate, \PAC-learning
  $\H_g$ cannot require more samples than \PAC-learning $\H$, so this
  sample complexity bound on $\sam$ also applies to $\H_g$.  Then,
  \cref{thm:pac-convergence} implies that the VC dimension of the binary class
  $\H_g$ satisfies
  \[
    c \cdot \left(\frac{m + \ln
        \frac{1}{\delta}}{\epsilon}\right) \geq \frac{VC(\H_g) +
      \ln\frac{1}{\delta}}{\epsilon}
  \]
  for some constant $c$.
The claim follows by basic algebra.
\jh{Ignoring the dependence on $\delta$?}
\end{proof}

We briefly sketch the proof of \Cref{lem:linsep} (the full proof can
be found in
\thelongref{sec:omit-linsep}).  In order to show $\H$ is
$(\ngood+1)$-linearly separable, we must define two things. First, we
define a mapping $\Psi : \V \times \cG \to \R^{\ngood + 1}$ where
$\Psi(v, B)_g$ encodes whether or not $g\in B$ for each good
$g\in [\ngood]$, and $\Psi(v, B)_{m+1} = v(B)$ encodes the buyer's
value for a bundle $B$.  Second, we define a weight
vector $w^\p \in \R^{\ngood+1}$ for each price vector $\p$ where
$w^\p_g = -\p_g$ encodes the \emph{cost} of good $g\in[\ngood]$ at
these prices, while $w^\p_{\ngood+1} = 1$. Accordingly, the dot product of $\Psi(v, B)$ and $w^\p$
encodes the utility of a
buyer $v$ buying a bundle $B$ at prices $\p$:
$\Psi(v, B) \cdot w^\p = v(B) - p(B)$. Thus, a
utility-maximizing bundle $B^*$ maximizes the dot product
$B^* \in \argmax_{B} \Psi(v,B) \cdot w^\p$.

Unfortunately, there are two obstacles with this plan of
attack. First, the statement of \cref{thm:linsep} assumes that the
maximum is \emph{unique}---if not, prediction is not even
well-defined. Second, in order to obtain welfare guarantees, we need
some assumptions on how ties are broken, namely, that a buyer buys her
canonical bundle. To solve both problems simultaneously, we describe
how to perturb $\Psi$ and $w^\p$ to ensure that the argmax is unique,
and that
$B^* = \argmax_B(\Psi(v, B) \cdot w^\p) = e(\argmax_B v(B) - p(B))$;
namely, that ties are broken appropriately by the linear mapping
$\Psi$.

The second ingredient in proving the demand for good $g$ on $N$ is
close to $N'$ is the following theorem, which states that the
empirical average of $\ell(f, \cdot)$ on a sample $N$ should be close
to the expectation of $\ell$ on the distribution $\D$ for all
functions $f\in \F$, for any function
$\ell : \F \times (\V, \cG) \to \{0,1\}$.
\begin{theorem}\LABEL{thm:vc-bernstein}
  Consider a hypothesis class $\F$, and $\ell(f,x)\in \{0,1\}$. Let
  $N\sim \D$ be a sample such that $|N|=\sam$.  Then, with probability
  $1-\delta$, for all $f\in \F$,
  \[
    \ell_\D(f) - \ell_N(f) 
\ifshort\else
    \leq \sqrt{\frac{2\VC(\F)\ell_{N}(f)\ln\frac{1}{\delta}}{\sam}}
    + \frac{\VC(\F)\ln\frac{1}{\delta}}{\sam}
\fi
    \leq \frac{3\VC(\F)\ln\frac{1}{\delta}\sqrt{\ell_N(f)}}{\sqrt{n}}.
  \]
\end{theorem}
The bound on the distance depends on the size of the sample $N$ as
well as the VC dimension of the class $\F$. The proof of this theorem
and the formal proof of \cref{thm:demand-bernstein} can be found in
the appendix; the argument follows straightforwardly from
\citet{shalev2014book} (specifically, Theorem 6.11 replacing
the bound on the difference in losses for a fixed hypothesis by the bound
in Lemma B.10).

\SUBSECTION{Welfare Concentration}\LABEL{sec:welfare}

Somewhat unusually, our proof of \cref{thm:demand-bernstein} did not
go through a combinatorial shattering argument to prove a bound on VC
dimension.  In contrast, our proof that the \emph{welfare} of
Walrasian prices generalizes (\cref{thm:welfare-bernstein}) relies on
an explicit shattering argument. First, we bound the pseudo-dimension
of the class of welfare predictors $\H_V$ (\cref{thm:pseudo-welfare}).
This follows from a combinatorial shattering argument using the linear
separator for $\H$.  Then, we show that classes with small
pseudo-dimension yield strong concentration
\preflong{thm:pseudo-bernstein}. Similar to the
results for over-demand, the latter statement follows from a standard
sample complexity argument using Bernstein's inequality in place of
Hoeffding's inequality to get a multiplicative bound. We then show
that optimal welfare is well-concentrated
(\cref{lem:welfare-concentration}). Thus, if the welfare on $N'$ for
prices $p$ is close to the optimal welfare for $N$, then the welfare of
$N'$ at $p$ must be nearly optimal for $N'$.

We begin by proving a bound on the pseudo-dimension of the welfare predictors
$\H_V$.
\begin{theorem}\LABEL{thm:pseudo-welfare}
  The pseudo-dimension of $\H_V$ is $O\left(m^2\ln m\right)$ for
  general valuations $\V$ and $O\left(m\ln^3 m\right)$ for unit demand
  valuations\ifshort.\else, if $e$ is encodable.\fi
\end{theorem}
The proof of \cref{thm:pseudo-welfare} (see \thelongref{sec:omit-shattering})
relies on the following combinatorial lemma.

\begin{lemma}\LABEL{lem:labelings}
  $\H$ can induce at most ${\sam \choose m+1} \cdot 2^{m(m+1)}$ bundle
  labelings of $\sam$ sampled valuations, and at most
  ${\sam \choose m+1} \cdot m^{m+1}$ for unit demand
  buyers\ifshort.\else, if $e$ is encodable.\fi
\end{lemma}

\begin{proof}
  Let $h_\p(N) = (h_\p(v_1), \ldots, h_\p(v_\sam))$ (by abuse of
  notation) represent the labeling of a sample $N$ by $h_\p$.
  \Cref{lem:linsep} shows that $\H$ is $(\ngood+1)$-linearly
  separable. By \cref{thm:linsep}, there exists a
  $(\ngood+1)$-sized compression scheme $\Psi$ for learning the class
  $\H$.  Thus, there exists a pair of functions $(\decom, \com)$
  such that for any price vector $\p$, there exists
  $N_p \subseteq N$ with $|N_\p| = \ngood + 1$ such that
  (i) $\com(h_\p(N)) = h_\p(N_\p)$, (ii) $\decom(\com(h_\p(N))) =
  \decom(h_\p(N_\p)) = f$, and and $f(v_q) = B_q = h_p(v_q)$ for every $v_q\in
  N$.

  We will bound the number of distinct labelings $\H$ can induce on a sample
  $N = (v_1, \ldots, v_\sam)$ of size $\sam$.
  To do so, we the \emph{decompression} function
  (which depends only on the separator $\Psi$) to
  upper-bound the total number of labelings $\H$ can induce on $S$.
  
  Our approach will be to upper bound the size of the \emph{range} of
  $\decom \circ \com$ when restricted to subsamples of a
  particular population $N$ by some quantity $D$. By
  definition of a compression scheme, for every labeling of
  buyers in $N$ by bundles that can be induced by some price vector,
  there is an element of the range of $\decom \circ \com$ that
  recovers that labeling on $N$. So, we can upper bound the total number
  of distinct labelings that can be induced on $N$ by price vectors by
  upper bounding the range of $\decom \circ \com$ on $N$. The idea is
  that each \emph{compression} is defined by a subset $N'$ of at most
  $\ngood+1$ of the buyers in $N$, and a labeling of the
  buyers with bundles. So, bounding the number of such
  labeled subsets suffices to bound the range of $\decom \circ \com$.

  We prove such a bound for 
  $ D = {\sam \choose \ngood+1} \cdot
  \left(2^\ngood\right)^{\ngood+1}$.
  Fix any subset of unlabeled examples $N'\subseteq N$, where
  $|N'| = \ngood+1$. For any $x\in N'$, there are at most $2^{\ngood}$
  labelings of $x$ from $\cG$. Thus, there are at most
  $\left({2^\ngood}\right)^{\ngood+1} = 2^{\ngood(\ngood+1)}$
  labelings of all of $N'$ from $\cG$.\footnote{%
    If buyers are unit demand, we have a sharper bound: there are at most
    $\ngood$ labelings per $x\in N'$, or $\ngood^{\ngood+1}$ labelings in total.}
  Thus, $\decom\circ \com$ can output at
  most $ {\sam \choose \ngood+1} \cdot 2^{\ngood(\ngood+1)}$ distinct
  functions based on labeled subsets of $N$.
\end{proof}

\cref{thm:pseudo-welfare}, along with known results for uniform
convergence over classes with small pseudo-dimension, implies that
the welfare of any set of prices $\bp$ on a sufficiently large sample
is very close to the welfare $\bp$ induces on a fresh sample of
the same size. The next lemma shows that the \emph{optimal} welfare
allocation on the training and test samples have very similar
welfare, since the optimal welfare is a well-concentrated quantity.

\begin{lemma}\LABEL{lem:welfare-concentration}
  Recall that $\optw{N}$ denotes the optimal welfare achievable by an allocation
  of a fixed set of goods in a market to buyers with valuations
  $N = (v_1, \ldots, v_\sam)$ bounded in $[0, \maxval]$. Then, the optimal
  welfare on two samples $N, N'$ of size
  $\sam$ drawn i.i.d. from distribution $\D$ satisfies
  \[
    \Pr[|\optw{N}-\optw{N'}|\geq 4\epsilon] \leq
    \exp\left({-\frac{2\epsilon^2}{\sam \maxval^2}}\right)
  \]
  for all $\epsilon \in (0, 1)$. In particular, $|\optw{N}-\optw{N'}|\leq
  H\sqrt{n\ln\frac{1}{\delta}}$ with probability at least $1-\delta$.
\end{lemma}

The proof of \cref{lem:welfare-concentration} uses McDiarmid's
inequality, and can be found in
\thelongref{sec:omit-concentration}. We now prove the main result of this
section, which guarantees first that for any $p$, the welfare induced
on $N$ is similar to the welfare induced on $N'$. In particular,
if $p$ is the welfare-optimal pricing for $N$, then it is also
nearly welfare-optimal for $N'$.

\ifshort
\else
\begin{proof}[of \cref{thm:welfare-bernstein}]
We wish to show that with probability $1-\delta$,
\[\optw{N'} - \wel{p}{N'} \leq \alpha \optw{N'}.\]
We begin by noting that
\begin{align}
& \optw{N'} - \wel{p}{N'} \nonumber \\
& \leq \optw{N} - \wel{p}{N'} + \maxval\sqrt{\sam\ln\frac{1}{\delta}}
\tag{\cref{lem:welfare-concentration}}\\
& \leq \wel{p}{N} + 2\ngood \maxval - \wel{p}{N'}
  + \sqrt{\sam\maxval\ln\frac{1}{\delta}}
\tag{\cref{thm:gs_welfare}}\\
& \leq \rwel{p}{N} + 3\ngood \maxval - \wel{p}{N'}
  + \maxval\sqrt{\sam\ln\frac{1}{\delta}}
\tag{\cref{thm:gs_od2,lem:welfare-equiv}}\\
& \leq \rwel{p}{N} - \rwel{p}{N'}
  + \ngood\maxval\left(\max_g\sqrt{s_g\ngood\ln\frac{1}{\delta}}
                 + \ngood\ln\frac{1}{\delta}\right)
  + \maxval\sqrt{\sam\ln\frac{1}{\delta}}
  + 3\ngood \maxval \notag \\
& \leq \rwel{p}{N} - \rwel{p}{N'}
  + \ngood\maxval\left(\sqrt{n\ngood\ln\frac{1}{\delta}}
                 + \ngood\ln\frac{1}{\delta}\right)
  + \maxval\sqrt{\sam\ln\frac{1}{\delta}}
  + 3\ngood \maxval
\label{eq:five}
\end{align}
where the second to last step follows by~\cref{thm:demand-bernstein} and
\cref{lem:welfare-equiv}, and the last step by $s_g \leq
n$, otherwise, we can replace $s_g$ by $n$ in the bound in
\cref{thm:demand-bernstein}, since $\ndem{N}{p}{g}{e} \leq n+1$.

Applying \cref{thm:pseudo-bernstein} to
  $\sum_{q\in N} v_q(\canon{q}(p)) = n\cdot \frac{\sum_{q\in
      N}v_q\left(\canon{q}(p)\right)}{n}$
  and the analagous term for $N'$, we know that for all pricings $\p$,
\[
  \rwel{\p}{N} - \rwel{p}{N'}
  \leq
  \left(\maxval^{3/2}\pd(\F)\ln\frac{1}{\delta}\right)
  \left(\frac{1}{3} + \sqrt{19\rwel{p}{N}}\right) .
\]
Using \cref{thm:pseudo-welfare}, this reduces to

\begin{align}
& \rwel{\p}{N} - \rwel{p}{N'} \notag \\
& \leq \left(\maxval^{3/2}\ngood^2\ln\ngood\ln\frac{1}{\delta}\right)
    \left(\frac{1}{3} + \sqrt{19\rwel{\p}{N}}\right)
\tag{\cref{thm:pseudo-welfare}'s bound on $\pd(\H_\V)$} \\
& \leq \left(\maxval^{3/2}\ngood^2\ln\ngood\ln\frac{1}{\delta}\right)
    \left(2 \sqrt{19\rwel{\p}{N}}\right)
\notag \\
& \leq \left(\maxval^{3/2}\ngood^2\ln\ngood\ln\frac{1}{\delta}\right)
    \left(2 \sqrt{19\wel{\p}{N} + \ngood\maxval}\right)
\tag{\cref{thm:gs_od2,lem:welfare-equiv}} \\
& \leq \left(\maxval^{3/2}\ngood^2\ln\ngood\ln\frac{1}{\delta}\right)
    \left(2 \sqrt{19\optw{N}  + \ngood\maxval}\right)
\tag{$\wel{\p}{N}$ is feasible} \\
& \leq \left(\maxval^{3/2}\ngood^2\ln\ngood\ln\frac{1}{\delta}\right)
    \left(2 \sqrt{19\optw{N'} + \maxval\sqrt{\sam\ln\frac{1}{\delta}}
    + \ngood\maxval}\right) .
\tag{\cref{lem:welfare-concentration}}
\end{align}
Combining with \cref{eq:five}, we have
\begin{align*}
\optw{N'} - \wel{p}{N'}
& \leq \left(\maxval^{3/2}\ngood^2\ln\ngood\ln\frac{1}{\delta}\right)
    \left(2 \sqrt{19\optw{N'} + \maxval\sqrt{\sam\ln\frac{1}{\delta}}  + \ngood\maxval}\right)\\
& + \ngood\maxval\left(\sqrt{n\ngood\ln\frac{1}{\delta}} + \ngood\ln\frac{1}{\delta}\right)
  + \maxval\sqrt{\sam\ln\frac{1}{\delta}} + 3\ngood \maxval.
\end{align*}
If we wish for this to be at most $\alpha\optw{N'}$, it suffices for
\[
  \optw{N'} = \Omega\left(
    \frac
    {\maxval^{3}\sam^{0.5}\ngood^4\ln^{2}\left(\ngood\right)\ln^{2}\frac{1}{\delta}}
    {\alpha^2}
  \right).
\]
\jh{I got something different, but maybe bounding things
  differently. I got:
  ${\maxval^{3}\sam^1\ngood^4\ln^{2} \ngood\ln^{2}\frac{1}{\delta}} /
  {\alpha^2}.$
  Not sure which parameters are important though, is the important
  thing to preserve the $\sqrt{n}$ dependence?}  With probability
$1-\delta$, we know that
$\optw{N} - \optw{N'} \leq \maxval\sqrt{\sam\ln\frac{1}{\delta}}$, so
this is satisfied by $\optw{N'}$ whenever it is satisfied by
$\optw{N}$ with probability $1-\delta$, which holds by assumption.
\end{proof}

\begin{proof}[sketch of~\cref{thm:unit-welfare-bernstein}]
  This proof follows the identical calculation
  from~\cref{thm:welfare-bernstein}, replacing the bound on
  pseudo-dimension by $\ngood\ln^3\ngood$, as is implied
  by~\cref{thm:pseudo-welfare} when buyers are unit demand.
\end{proof}
\fi

Finally, we sketch the proof of~\cref{thm:welfare-bernstein}. The
optimal welfare for the two markets must be close, by
\cref{lem:welfare-concentration}. The welfare of the pricing $p$
output for the first market is nearly optimal even when over-demand is
resolved adversarially, since there is over-demand of at most $1$ for
each good by assumption. By \cref{thm:pseudo-welfare},
the pseudo-dimension of the class is at most $\ngood^2\ln^2\ngood$ for
general valuations and at most $\ngood \ln^3 \ngood$ for unit-demand
valuations. A pseudo-dimension analogue
of \cref{thm:vc-bernstein} implies that the welfare of applying $p$ to
the second market will be close to the welfare of applying $p$ to the
initial market.
 
\ifshort\else
\SUBSECTION{Lower bounds for learning}\LABEL{sec:learning-lbs}

We now show that the VC dimension of $\H_g$, the class of bundle
predictors for a fixed good $g$ (at a pricing $p$) is at least
$\ngood$, implying our bound on the VC dimension of $\H_g$ from the
previous section is tight up to constant factors.

\begin{theorem}\LABEL{thm:unit-lb}
  Let $V$ be the set of unit demand valuations. Then
  $\VC(\H_g)\geq \ngood$ and $\pd(\H_V)\geq \ngood$, for $\H_g, \H_V$
  over valuations $\V$.
\end{theorem}

\begin{proof}
  Fix a particular good $g$.  We will show a set
  $N = (v_1, \ldots v_\sam)$ of unit demand buyer valuations which can
  be shattered by both $\H_g$ and $\H_V$. Suppose $\sam=\ngood$ and
  label buyers such that $q\in N$ also corresponds to some
  $q\in[m]$. Then, for all $q\neq g$, define $v_q(\{q\}) = 2$,
  $v_q(\{g\}) = 1$, and $v_q(\{g'\}) = 0$ for all goods
  $g'\notin\{q,g\}$. Let $v_g(\{g\} = \frac{1}{2}$ and
  $v_g(\{g'\}) = 0$. The intuition for this construction is that buyer
  $q$ will buy good $g$ if the price of their ``individual'' good $q$
  is sufficiently larger than the price of $g$ ($p_q > p_g + 1$), and
  buyer $g$ will buy good $g$ if the price is less than her value
  ($p_g \leq v_g$).  Consider an arbitrary set $N'\subseteq N$.  We
  will first show that $\H_g$ can label $N'$ positive and
  $N\setminus N'$ negative.  We define prices $p$ as follows.

  If $g\in N'$, then let $\p_g = 0$ and $\p_{g'} = 1 + 2\epsilon$ if
  $g' \in N'$ and $\p_{g'} = \epsilon$ if $g'\notin N'$. Then, for
  buyer $g$, $v_g(\{g\}) = \frac{1}{2} > \p_g = 0$, and
  $v_g(\{g'\}) = 0 < p_{g'}$ for all other $g'$; thus, $\{g\}$ is the
  unique bundle in $\dem{g}(p)$ (and so $h_{g,p}(v_g) = 1$). Then, for
  $g'\in N'$,
  $v_{g'}(g') - \p_{g'} = 2 - (1+2\epsilon) < 1 = v_{g'}(g) - p_g$,
  thus $\{g\}$ is the unique bundle in $\dem{g'}(p)$ (so,
  $h_{g, p}(v_{g'}) = 1$). For any $g'\notin N'$,
  $v_{g'}(\{g'\} - p_{g'} = 2 - \epsilon > 1 = v_{g'}(\{g\}) - p_g$;
  thus, $\{g'\}$ is the unique bundle in $\dem{g'}(p)$ and so
  $h_{g, p}(v_{g'}) = 0$.

  If, on the other hand, $g\notin N'$, let
  $p_g = \frac{1}{2}+ \epsilon$, $p_{g'} = \tfrac{3}{2} + 2\epsilon$
  if $g'\in N'$ and $p_{g'} = \epsilon$ if $g'\notin N'$. Then, since
  $v_g(\{g\}) = \frac{1}{2} < \frac{1}{2} + \epsilon = p_g$,
  $h_{g,p}(v_g) = 0$. For $g'\in N'$,
  $v_{g'}(\{g\}) - p_g = 1 - \frac{1}{2} - \epsilon > 2 - \frac{3}{2}
  - 2\epsilon = v_{g'}(\{g'\}) - p_{g'}$,
  thus $\{g\}$ is the unique bundle in $\dem{g'}(p)$ and so
  $h_{g, p}(v_{g'}) = 1$. For $g'\notin N'$,
  $v_{g'}(\{g\}) - p_g = 1 - \frac{1}{2} - \epsilon < 2 -\epsilon =
  v_{g'}(\{g'\}) - p_{g'}$,
  so $\{g'\}$ is the unique bundle in $\dem{g'}(p)$, and so
  $h_{g, p}(v_{g'}) = 0$.

  Thus, for arbitrary $N'\subseteq N$, we can choose prices $p$ such
  that $h_{g,p} (v) = 1$ if and only if $v\in N'$, so we have shown
  how to shatter $N$ with $\H_g$.

  The same $N$ is shatterable by $\H_V$, as we now argue.  Consider
  $(r_1, \ldots, r_\sam)$ a set of targets, with $r_g = \frac{1}{2}$
  and $r_{g'} = \frac{3}{2}$ for all $g'\neq g$. Thus, the target is
  hit for $v_g$ only when buyer $g$ buys good $g$ (and for $v_{g'}$
  when buyer $g'$ buys good $g'$).  For a set of buyers
  $N'\subseteq N$, set the prices $p$ as above for $\H_g$, but swap
  the definition of $p_{g'}$ for $g'\in N'$ with that of $p_{g'}$ for
  $g'\notin N'$. In the previous definition of $p$, whenever
  $g'\in N'$, buyer $g'$ bought $g$ and not $g'$, and when
  $g'\notin N'$ buyer $g'$ bought $g'$ and not $g$. So, when these are
  swapped, $g'$ will buy $g'$ if and only if $g'\in N'$ (thus, hitting
  the welfare target of $r_{g'}$ if and only if $g'\in N'$). For buyer
  $g$, when $g\in N'$, the prices above cause buyer $g$ to buy good
  $g$, hitting welfare target $r_g$ (and, when $g\notin N'$, $g$
  didn't buy $g$, therefore missing welfare target $r_g$). Thus, for
  this modified definition of $p$, all buyers $q\in N'$ will have
  $f_p(v_q) \geq r_q$ and for all $q\notin N'$, $f_p(v_q) <
  r_q$. Thus, we can shatter $N$ according to these targets.
\end{proof}
\fi

%% file: appendix-multiple.tex
\section{Valuations over bundles of copies versus bundles of goods} \LABEL{sec:copies-appendix}

We make clear in this section that although the market that buyers face includes the copies $s_g$ of goods $g \in [m]$, we can instead focus on valuations defined on the market of just the goods $[m]$ when buyers demand at most one copy of each good.  We closely follow the work of \citet{ST15} to work with multiple copy valuation functions, which are functions on integer lattice points.  We define the set of all allocations in the multiple copy market (where each good $g \in [m]$ has $s_g$ many copies) as
$$
X = \{ 0,1, ,\cdots, s_1\} \times \cdots \times \{ 0,1,\cdots, s_m\}.
$$
Thus each buyer $q \in N$ has a valuation $v_q: X \to [0,\maxval]$.  As we do in the main body of the paper, we assume that the valuations are monotone.  We will denote $\vx \in X$ where
$$\vx = ( \vx(g): g \in [m] ).$$
  Note that buyers demand at most one copy of each good, so we define the projection $\pi: X \to \{0,1 \}^m$ as
$$
\pi(\vx) = ( \min\{\vx(1),1\}, \min\{\vx(2),1 \}, \cdots, \min\{\vx(m),1\} )
$$
The set of valuations we consider can be written as
$$
V^X = \{v: X \to [0,H] : v(\vx) = v(\pi(\vx)) \}.
$$
We now define the demand correspondence $\cD^X(\bp;v_q)$ and Walrasian equilibrium in this setting where buyers have valuations $v_q \in V^X$ for each $q \in N$.
\begin{definition}[Demand Correspondence]
Let buyer $q$ have valuation $v_q: X \to [0,H]$.  We then define the demand correspondence for any price $\bp$ as
$$
\cD^X(\bp;v_q) = \myargmax_{\vx \in X} \{v_q(\vx) - \bp^T \vx \}.
$$
\end{definition}

\begin{definition}[WE-X]
A Walrasian Equilibrium for buyers with valuations $v_q: X \to [0,\maxval]$ is a tuple $(\bp,(\vx_q: q \in N))$ where
\begin{itemize}
\item $\vx_q\in\cD^X(\bp;v_q)$
\item $\sum_{q \in N} \vx_q(g) \leq s_g$ for any $g \in [m]$, and if the inequality is strict, then the price for that good is zero.
\end{itemize}
\end{definition}

We will write $\chi_S \in \{ 0,1\}^m$ as the characteristic vector of $S\subseteq [m]$.  Note that if $(\bp,\mu)$ is a WE as we defined in \cref{def:WE}, then when we write the characteristic function $\vx_q = \chi_{\mu_q}$ for each $q \in N$ we will have $(\bp,(\vx_q: q \in N) )$ is a WE-X as we defined above.  Thus, we only need to focus on valuations defined over the boolean cube $\{ 0,1\}^m$ or equivalently over the subsets of $[m]$.  

%% file: appendix-genericity.tex
\section{Non-minimal Walrasian prices and genericity}\LABEL{sec:nonmin}
We show here that even in a market that contains only generic unit demand buyers, the over-demand at a non-minimal Walrasian price vector can be high.  This is in contrast to our result in \cref{thm:unit_od}, in which we showed that the over-demand can be at most 1 at \emph{minimal} Walrasian prices when buyers have generic valuations.
\begin{lemma}
  There exists generic valuations such that the over-demand
  $\max_{g \in [m]} OD_\bp(g,\mu) = n-1$ for some Walrasian
  equilibrium $(\bp,\mu)$ where $\bp$ is not minimal.
\end{lemma}
\begin{proof}
  We consider a slight variant of valuations given in the proof of \cref{lem:bad1}.
  Consider $n$ unit demand agents $N = [n]$ with generic valuations and good
  types $[n]$ where each type has a single copy.  We assume that each
  buyer $q$'s generic valuation satisfies the following for some special good $g \in [m]$
$$
\val{h}{q} < \val{g}{q}  < \val{q}{q} \qquad \forall h \neq q,g
$$
Consider the prices $p_q = \val{q}{q} - \val{g}{q}$ for $q \in [n]$.  This price vector still satisfies the maximum allocation where every buyer $q$ gets good $q$, but it now has every buyer $q$ have an indifference between good $g$ and $q$.  Hence, the over-demand for good $g$ at these prices is $n-1$.  We know that $\bp$ is not minimal due to \cref{thm:unit_od}.
\end{proof}

\section{Minimal demand bundles form a matroid basis} \LABEL{sec:omit_gs}
To establish this fact (\cref{lem:bases_orig}), we look to \citet{BLN13}. They
prove essentially the same lemma, but only when valuations are rational numbers.
Since we require genericity, it may not be reasonable to also assume that
valuations are rational. Fortunately, it is straightforward to modify the proof
by \citet{BLN13} for the general case, which we do here.  We begin by showing
the following claim which will help in our analysis.
\begin{claim}
If $B, B' \in \cD^*(\bp)$ then $|B| = |B'|$.
\end{claim}
\begin{proof}
Let $D_1, D_2 \in \cD^*(\bp)$.  Let $\bp' = \bp$, except $p'(g) = \infty$ for $g \notin D_1 \cup D_2$.  Note that $D_1, D_2 \in \cD^*(\bp')$.  We define a small quantity $\epsilon$ as
$$
\epsilon = \min_{\substack{B,B' \subseteq [m] \\ |B \setminus B'|, |B' \setminus B| \leq 1 }}
\left\{ \util{B}{} - \util{B'}{} : \util{B}{} - \util{B'}{} >0\right\} .
$$
Let $g_2 \in D_2 \setminus D_1$ and define $\bp_2'= \bp'$ except $p_2'(g_2) = p'(g_2)+ \epsilon$.  We then know that $D_2 \notin \cD(\bp_2')$ because $u(D_1;p_2') =u(D_1;p') = u(D_2;p') > u(D_2;p'_2)$.  Due to the single improvement property of GS valuations, we know that there exists $D_3$ such that both $|D_3\setminus D_2|, |D_2 \setminus D_3|\leq 1$ and $u(D_3;p_2') > u(D_2;p_2')$.  We then have
\begin{align}
u(D_3;p_2') & > u(D_2;p_2') \implies u(D_3;p) - \epsilon \1\{g_2 \in D_3\} > u(D_2;p) - \epsilon \nonumber\\
& \implies \epsilon - \epsilon \1\{g_2 \in D_3\}  > u(D_2;p) - u(D_3;p) .
\LABEL{eq:contradiction}
\end{align}
We know that $D_2 \in \cD(\bp)$ so that $u(D_2;p) - u(D_3;p) \geq 0$, which
tells us that $g_2 \notin D_3$.  Further, we know that if $u(D_2;p) - u(D_3;p)
>0$ then $u(D_2;p) - u(D_3;p) \geq  \epsilon$.  However, \cref{eq:contradiction} tells us that $u(D_2;p) - u(D_3;p) < \epsilon$, thus a contradiction.  We then have $u(D_3;p) = u(D_2;p)$.  Hence, we have shown that $g_2 \in D_3$, $D_3 \in \cD(\bp)$ and $D_3$ is not contained in $D_2$.

Thus, $\exists g_1 \notin D_2$ such that $g_1 \in D_3 \subseteq D_1 \cup D_2$ which implies that $g_1 \in D_1 \setminus D_2$.    Thus, we can say that $D_3 = D_2 \cup g_1 \setminus g_2$.  We continue with induction on $|D_2 \setminus D_1|$ to conclude that $|D_2| = |D_1|$.
\end{proof}
\begin{proof}[Proof of \cref{lem:bases_orig}]
We now let $\hat \cD(\bp) = \{S \in \cD(\bp): |S|\leq |T|, \forall T \in \cD(\bp) \}$, which we know forms the bases of some matroid from \citet{GS99a}.  We first show that $\cD^*(\bp) \subseteq \hat \cD(\bp)$.  If $D \in \cD^*(\bp)$ then we know that $|D| = |D^*|$ for all $D^* \in \cD^*(\bp)$ by the previous claim.  Hence, $|D| \leq |B|$ for all $B \in \cD(\bp)$, i.e. $D \in \hat \cD(\bp)$.  

We now show that $\hat \cD(\bp) \subseteq \cD^*(\bp)$.  Assume that $D \in \hat \cD(\bp)$, but $D \notin \cD^*(\bp)$.  We then know that there exists a $D^* \subsetneq D$ where $D^* \in \cD^*(\bp)$.  This must mean that $|D^*| < |D|$, but this contradicts the fact that $D \in \hat\cD(\bp)$, i.e. $|D| \leq |B|$ for every $B \in \cD(\bp)$.  Thus, we have a contradiction, so that $D \in \cD^*(\bp)$.  Further, we have shown that $\cD^*(\bp) = \hat\cD(\bp)$.  

\end{proof}

\section{Constructing genericity via perturbation} \LABEL{sec:perturb}
When working with matchings (\cref{sec:matchings}), we assumed that valuations
are generic in the sense of
\cref{def:generic-unit}. Instead, we can imagine perturbing the valuations
slightly to achieving genericity. While ``any'' continuous perturbation should
achieve genericity, we can also give a more concrete perturbation procedure.
While it is difficult to guarantee a definition like \cref{def:generic-unit}, we
can limit over-demand using a relaxed version of genericity.

Let's suppose that bidders start with valuations $\val{j}{i}$ (possibly
non-generic) and minimal Walrasian prices $\bp$. Each bidder will perturb their
valuation by selecting a uniformly random element of $\cP \subseteq \R$, a
(finite) set of perturbations. We will continue to assume that the maximum
matching $\mu$ is unique.  We will write $\Delta$ for the smallest positive
difference social welfare between two allocations---possibly not allocating all
goods, or infeasible For instance, this assumption holds in the typical case
where all valuations are given with finite precision, when we can treat all
valuations as integer multiples of $\Delta$.

Our goal is to show that with high probability, the perturbed values
$\hat{\val{j}{i}}$ and corresponding minimal prices $\hat{\bp}$ have limited
over-demand when bidders select an arbitrary most-demanded good. Our analysis
will begin with the original valuations $\val{j}{i}$ and prices $\bp$. As we have done
throughout, we consider the swap graph. By identical reasoning as before, the
graph is acyclic and the nodes can be partially ordered. We will consider the
nodes according to this order, using the principle of deferred decision to
successively perturb the valuations. As we go, we will construct a modified set
of prices $\hat{\bp}'$ that guarantees each good is over-demanded by at most $1$.
Finally, we will argue that $\hat{\bp}'$ is indeed the minimal Walrasian prices for
$\hat{\val{j}{i}}$, tying the knot.

Let the goods be numbered $1, \dots, k$ according to the partial order in
the original swap graph $G$; for convenience, we will also label bidders with
the same index according to their matched good in the original matching $\mu$.
We will construct the modified prices $\hat{\bp}'$ inductively. We begin by
setting $\hat{p}_1 = 0$, and sampling bidder $1$'s perturbation for good $1$
but leaving the other perturbations unsampled, by the principle of deferred
decisions.

For the inductive case, we have prices $\hat{p}$ for goods in $[j]$ and sampled
bidder valuations by bidders in $[j]$ for goods in $[j]$. To set the price of
good $j + 1$, we first have bidder $j + 1$ sample valuations for all goods in
$[j + 1]$. Now, we need to specify the valuations of bidders in $[j]$ for the
new good $j + 1$. We consider a set of candidate price functions $r_1, \dots,
r_j$, where
\[
  r_l(v) := v - \hat{\val{l}{l}} + \sum_{(a, b) \in C} \hat{\val{b}{a}} -
  \hat{\val{a}{a}}
\]
and $C$ is a directed path from a source to $l$ in $\hat{G}$ (recall that
$\hat{G}$ is acyclic).  Roughly, $r_l$ maps the valuation of bidder $l$ for $j +
1$ to a price for $l$ at which bidder $l$ becomes indifferent between $l$ and $j
+ 1$. Note that by induction, $l \in [j]$ and we have already fixed all
$\hat{v}$ when defining $r_l$. We now sample $\hat{\val{j+1}{l}}$ for $l \in
[j]$, and define
\[
  p_{j + 1}' := \max_l (0, r_l(\hat{\val{j + 1}{l}})) .
\]
That is, if all $r_l$ are negative, we define $p_{j + 1}' := 0$.

The key point is that in the new swap graph $\hat{G}$, there is an edge from $l$ to
$j + 1$ exactly when $l$ is binding (i.e., in the argmax) in the definition of
$p_{j + 1}'$. Since $r_l(v)$ is a linear function, if the number of possible
perturbations $|\cP|$ is large enough, we will be unlikely to have a collision.
In other words, with high probability the in-degree of $j + 1$ will either be
$1$, or zero with price $0$.

First, we show that the new swap graph $\hat{G}$ is a subgraph of the old swap
graph $\hat{G}$, and that the prices $\bp'$ are close to $\bp$.

\begin{lemma}
  For each good $j$, the following two statements hold:
  \begin{enumerate}
    \item If $C$ is a path in $\hat{G}$ ending at $j$, then $C$ is also a path
      in $G$.
    \item Suppose that all perturbations in $\cP$ are bounded by $\Delta/2k$.
      Then, $|p_j' - p_j| < \Delta j/k$ for each good $j$.
  \end{enumerate}
\end{lemma}
\begin{proof}
  We prove both points simultaneously by induction on $j$. The base case $j = 1$
  is clear: the good $1$ is a source node in both graphs, and has price $0$.

  For the inductive case, we prove the first point first. If $j$ is a source
  node in $\hat{G}$ then the first point is trivial. Otherwise, there is a path
  $C$ ending at $j$ in $\hat{G}$. Suppose the node before $j$ in $C$ is $l$. By
  induction, the segment of $C$ ending at $l$ is a path in $G$. So, we just need
  to show that there is an edge $(l, j)$ in $G$.

  By construction of $p'$, we know
  \[
    p_j' = \hat{\val{j}{l}} - \hat{\val{l}{l}} + p_l' \geq 0 .
  \]
  If there is no edge $(l, j)$, then
  \[
    (\val{l}{l} - p_l) - (\val{j}{l} + p_j) > 0
  \]
  We can write the price as a sum of differences of valuations
  (\cref{lem:lincomb}), and collecting the positive and negative parts, the
  above equation shows that the difference in welfare of two (possibly partial)
  allocations is non-negative. By assumption, the difference is at least
  $\Delta > 0$, and hence
  \[
    (\val{l}{l} - p_l) - (\val{j}{l} + p_j) > \Delta
    \quad \text{so} \quad
    \val{j}{l} - \val{l}{l} + p_l < - \Delta - p_j \leq - \Delta .
  \]
  This is a contradiction: by induction, $|p_l' - p_l| \leq \Delta l/k$, and
  $\val{l}{l}, \val{l}{j}$ differ from $\hat{\val{l}{l}}, \hat{\val{j}{l}}$ by at most
  $\Delta/2k$ at most, and $l \leq j - 1 \leq k - 1$.

  Now, for the second point. Suppose that $j$ is a source node in $\hat{G}$.
  So, $p_j' = 0$. Suppose that there is a non-empty path $C$ from a source
  node to $j$ in $G$ (if not, then $p_j = p_j' = 0$), and suppose $C$ hits $l$
  before landing at $j$, so
  \[
    p_j = \val{j}{l} - \val{l}{l} + p_l .
  \]
  By induction, we know $|p_l' - p_l| < \Delta l /k$.  Furthermore, we know
  \[
    r_l(\hat{\val{j}{l}}) := \hat{\val{j}{l}} - \hat{\val{l}{l}} + p_l' < 0 .
  \]
  Since $\hat{\val{j}{l}}$ and $\hat{\val{l}{l}}$ are at most $\Delta/2k$ away from
  $\val{j}{l}$ and $\val{l}{l}$ respectively, we know $p_j \leq \Delta/2k +
  \Delta/2k + \Delta l/k \leq \Delta j/k$ by induction, and since $l \leq j -
  1$.

  Otherwise, $j$ is not a source node in $\hat{G}$. Suppose that $C$ is the
  shortest path from a source node to $j$ in $\hat{G}$, with $l$ the node before
  $j$ in $C$. We have:
  \[
    p_j' = r_l(\hat{\val{j}{l}}) = \hat{\val{j}{l}} - \hat{\val{l}{l}} + p_l' .
  \]
  We know that $C$ is also a path from a source node to $j$ in $G$, so
  \[
    p_j = \val{j}{l} - \val{l}{l} + p_l .
  \]
  The claim follows by induction.
\end{proof}

Finally, we can show that $\bp'$ is indeed the minimal Walrasian prices for the
perturbed valuations.

\begin{lemma}
  Suppose that all perturbations in $\cP$ are bounded by $\Delta/2k$.  The
  prices $\bp'$ are minimal Walrasian prices for valuations $\hat{v}$.
\end{lemma}
\begin{proof}
  The new prices $\bp'$ support $\mu$, a feasible allocation. Consider any other
  feasible allocation $\mu'$. We have changed each bidder's valuation by at most
  $\Delta/2k$, so the welfare of $\mu$ and $\mu'$ each change by at most
  $\Delta/2$ in the new market. Since the original difference in social welfare
  between $\mu$ and $\mu'$ is at least $\Delta$, $\mu$ must remain the maximum
  matching in the perturbed market.

  The only thing we need to check is that nodes with in-degree $0$ (i.e., goods
  that are not over-demanded) have price $0$. But this is also clear, from the
  construction of $\bp'$.
\end{proof}

Like before, we can bound the worst-case over-demand by the in-degree of
any node in $\hat{G}$.

\begin{theorem}
  Suppose that all perturbations in $\cP$ are bounded by $\Delta/2k$. If $|\cP|
  = \Omega(n^2 k/\beta)$, then with probability at least $1 - \beta$, each node
  in $\hat{G}$ has in-degree at most $1$.
\end{theorem}
\begin{proof}
  Note that setting prices for good $j$ can only add edges to $j$, and leaves
  the rest of the swap graph $\hat{G}$ unchanged. So, consider a single good
  $j$, and consider any pair of bidders $(r, s)$. Fixing $\hat{\val{j}{r}}$, there
  is at most one value of $\hat{\val{j}{s}}$ that collides (possibly leading to
  two arcs into $j$), so the probability of any pair colliding is $1/|\cP|$.
  Taking a union bound over the $O(n^2)$ pairs and $k$ goods, the collision
  probability in the two graph $\hat{G}$ is at most $O(n^2 k/|\cP|) = O(\beta)$.
\end{proof}

\section{Proof of~\cref{lem:welfare-equiv}}\LABEL{app:welfare-equiv}

\begin{proof}[Proof of~\cref{lem:welfare-equiv}]
  We first prove the first condition. We know by assumption that
  $|\{q : g\in B_q\}| \leq s_g + d$ (at most $s_g + d$ buyers have $g$
  in their bundle).  Let $q_g$ be the label of the
  $(s_g + 1)$st  (according to the ordering $\sigma$) buyer with $g$ in their bundle, i.e. $g \in B_{q_g}$ and
  $$
  \sum_{q' : \sigma(q') < \sigma(q_g) } \1\{g \in B_{q'} \} = s_g.
  $$
  Then, let $S_g = \{q : g\in B_q \textrm{ and } \sigma(q) \geq \sigma(q_g)\}$ be the
  set of buyers whose bundle contains $g$ and who, according to
  the ordering $\sigma$, are after the $s_g$th buyer with $g\in B_q$.
  For each such buyer $q$ let
  $\hat{B}_q = \emptyset$. Otherwise, let $\hat{B}_q = B_q$.
  Then,

\begin{align*}
  \rwel{B_1, \ldots, B_\sam}{N} & = \sum_{q } v_q(B_q) 
                                = \sum_{q\notin \bigcup_g S_g} v_q(B_q) + \sum_{q\in \bigcup_g S_g} v_q(B_q)\\
                               & = \sum_{q\notin \bigcup_g S_g} v_q(\hat{B}_q) + \sum_{q\in \bigcup_g S_g} v_q(B_q)
                                \leq \sum_{q\notin \bigcup_g S_g} v_q(\hat{B}_q) + \sum_g\sum_{q\in S_g} v_q(B_q)\\
                               & \leq \sum_{q\notin \bigcup_g S_g} v_q(\hat{B}_q) + \sum_g\sum_{q\in S_g} \maxval
                                \leq \sum_{q\notin \bigcup_g S_g} v_q(\hat{B}_q) +  d\cdot m \cdot \maxval\\
                               & \leq \sum_{q} v_q(\hat{B}_q) +  d\cdot \ngood\cdot \maxval
                                = \rwel{\hat{B}_1,\ldots, \hat{B}_\sam}{N} +  d\cdot \ngood\cdot \maxval
\end{align*}
where the third equality follows from the definition of $\hat{B}_q$, the
first inequality from the fact that $q\in \cup_g S_g$ implies $q\in S_g$ for
some $g$, the second inequality from the bound on all the valuations, the
third inequality from the assumption on $|S_g| \leq d$ and there being $m$ goods, and the last inequality from the definition of $\hat{B}_q$.

\newcommand{\lp}{\underline{p}}

We now prove the second statement. Let $N_g$ be the set of sampled bidders for
that demand $g$, and consider a fixed buyer $q$. For
each good $g\in B_q$, let
$p_g = \Prob{N_g}{q\in N_g} \geq \frac{s_g}{s_g + d}$ and
$\lp = \min_g p_g$.  We will now prove, for an arbitrary partition $(D_1, D_2)$
of $B_q$, that
\[
  \Ex{N_1, \ldots, N_\ngood}{v_q(\hat{B}_q\setminus D_2)} \geq \lp
\cdot v_q(B_q\setminus D_2) .
\]
Our claim follows when $D_2 = \emptyset$ after
summing over all buyers $q$.

We proceed by induction on $|D_1|$.
If $|D_1| = 0$, then $D_2 = B_q$ and
$B_q \setminus D_2 = \hat{B}_q \setminus D_2 = \emptyset$, so
\[\Ex{N_1, \ldots, N_\ngood}{v_q(\hat{B}_q\setminus D_2)} = 0 = \lp
\cdot 0 = \lp\cdot v_q(\emptyset).\]
Now, assume for all partitions $(D'_1, D'_2)$ of $B_q$ with $|D'_1| < t$, we
have
\[\Ex{N_1, \ldots, N_\ngood}{v_q(\hat{B}_q\setminus D'_2)} \geq \lp
\cdot v_q(B_q\setminus D'_2).\]
We wish to show that for a partition $(D_1, D_2)$ of $B_q$ with
$|D_1| = t\geq 1$, we have
\[\Ex{N_1, \ldots, N_\ngood}{v_q(\hat{B}_q\setminus D_2)} \geq \lp
\cdot v_q(B_q\setminus D_2).\]
For notational cleanliness, let $R$ be the random variable $\hat{B}_q \setminus
D_2$.  For any good $g\in B_q\setminus D_2$ and bidder $q$, we have

\begin{align*}
\Ex{N_1, \ldots, N_\ngood}{v_q(R)} & = \prob{q \in N_g} \cdot \ex{v_q(R ) \mid q \in N_g}
  + \prob{q\notin N_g} \cdot \ex{v_q(R) \mid q \notin N_g}\\
& = p_g \cdot \ex{v_q(R) \mid g\in R}+ (1-p_g)\cdot \ex{v_q(R) \mid g\notin R}\\
& = p_g \cdot \ex{v_q(R) \mid g\in R}+ (1-p_g)\cdot \ex{v_q(R \setminus \{g\})}\\
& = p_g \cdot \ex{v_q(R \setminus \{g\} \cup \{g\})}+ (1-p_g)\cdot \ex{v_q(R \setminus \{g\})}\\
& = p_g \cdot \ex{v_q(R \setminus \{g\})
  + \left(v_q(R \setminus \{g\} \cup \{g\}) - v_q(R \setminus \{g\})\right) }
  + (1-p_g) \cdot \ex{v_q(R \setminus \{g\})}\\
& = p_g \cdot \ex{ \left(v_q(R \setminus \{g\} \cup \{g\}) - v_q(R \setminus \{g\})\right) }
  + \ex{v_q(R \setminus \{g\})}\\
& \geq p_g \cdot \ex{ \left(v_q(D_1 \setminus \{g\} \cup \{g\})
  - v_q(D_1 \setminus \{g\})\right) }
  + \ex{v_q(R) \setminus \{g\})}\\
& = p_g \cdot \ex{ \left(v_q(D_1) - v_q(D_1 \setminus \{g\})\right)} + \ex{v_q(R \setminus \{g\})}\\
& = p_g \cdot \left(v_q(D_1) - v_q(D_1 \setminus \{g\})\right) + \ex{v_q(R \setminus \{g\})}\\
& = p_g \cdot \left(v_q(D_1) - v_q(D_1 \setminus \{g\})\right) + \ex{v_q(\hat{B}_q \setminus D_2 \setminus \{g\})}\\
& = p_g \cdot \left(v_q(D_1) - v_q(D_1 \setminus \{g\})\right) + \ex{v_q(\hat{B}_q \setminus (D_2\cup \{g\}))}\\
& \geq p_g \cdot \left(v_q(D_1) - v_q(D_1 \setminus \{g\})\right) + \lp\cdot v_q(B_q \setminus (D_2\cup \{g\}))\\
& = p_g \cdot \left(v_q(D_1) - v_q(D_1 \setminus \{g\})\right) + \lp\cdot v_q(D_1 \setminus \{g\}))\\
& \geq \lp \cdot \left(\left(v_q(D_1) - v_q(D_1 \setminus \{g\})\right) + v_q(D_1 \setminus \{g\})\right)\\
& = \lp \cdot v_q(D_1) = \lp \cdot v_q(B_q \setminus D_2)
\end{align*}
where the first inequality follows from the subadditivity of $v_q$ and $R =
\hat{B}_q\setminus D_2\subseteq B_q\setminus D_2 = D_1$; the following equality
follows from $g\in D_1 = B_q\setminus D_2$; the second to last inequality
follows from our inductive hypothesis (since $g\in B_1 \setminus D_2$, and $(D_1
\setminus \{ g \}, D_2 \cup \{g\})$ form a partition satisfying the induction
criterion); the final inequality follows from $\lp\leq p_g$ for all $g$.
\end{proof}

%% file: appendix-pricedim.tex
\section{Omitted Learning Proofs}\LABEL{sec:omit-learning}

\subsection{Proof of \cref{lem:linsep}}\LABEL{sec:omit-linsep}

The high-level idea for the proof of \cref{lem:linsep} is as
follows. In order to show $\H$ is $\ngood+1$-linearly separable, we
must define two things. First, we need a mapping
$\Psi : \V \times \cG \to \R^{\ngood+1}$ , where $\Psi(v, B)_g$ will
encode whether or not $g\in B$ for each good $g\in [\ngood]$, and
$\Psi(v, B)_{\ngood+1} = v(B)$ encodes the agent's value for a bundle
$B$), and (2), for each price vector $\p$ a weight vector
$w^\p \in \R^{\ngood+1}$ where $w^\p_g = -\p_g$ encodes the
\emph{cost} of good $g\in[\ngood]$ at these prices, while
$w^\p_{\ngood+1} = 1$.  These will have the property that their dot
product encodes the utility of an agent $v$ buying a bundle $B$ at
prices $\p$: that is, $\Psi(v, B) \cdot w^\p = v(B) - p(B)$. Thus, as
desired, a utility-maximizing bundle $B^*$ will have the property that
$B^* \in \argmax_{B} \Psi(v,B) \cdot w^\p$.  

Unfortunately, two related problems remain with this
formulation. First, the statement of \cref{thm:linsep} assumes
this argmax is unique.\footnote{If the argmax
  is not unique, prediction is not even well defined.}  Second, we
want to assume that amongst the demanded bundles, agents with
valuation $v_q$ break ties according to the encodable tie-breaking
rule $e = (L_{v_q}, y)$.  Fortunately, we can slightly perturb $\Psi$
and the set of $w^p$s such that they encode the tie-breaking rule $e$,
causing $e(\dem{q}(p))$ to be the unique bundle which maximizes
$\Psi(v, B) \cdot w^\p $ over $B$, solving both problems at once. 
It will be useful to have
notation for the set of bundles which maximize the dot product between
these quantities, so we define
$M(v, p) = \argmax_{B} \hat{\Psi}(v, B) \cdot \hat{w}^p$.

We now define the perturbed versions of $\Psi$ and $w^p$. Recall that
$y\in \R^\ngood$, so $y_g$ will refer to the $g$th coordinate of this
vector.  Let
\[\alpha = \min\left\{1, \min_{t, B, B' : \hspace{.05in}v_t(B) - v_t(B')
  \neq 0}|v_t(B) - v_t(B')|\right\}\]
be the smallest gap between any valuation's distinct values for two
bundles on the sample, or $1$ if the gap is above $1$.\footnote{The
  gap might be above one if, for example, the agent's valuation for
  each bundle is a distinct even number.} Then, let 
$y'_g = \frac{y_g \cdot \alpha}{4\ngood \cdot \max_{g'} y_{g'} }$.
Then, we define
\[
 \hat{w}^\p_g =
  \begin{cases}
      \hfill -\p_g + y'_g    \hfill & \text{ if $g\in [1,\ngood]$} \\
      \hfill 1 & \text{ if $g = \ngood+1$} \\
  \end{cases}
\]
and
\[
 \hat{\Psi}(v_t, B_t)_g =
  \begin{cases}
      \hfill \I\lbrack[g\in B_t]    \hfill & \text{if $g \in [1,\ngood]$} \\
      \hfill  v_t(B_t) \hfill & \text{if $g = \ngood+1$ and $B_t\notin \cL_{v_t}$} \\
      \hfill  -\maxval \hfill & \text{if $g = m+1$ and $B_t\in \cL_{v_t}$} \\
  \end{cases}.
\]
These definitions allow us to prove the following lemmas about the
bundles $B'\in M(v, p)$. The first shows that all bundles in $M(v,p)$
are utility-maximizing. The second implies that $|M| =1$ and that
$e(\dem{q}(p)) \in M(v,p)$.

\begin{lemma}\LABEL{lem:linear-util}
  If $B' \in M(v,p)$, then $B'$ is utility-maximal for $v$ at prices
  $\p$ and $B'\notin \cL_{v}$.
\end{lemma}
\begin{proof}
  We first show that $B\notin \cL_{v}$ for every $B\in
  M(v,p)$. Suppose $B\in \cL_{v}$. Then,
\begin{align*}
\hat{\Psi}(v,B) \cdot \hat{w}^p &=  -\maxval -\sum_{g\in B} p_g + \sum_{g\in B} y'_g \\
& \leq - \maxval + \sum_{g = 1}^\ngood y'_g \\
& = - \maxval + \sum_{g = 1}^\ngood \frac{\alpha \cdot y_g}{4\ngood\cdot \max_{g'}y_{g'}} \\
& \leq - \maxval + \sum_{g = 1}^\ngood \frac{\alpha}{4\ngood} \\
& \leq 0
\end{align*}
where the final inequality follows from $\maxval\geq1$ and
$\alpha \leq 1$. Thus, any $B' \in \cL_v$ will have negative
dot-product. Thus, if there is any bundle with positive dot product,
$B'$ will not be chosen. If some utility-maximizing bundle
$B\notin \cL_{v}$, it is not hard to show that
$\hat{\Psi}(v,B) \cdot \hat{w}^p \geq 0$. Since $e$ is encodable, we
know that some utility-maximizing bundle is not in $\cL_v$. Thus, no
losing set will maximize the dot product.

We will now show that $B' \in \argmax_{B} v(B) - p(B)$. Since
$v(B) - p(B) \leq \hat{\Psi}(v, B) \cdot \hat{w}^p \leq v(B) - p(B) +
\frac{\alpha}{4}$
for every $B\notin \cL_v$, any $B$ which is utility-maximizing and not
in $\cL_v$ will have larger dot product than any bundle which is not
utility-maximizing, by the definition of $\alpha$.
\end{proof}

\begin{lemma}\LABEL{lem:ties-correct}
  For any $v, p$, $|M(v,p)| =1$ and $\{e(\dem{q}(p))\} = M(v,p)$.
\end{lemma}
\begin{proof}
  \cref{lem:linear-util} implies that any $B'\in M(v,p)$ is not
  in $\cL_v$ and is utility-maximizing, so we know for all
  $B'\in M(v,p)$ that
  \[v(B') - p(B') \leq \hat{\Psi}(v, B') \cdot \hat{w}^p(B') = v(B') -
  p(B') + \sum_{g \in B'} \frac{\alpha \cdot y_g}{4\ngood\cdot
    \max_{g'}y_{g'}}.\]
  For all such $B'$, the first term is equal (they are all
  utility-maximizing). Then, let us consider $B_1, B_2$ both of which
  are utility-maximizing and not in $\cL_v$. It must be the case that

\[\left(\frac{4\ngood\cdot
    \max_{g'}y_{g'}}{\alpha}\right)\left(\hat{\Psi}(v, B_1) \cdot \hat{w}^p(B_1) - \hat{\Psi}(v, B_2)\right) \cdot \hat{w}^p(B_2) = \sum_{g \in B_1}  y_g - \sum_{g \in B_2}  y_g.\]

  Then, since $e$ is encodable, we know there exists a unique
  utility-maximizing bundle $B_1$ for which this quantity is strictly
  positive with respect to each other $B_2$, and that
  $B_1 = e(\dem{q}(p))$. Thus, it must also be that $|M(v,p)| = 1$.
\end{proof}

We now prove \cref{lem:linsep}.

\begin{proof}[Proof of \cref{lem:linsep}]
  For any sample $S$ of size $\sam$, $\hat{\Psi}$, $\hat{w}^p$ are
  defined as above, and for any $v_t$, $M$ is defined as above.
  \cref{lem:linear-util} implies that any $B \in M(v, p)$ is
  utility-maximizing, and \cref{lem:ties-correct} implies that
  $M(v,p)$ is a singleton and agrees with $e$. Thus, $\hat{\Psi}$ and
  the set of $ \hat{w}^p$ are a linear prediction rule which correctly
  predicts buyers will buy the utility-maximizing bundle according to
  tie-breaking rule $e$.
\end{proof}

\subsection{The existence of an encodable tie-breaking rule for GMBV}\LABEL{sec:encodable-existence}

\begin{lemma}\LABEL{lem:existence-ties}
  There exists an encodable tie-breaking rule $e = (\cL_v, y)$ which
\begin{itemize}
\item Always selects type-minimal bundles,
\item and always breaks ties in favor of larger type-minimal bundles.
\end{itemize}
\end{lemma}

\begin{proof}
  We need to define $\cL_v$ and $y$ and prove they indeed always
  select a unique type-minimal bundle of largest size, and also that
  $\cL_v$ \emph{never} contains all utility-maximizing bundles.  Let
  $y_g = 1 + \frac{4^g}{4\ngood \cdot 4^\ngood}$,\footnote{This corresponds to
    breaking ties lexicographically over bundles.} and $\cL_v$ contain
  the set of non-type-minimal bundles for buyers with valuation $v$
  (that is, if $v(B) = v(B\cup \{g\})$ for some $g\notin B$, then
  $B\cup \{g\}\in \cL_v$).  By definition of $\cL_v$, $e$ will never
  select a bundle which is not type-minimal.
  
  We now show that $\cL_v$ never contains all utility-maximizing
  bundles. If there is some bundle $B \in \cL_v$, this implies there
  is some $B'\subsetneq B$ for which $v(B') = v(B)$, and such that
  $B'$ is type-minimal (thus, $B'\notin \cL_v$). Since
  $p(B') \leq p(B)$, $B'$ is also a utility-maximizing; thus, $\cL_v$
  does not contain all utility-maximizing bundles for any pricing $\p$.

  Finally, we show that, for any $\X$, there is always a unique
  $B\in \X$ which maximizes $Y = \sum_{g\in B}y_g$, and that this
  unique maximizer is amongst the bundles of the largest size in
  $\X\setminus \cL$. Notice that, for any $B$,
  $\sum_{g\in B}y_g = |B| + \sum_{g\in B}\frac{4^g}{4\ngood \cdot
    4^\ngood}$;
  this quantity is strictly above $|B| - 1$ and strictly below
  $|B| + 1$, so the candidate maximizers of $Y$ are all of maximum
  size. Then, each of those has a unique value of
  $\sum_{g\in B}\frac{4^g}{4\ngood \cdot 4^\ngood}$, so there is a
  unique maximizer of $Y$.
\end{proof}

\subsection{Shattering arguments}\LABEL{sec:omit-shattering}
\begin{proof}[Proof of \cref{thm:pseudo-welfare}]
  Consider some sample $S = (v_1, \ldots, v_\sam)$ of arbitrary
  valuations which can be shattered with targets
  $(r_1, \ldots, r_\sam)$.  By \cref{lem:labelings}, $\H$ induces at
  most ${\sam \choose m+1} \cdot 2^{m(m+1)}$ bundle labelings on a
  sample of size $\sam$. Fixing the bundle label of a given $v_q$
  fixes the welfare for $v_q$ (if $\hat{B}$ is purchased by $v_q$, her
  welfare is $v_q(\hat{B})$). So, fixing the bundle labeling of all of
  $S$ fixes the welfare for all of $S$. Thus, there are at most
  ${\sam \choose \ngood+1} \cdot 2^{\ngood(\ngood+1)}$ distinct
  welfare labelings of all of $S$; thus, there must be at most that
  many binary labelings of $S$ according to the welfare targets
  $(r_1, \ldots, r_\sam)$. Since $S$ is shatterable,
\[2^\sam \leq {\sam \choose \ngood+1} \cdot 2^{\ngood(\ngood+1)}\]
so $\sam = O(\ngood^2 \log \ngood)$ as claimed.

The proof is identical for unit demand valuations, replacing the
$2^\ngood$ possible labels per valuation with the upper bound of
$\ngood+1$ (since, when buyers are unit demand, they will buy one item
or no items).
\end{proof}

\subsection{Omitted concentration proofs}\LABEL{sec:omit-concentration}

For this section, we will be referring to loss functions on samples,
$\ell(f,x)$, as well as loss of those functions on entire samples
$\ell_N(f)$ and distributions $\ell_\D(f)$. In the latter case, let
$\ell_N(f) = \frac{1}{|N|}\sum_{x\in N} \ell(f,x)$ denote the
\emph{empirical} loss of the function $f$ on a sample $N$; in the
latter case, let $\ell_\D(f) = \ex{x\sim\D}{\ell(f,x)}$ denote the
\emph{true expected} loss of $f$ with respect to $\D$.

\begin{theorem}[McDiarmid's Inequality]
  \LABEL{thm:mcdiarmid}
  Suppose $X^1, \ldots, X^k$ are independent, and assume $f$ has the
  property that, for any $i, x^1, \ldots, x^k, \hat{x}^i$,
\[|f(x^1, \ldots, x^k) - f(x^1, \ldots, \hat{x}^i, \ldots, x^k)| \leq c_i.\]
Then,
\[\Pr\lbrack|\E\lbrack f(X^1, \ldots, X^k)\rbrack - f(X^1, \ldots, X^k)
| \geq \epsilon\rbrack\leq 2 e^{-\frac{2\epsilon^2}{\sum_{i=1}^n c^2_i}}.\]
\end{theorem}

\begin{proof}[Proof of \cref{lem:welfare-concentration}]
  This follows from a simple application of McDiarmid's inequality,
  noting that welfare is a smooth function in each argument (i.e., the
  welfare cannot change by more than $\maxval$ by changing any one of the
  sampled valuations).
\end{proof}

\begin{theorem}\LABEL{thm:binary-bernstein}
  Fix a hypothesis $f\in \F$. Let $\ell$ be some function defined on
  samples such that $\ell(f, x) \in \{0,1\}$,
  $\ell_N(f) = \frac{1}{|N|}\sum_{q\in N}\ell(f, q)$, and
  $\ell_\D(f) = \ex{q\sim\D}{\ell(f, q)}$. Then, with probability
  $1-\delta$, if $|N| = \sam$, for a sample $N\sim\D$,
\[\ell_\D(f) - \ell_N(f) \leq \sqrt{\frac{2\ell_{N}(f)\log\frac{1}{\delta}}{\sam}} + \frac{4\log\frac{1}{\delta}}{\sam}.\]
\end{theorem}

\begin{proof}[Proof of \cref{thm:binary-bernstein}]
  This follows from Lemma B.10 (pg. 427) by \citet{shalev2014book}.
\end{proof}
We now prove \cref{thm:vc-bernstein}, which is a direct
corollary of \cref{thm:binary-bernstein}.
\begin{proof}[Proof of \cref{thm:vc-bernstein}]
  Applying \cref{thm:binary-bernstein} and taking a union bound
  over all possible loss values on the sample (that is, using Sauer's
  Lemma in conjunction with \cref{thm:binary-bernstein} rather
  than the standard Hoeffding bound).
\end{proof}

\begin{proof}[Proof of \cref{thm:demand-bernstein}]
  This follows from \cref{thm:vc-bernstein} along with the
  bound of $s_g + 1$ on $\ell_{S}(h_{g,p})$ and the bound on
  $\VC_{\H_g}$ given in \cref{thm:tighter}. Thus, we have
  proved the first statement of the theorem. The second follows from
  the first, with basic algebra.\jmcomment{check whether a better
    bound is attainable being more careful with calculations.}
\end{proof}

\begin{theorem}\LABEL{thm:real-bernstein}
  Fix a hypothesis $f\in \F$. Let $\ell$ be some loss function defined
  on samples $N$ drawn from $\D$, such that
  $\ell(f, x) \in [0,\maxval]$. Then, with probability $1-\delta$, if
  $|N| = \sam$,
\[\ell_N(f) - \ell_\D(f) \leq
\left(\frac{\maxval^{3/2}\log\frac{1}{\delta}}{\sam}\right)\left(\frac{1}{3} + \sqrt{19\sam\ell_N(f)}\right)\]
\end{theorem}
\begin{proof}
  Let $\alpha_i =\frac{\ell_\D(f) - \ell(f, x_i)}{\sqrt{\maxval}}$ be the
  (normalized) difference between the true error of $f$ and the error
  of $f$ on sample $x_i$. Then, applying Bernstein's inequality to the
  negative of the $\alpha_i$s, we have that for any $t<0$,
\begin{align*}\prob{\sum_{i=1}^m \alpha_i < t}
& \leq \exp\left(-\frac{t^2/2}{\sum_{g=1}^\sam \ex{\alpha^2_g}+ \maxval \cdot t/3}\right)\\
& \leq \exp\left(-\frac{t^2/2}{\frac{\sam}{\sqrt{\maxval}} \cdot \ell_N(f)+ \maxval \cdot t/3}\right)\\
& = \exp\left(-\frac{\sqrt{\maxval}t^2/2}{\sam \cdot \ell_N(f)+ \maxval^{3/2} \cdot t/3}\right).\\
\end{align*}
Setting the right-hand side to equal $\delta$, and solving for $t$, we
have
\begin{align*}
\log\frac{1}{\delta} &= \frac{\sqrt{\maxval}t^2/2}{\sam \ell_N(f) + \frac{\maxval^{3/2}\cdot t}{3}}\\
 \Rightarrow t &\leq \frac{\maxval}{3}\log\frac{1}{\delta} + \sqrt{\maxval^2\log^2\frac{1}{\delta} + \frac{18 \ell_N(f) \sam \log\frac{1}{\delta}}{\sqrt{\maxval}}}\\
& \leq \frac{\maxval}{3}\log\frac{1}{\delta} + \sqrt{19\cdot \maxval^2\log^2\frac{1}{\delta} \ell_N(f)\sam}\\
& \leq \frac{\maxval}{3}\log\frac{1}{\delta} + \maxval \log\frac{1}{\delta}\sqrt{19 \ell_N(f)\sam}\\
& \leq \frac{\maxval}{3}\log\frac{1}{\delta} + \maxval \log\frac{1}{\delta}\sqrt{19 \ell_N(f)\sam}.\\
\end{align*}
As
$\frac{\sqrt{\maxval}}{\sam}\sum_{i=1}^\sam\alpha_i = \ell_\D(f) -
\ell_N(f)$, this implies that with probability $1-\delta$,
\[\ell_N(f) - \ell_\D(f) \leq
 \left(\frac{\maxval^{3/2}\log\frac{1}{\delta}}{\sam}\right)\left(\frac{1}{3} + \sqrt{19\ell_N(f)\sam}\right)\]
which proves the claim.
\end{proof}

\begin{theorem}\LABEL{thm:pseudo-bernstein}
  Consider a hypothesis class $\F$, and $\ell(f,x)\in [0,\maxval]$.
  Then, with probability $1-\delta$, if $|N| = \sam$, for all
  $f\in \F$,
\[\ell_N(f) - \ell_\D(f) \leq
 \left(\frac{\maxval^{3/2}\pd_\F\log\frac{1}{\delta}}{\sam}\right)
 \left(\frac{1}{3} + \sqrt{19\ell_N(f)\sam}\right) .\]
\end{theorem}
\begin{proof}[Proof of \cref{thm:pseudo-bernstein}]
  Applying \cref{thm:real-bernstein} and taking a union bound
  (that is, using Sauer's Lemma in conjunction with
  \cref{thm:real-bernstein} rather than the standard Hoeffding
  bound).\jmcomment{ flesh out in the full version}
\end{proof}